\documentclass[11pt,letterpaper]{article}
\usepackage{floatrow}
\usepackage{typearea}
\typearea{14}

\usepackage{comment,multicol}

\usepackage[ruled,linesnumbered,vlined]{algorithm2e}

\usepackage{multirow}
\usepackage{wrapfig}
\usepackage{color,colortbl}
\usepackage{float}
\usepackage{amsthm}
\usepackage{amsmath}
\usepackage{amssymb}
\usepackage{graphicx}
\usepackage{xspace}
\usepackage{nicefrac}
\usepackage[colorinlistoftodos,prependcaption,textsize=tiny]{todonotes}

\usepackage{dsfont}
\usepackage{thmtools}
\declaretheorem[name=Lemma]{lemma}

\definecolor{Darkblue}{rgb}{0,0,0.4}
\definecolor{Brown}{cmyk}{0,0.61,1.,0.60}
\definecolor{Purple}{cmyk}{0.45,0.86,0,0}

\usepackage[colorlinks,linkcolor=Darkblue,filecolor=blue,citecolor=blue,urlcolor=Darkblue,pagebackref]{hyperref}
\usepackage[nameinlink]{cleveref}

\usepackage[toc]{multitoc}

\newcommand{\initOneLiners}{%
	\setlength{\itemsep}{0pt}
	\setlength{\parsep }{0pt}
	\setlength{\topsep }{0pt}
}
\newenvironment{OneLiners}[1][\ensuremath{\bullet}]
{\begin{list}
		{#1}
		{\initOneLiners}}
	{\end{list}}

\newcommand{\initTwoLiners}{%
	\setlength{\itemsep}{1pt}
	\setlength{\parsep }{0pt}
	\setlength{\topsep }{0pt}
}
\newenvironment{TwoLiners}[1][\ensuremath{\bullet}]
{\begin{list}
		{#1}
		{\initTwoLiners}}
	{\end{list}}

\newtheorem{theorem}{Theorem}
\newtheorem{corollary}{Corollary}
\newtheorem{remark}{Remark}
\newtheorem{claim}{Claim}
\newtheorem{definition}{Definition}
\newtheorem{observation}[lemma]{Observation}
\newtheorem{conjecture}{Conjecture}

\numberwithin{equation}{section}

\newcommand{\namedref}[2]{\hyperref[#2]{#1~\ref*{#2}}}

\newcommand{\E}{{\mathbb{E}}}
\newcommand{\N}{\mathbb{N}}
\newcommand{\R}{\mathbb{R}}
\newcommand{\Z}{\mathbb{Z}}
\newcommand{\cP}{\mathcal{P}}
\newcommand{\cR}{\mathcal{R}}
\newcommand{\cC}{\mathcal{C}}

\newcommand{\poly}{{\rm poly}}
\newcommand{\polylog}{{\rm polylog}}

\newcommand{\Exp}{\mathsf{Exp}}

\newcommand{\ctop}{c_\top}

\newcommand{\etal}{{et al. \xspace}}
\newcommand{\vol}{\mbox{\rm Vol}}

\newcommand{\rt}{\mbox{\rm rt}}

\newcommand{\ddim}{\mbox{\rm ddim}}
\newcommand{\sddim}{\mbox{\small\rm ddim}}
\newcommand{\tddim}{\mbox{\tiny\rm ddim}}
\newcommand{\eps}{\epsilon}

\newcommand{\SPD}{\textsf{SPD}\xspace}
\newcommand{\SPDs}{{\SPD}s\xspace}
\newcommand{\SPDdepth}{\textsf{SPDdepth}\xspace}

\newcommand{\SPR}{\textsf{SPR}\xspace}
\newcommand{\UTSP}{\textsf{UTSP}\xspace}
\newcommand{\TSP}{\textsf{TSP}\xspace}
\newcommand{\UST}{\textsf{UST}\xspace}

\newcommand{\cint}{c_{\text{\tiny int}}}

\newcommand{\QED}{\\{\color{white}.}\hfill\qedsymbol}

\usepackage{pdfsync}

\title{Scattering and Sparse Partitions, and their Applications\thanks{This research was supported by the Israel Science Foundation (grant No. 1042/22).}}
\author{Arnold Filtser\\Bar Ilan University\\
	Email: \texttt{arnold.filtser@biu.ac.il} }
\date{\today}

\begin{document}
	\maketitle
	\thispagestyle{empty}
	\nonumber
\begin{abstract}
A partition $\mathcal{P}$ of a weighted graph $G$ is $(\sigma,\tau,\Delta)$-sparse if every cluster has diameter at most $\Delta$, and every ball of radius $\Delta/\sigma$ intersects at most $\tau$ clusters.
Similarly, $\mathcal{P}$ is $(\sigma,\tau,\Delta)$-scattering if instead for balls we require that every shortest path of length  at most $\Delta/\sigma$ intersects at most $\tau$ clusters.
Given a graph $G$ that admits a $(\sigma,\tau,\Delta)$-sparse partition for all $\Delta>0$,  
Jia et al. [STOC05] constructed a solution for the Universal Steiner Tree problem (and also Universal TSP) with stretch $O(\tau\sigma^2\log_\tau n)$.
Given a graph $G$ that admits a $(\sigma,\tau,\Delta)$-scattering partition for all $\Delta>0$,  
we construct a solution for the Steiner Point Removal problem with stretch $O(\tau^3\sigma^3)$.
We then construct sparse and scattering partitions for various different graph families, receiving many new results for the Universal Steiner Tree and Steiner Point Removal problems.
\end{abstract}
\vfill
{\small \setcounter{tocdepth}{1} \tableofcontents}
\newpage
\pagenumbering{arabic}

\section{Introduction}
Graph and metric clustering are widely used for various algorithmic applications (e.g., divide and conquer). Such partitions come in a  variety of forms, satisfying different requirements.
This paper is dedicated to the study of bounded diameter partitions, where small neighborhoods are guaranteed to intersect only a bounded number of clusters.

The first problem we study is the \emph{Steiner Point Removal} (\SPR) problem.
Here we are given an undirected weighted graph $G=(V,E,w)$ and a subset of terminals $K\subseteq V$ of size $k$ (the non-terminal vertices are called Steiner vertices).
The goal is to construct a new weighted graph $M=(K,E',w')$, with the terminals as its vertex set, such that: (1) $M$ is a graph minor of $G$, and (2) the distance between every pair of terminals $t,t'$ in $M$ is distorted by at most a multiplicative factor of $\alpha$, formally
$$\forall t,t'\in K,~~d_G(t,t')\le d_{M}(t,t')\le \alpha \cdot d_G(t,t')~.$$
Property (1) expresses preservation of the topological structure of the original graph. For example if $G$ was planar, so will $M$ be. Whereas property (2) expresses preservation of the geometric structure of the original graph, that is, distances between terminals.
The question is thus: given a graph family $\mathcal{F}$, what is the minimal $\alpha$ such that every graph in $\mathcal{F}$ with a terminal set of size $k$ will admit a solution to the SPR problem with distortion $\alpha$.

Consider a weighted graph $G=(V,E,w)$ with a shortest path metric $d_G$.
The \emph{weak} diameter of a cluster $C\subseteq V$ is the maximal distance between a pair of vertices in the cluster w.r.t. $d_G$ (i.e.,  $\max_{u,v\in C}d_G(u,v)$). The \emph{strong} diameter is the maximal distance  w.r.t. the shortest path metric in the induced graph $G[C]$ (i.e., $\max_{u,v\in C}d_{G[C]}(u,v)$).
A partition $\mathcal{P}$ of $G$  has weak (resp. strong) diameter $\Delta$ if every cluster $C\in \mathcal{P}$ has weak (resp. strong) diameter at most $\Delta$. Partition $\mathcal{P}$ is \emph{connected}, if the graph induced by every cluster $C\in\mathcal{P}$ is connected.
Given a shortest path $\mathcal{I}=\{v_0,v_1,\dots,v_s\}$, denote by $Z_{\mathcal{I}}(\mathcal{P})=\sum_{C\in \mathcal{P}}\mathds{1}_{C\cap\mathcal{I}\ne\emptyset}$ the number of clusters in $\mathcal{P}$ intersecting $\mathcal{I}$. If $Z_{\mathcal{I}}(\mathcal{P})\le\tau$, we say that $\mathcal{I}$ is $\tau$-\emph{scattered} by $\mathcal{P}$. 
\begin{definition}[Scattering Partition]
	Given a weighted graph $G=(V,E,w)$, we say that a partition $\mathcal{P}$ is $(\sigma,\tau,\Delta)$-scattering if the following conditions hold:
	\begin{OneLiners}
		\item $\mathcal{P}$ is connected and has weak diameter $\Delta$.
		\item Every shortest path $\mathcal{I}$ of length at most $\Delta/\sigma$ is $\tau$-scattered by $\mathcal{P}$, i.e., $Z_{\mathcal{I}}(\mathcal{P})\le\tau$.
	\end{OneLiners}
	We say that a graph $G$ is $(\sigma,\tau)$-\emph{scatterable} if for every 
	parameter $\Delta$, $G$ admits a $(\sigma,\tau,\Delta)$-scattering partition that can be computed efficiently.
\end{definition}

The main contribution of this paper is the finding that scattering partitions imply solutions for the \SPR problem. 
The proof appears in \Cref{sec:scat_to_spr}.\footnote{In \Cref{obs:beta1} we argue that $(\sigma,\tau,\Delta)$-scattering partition is also $(1,\tau\sigma,\Delta)$-scattering. We study the general case, even though \Cref{thm:Scattering_Implies_SPR} requires only $\sigma=1$. This is as we find the more general case theoretically interesting, as well as potentially applicable.}

\begin{theorem}[Scattering Partitions imply \SPR]\label{thm:Scattering_Implies_SPR}
	Let $G=(V,E,w)$ be a weighted graph such that for every subset $A\subseteq V$, $G[A]$ is  $(1,\tau)$-scatterable.
	Let $K\subseteq V$ be some subset of terminals. Then there is a solution to the \SPR problem with distortion $O(\tau^3)$ that can be computed efficiently.
\end{theorem}

Jia, Lin, Noubir, Rajaraman, and Sundaram \cite{JLNRS05} \footnote{Awerbuch and Peleg \cite{AP90} were the first to study sparse covers (see \Cref{def:SparseCover}). Their notion of sparse partition is somewhat different from the one used here (introduced by \cite{JLNRS05}).} defined the notion of \emph{sparse partitions}, which is closely related to scattering partitions. Consider a partition $\mathcal{P}$. Given a ball $B=B_G(x,r)$, denote by $Z_{B}(\mathcal{P})=\sum_{C\in \mathcal{P}}\mathds{1}_{C\cap B\ne\emptyset}$ the number of clusters in $\mathcal{P}$ intersecting $B$.
\begin{definition}[Strong/Weak Sparse Partition]\label{def:sparsePartition}
	Given a weighted graph $G=(V,E,w)$, we say that a partition $\mathcal{P}$ is $(\sigma,\tau,\Delta)$-weak (resp. strong) sparse partition if the following conditions hold:
	\begin{OneLiners}
		\item $\mathcal{P}$ has weak (resp. strong) diameter $\Delta$.
		\item Every ball $B=B_G(v,r)$ of radius $r\le\Delta/\sigma$ intersects at most $\tau$ clusters, i.e., $Z_{B}(\mathcal{P})\le\tau$.
	\end{OneLiners}
	We say that a graph $G$ admits a $(\sigma,\tau)$-\emph{weak} (resp. strong) sparse partition scheme if for every parameter $\Delta$, $G$ admits an efficiently computable $(\sigma,\tau,\Delta)$-weak (resp. strong) sparse partition.
\end{definition}

\begin{wrapfigure}{r}{0.2\textwidth}
	\begin{center}
		\vspace{-20pt}
		\includegraphics[width=0.9\textwidth]{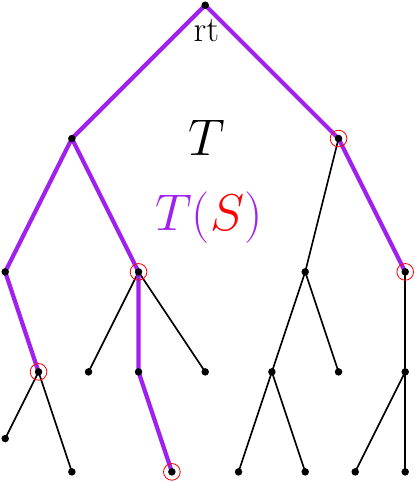}
		\vspace{-5pt}
	\end{center}
	\vspace{-10pt}
\end{wrapfigure}
Jia \etal \cite{JLNRS05} found a connection between sparse partitions to the \emph{Universal Steiner Tree Problem} (\UST).\footnote{A closely related problem is the \emph{Universal Traveling Salesman Problem} (\UTSP), see \Cref{subsec:related}.}
Consider a complete weighted graph $G=(V,E,w)$ (or a metric space $(X,d)$) where there is a special server vertex $\rt\in V$, which is frequently required to multicast messages to different subsets of clients $S\subseteq V$. The cost of a multicast is the total weight of all edges used for the communication.
Given a subset $S$, the optimal solution is to use the minimal Steiner tree spanning $S\cup\{\rt\}$. 
In order to implement an infrastructure for multicasting, or in order to make routing decisions much faster (and not compute it from scratch once $S$ is given), a better solution will be to compute a \emph{Universal Steiner Tree} (\UST). 
A \UST is a tree $T$ over $V$, such that for every subset $S$, the message will be sent using the sub-tree $T(S)$ spanning $S\cup\{\rt\}$. See illustration on the right, where the set $S$ is surrounded by red circles and $T(S)$ is purple.
The stretch of $T$ is the maximum ratio among all subsets $S\subseteq X$ between the weight of $T(S)$ and the weight of the minimal Steiner tree spanning $S\cup\{\rt\}$, $\max_{S\subseteq X}\frac{w(T(S))}{\mbox{Opt}(S\cup\{\rt\})}$.

Jia \etal \cite{JLNRS05} proved that given a sparse partition scheme, one can efficiently construct a \UST with low stretch (the same statement holds w.r.t. $\UTSP$ as well).
\begin{theorem}[Sparse Partitions imply \UST, \cite{JLNRS05}]\label{thm:JLNRS05}
	Suppose that an $n$-vertex graph $G$ admits a $(\sigma,\tau)$-weak sparse partition scheme, then there is a polynomial time algorithm that given a root $\rt\in V$ computes a  \UST with stretch $O(\tau\sigma^2\log_\tau n)$.
\end{theorem}

Jia \etal \cite{JLNRS05} constructed $(O(\log n),O(\log n))$-weak sparse partition scheme for general graphs, receiving a solution with stretch $\polylog(n)$ for the \UST problem.
In some applications the communication is allowed to flow only in certain routes. It is therefore natural to consider the case where $G=(V,E,w)$ is not a complete graph, and the \UST is required to be a subgraph of $G$, we refer to this as the Subgraph \UST problem.
Busch, Jaikumar, Radhakrishnan, Rajaraman, and Srivathsan \cite{BDRRS12} proved a theorem in the spirit of \Cref{thm:JLNRS05}, stating that given a $(\sigma,\tau,\gamma)$-\emph{hierarchical strong sparse partition},
one can efficiently construct a subgraph \UST with stretch $O(\sigma^2\tau^2\gamma\log n)$.
A $(\sigma,\tau,\gamma)$-hierarchical strong sparse partition is a laminar collection of partitions $\{\mathcal{P}_i\}_{i\ge 0}$ such that $\mathcal{P}_i$ is $(\sigma,\tau,\gamma^i)$-strong sparse partition which is a refinement of $\mathcal{P}_{i+1}$.\footnote{We assume here w.l.o.g. that the minimal distance in $G$ is $1$.}
Busch \etal constructed a $\big(2^{O(\sqrt{\log n})},2^{O(\sqrt{\log n})},2^{O(\sqrt{\log n})}\big)$-hierarchical strong sparse partition, obtaining a $2^{O(\sqrt{\log n})}$ stretch algorithm for the subgraph \UST problem. We tend to believe that  poly-logarithmic stretch should be possible. It is therefore interesting to construct strong sparse partitions, as it eventually may lead to hierarchical ones.

A notion which is closely related to sparse partitions is \emph{sparse covers}.
\begin{definition}[Strong/Weak Sparse cover]\label{def:SparseCover}
	Given a weighted graph $G=(V,E,w)$, a $(\sigma,\tau,\Delta)$-weak (resp. strong) sparse cover is a set of clusters $\mathcal{C}\subset 2^V$, where all the clusters have weak (resp. strong) diameter at most $\Delta$, and the following conditions hold:
	\begin{OneLiners}
		\item Cover: $\forall u\in V$, there exists $C\in\mathcal{C}$ such that $B_G(u,\frac\Delta\sigma)\subseteq C$.
		\item Sparsity: every vertex $u\in V$ belongs to at most  $\left|\{C\in\mathcal{C}\mid u\in C\}\right|\le \tau$ clusters.
	\end{OneLiners}
	We say that a graph $G$ admits a $(\sigma,\tau)$-\emph{weak} (resp. strong) sparse cover scheme if for every parameter $\Delta$, $G$ admits a $(\sigma,\tau,\Delta)$-weak (resp. strong) sparse cover that can be computed efficiently.
\end{definition}
It was (implicitly) proven in \cite{JLNRS05} that given $(\sigma,\tau,\Delta)$-weak sparse cover $\mathcal{C}$, one can construct a $(\sigma,\tau,\Delta)$-weak sparse partition. 
In fact, most previous constructions of weak sparse partitions were based on sparse covers.

\subsection{Previous results}
\paragraph{\SPR} 
Given an $n$-point tree, Gupta \cite{G01} provided an upper bound of $8$ for the \SPR problem (on trees). This result was recently reproved by the author, Krauthgamer, and Trabelsi \cite{FKT19} using the \texttt{Relaxed-Voronoi} framework.
Chan, Xia, Konjevod, and Richa \cite{CXKR06} provided a lower bound of $8$ for trees.  This is the best known lower bound for the general SPR problem. 
Basu and Gupta \cite{BG08} provided an $O(1)$ upper bound for the family of outerplanar graphs.\footnote{Actually the manuscript \cite{BG08} was never published, and thus did not go through a peer review process.\label{foot:BG08}}
For general $n$-vertex graphs with $k$ terminals the author \cite{Fil18,Fil19sicomp} recently proved an $O(\log k)$ upper bound for the \SPR problem using the  \texttt{Relaxed-Voronoi} framework, improving upon previous works by  Kamma, Krauthgamer, and Nguyen \cite{KKN15} ($O(\log^5 k)$), and Cheung \cite{Che18} ($O(\log^2 k)$) (which were based on the \texttt{Ball-Growing} algorithm). 
Interestingly, there are no results on any other restricted graph family, although several attempts have been made (see \cite{EGKRTT14,KNZ14,CGH16}).

\paragraph{\UST} 
Given an $n$-point metric space and root $\rt$, Gupta, Hajiaghayi and R{\"{a}}cke \cite{GHR06} constructed a \UST with stretch $O(\log^2 n)$, improving upon a previous $O(\log^4 n/\log\log n)$ result by \cite{JLNRS05}. \cite{JLNRS05} is based on sparse partitions, while \cite{GHR06} is based on \emph{tree covers}.
Jia \etal \cite{JLNRS05} proved a lower bound of $\Omega(\log n)$ to the \UST problem, based on a lower bound to the online Steiner tree problem by Alon and Azar \cite{AA92}.
Using the same argument, they \cite{JLNRS05} proved an $\Omega(\frac{\log n}{\log\log n})$ lower bound for the case where the space is the $n\times n$ grid (using \cite{IW91}).
Given a space with doubling dimension $\ddim$, \footnote{A metric space $(X, d)$ has doubling dimension $\ddim$ if every ball of radius $2r$ can be 	covered by $2^{\tddim}$ balls of radius $r$. The doubling dimension of a graph is the doubling dimension of its induced shortest path metric.}
Jia \etal \cite{JLNRS05} provided a solution with stretch $2^{O(\sddim)}\cdot \log n$, using sparse partitions.
Given an $n$ vertex planar graph, Busch, LaFortune, and Tirthapura \cite{BLT14} proved an $O(\log n)$ upper bound (improving over Hajiaghayi, Kleinberg, and Leighton \cite{HKL06}). 
More generally, for graphs $G$ excluding a fixed minor, both Hajiaghayi \etal \cite{HKL06} (implicitly) and  Busch \etal \cite{BLT14} (explicitly)  provided a solution with stretch $O(\log^2 n)$. Both constructions used sparse covers.
Finally, Busch \etal \cite{BDRRS12} constructed a subgraph \UST with stretch $\polylog (n)$ for graphs excluding a fixed minor (using hierarchical strong sparse partitions).

\paragraph{Scattering Partitions}
As we are the first to define scattering partitions there is not much previous work.
Nonetheless, Kamma \etal \cite{KKN15} implicitly proved that general $n$-vertex graphs are $(O(\log n),O(\log n))$-scatterable.\footnote{This follows from Theorem 1.6 in \cite{KKN15} by choosing parameters $t=\beta=O(\log n)$ and using union bounds over all $n^2$ shortest paths. Note that they assume that for every pair of vertices there is a unique shortest path.}

\paragraph{Sparse Covers and Partitions}
Awerbuch and Peleg \cite{AP90} introduced the notion of sparse covers and constructed $(O(\log n),O(\log n))$-strong sparse cover scheme for $n$-vertex weighted graphs.\footnote{More generally, for $k\in\N$, \cite{AP90} constructed a $(4k-2,2k\cdot n^{\frac1k})$-strong sparse cover scheme.} 
Jia \etal \cite{JLNRS05} induced an $\left(O(\log n),O(\log n)\right)$-weak sparse partition scheme .
Hajiaghayi \etal \cite{HKL06} constructed an $\left(O(1),O(\log n)\right)$-weak sparse cover scheme for $n$-vertex planar graph, concluding an $\left(O(1),O(\log n)\right)$-weak sparse partition scheme. Their construction is based on the \cite{KPR93} clustering algorithm.
Abraham, Gavoille, Malkhi, and Wieder \cite{AGMW10} constructed $(O(r^2),2^{O(r)}\cdot r!)$-strong sparse cover scheme for $K_r$-free graphs.
Busch \etal \cite{BLT14} constructed a $(48,18)$-strong sparse cover scheme for planar graphs \footnote{Busch \etal argued that they constructed  $(24,18)$-strong sparse covering scheme. However they measured radius rather than diameter.} and $(8,O(\log n))$-strong sparse cover scheme for graphs excluding a fixed minor, concluding a $(48,18)$  and $(8,O(\log n))$-weak sparse partition schemes for these families (respectively). 
For graphs with doubling dimension $\ddim$, Jia \etal \cite{JLNRS05} constructed an $(1,8^{\sddim})$-weak sparse scheme.
Abraham \etal \cite{AGGM06} constructed a $(2,4^{\sddim})$-strong sparse cover scheme.
In a companion paper, the author \cite{Fil19Padded} constructed an $\left(O(\ddim),O(\ddim\cdot\log \ddim)\right)$-strong sparse cover scheme.\footnote{More generally, for a parameter  $t=\Omega(1)$, \cite{Fil19Padded} constructed  $\left(O(t),O(2^{\nicefrac{\tddim}{t}}\cdot\ddim\cdot\log t)\right)$-sparse cover scheme.}
Busch \etal \cite{BDRRS12} constructed $\left(O(\log^4 n),O(\log^3 n),O(\log^4 n)\right)$-hierarchical strong sparse partition for graphs excluding a fixed minor.

\subsection{Our Contribution}\label{subsec:contribution}

\begin{figure}[p]
	\floatbox[{\capbeside\thisfloatsetup{capbesideposition={left,top},capbesidewidth=7cm}}]{figure}[\FBwidth]
	{\caption{\small \it Classification of various graph families according to the possibility of construction different partitions. 
			Graphs with bounded doubling dimension or \SPD$^{\ref{foot:SPD}}$ (pathwidth) admit strong sparse partitions with parameters depending only on the dimension/\SPDdepth. Trees, Chordal and Cactus graphs admit both $(O(1),O(1))$-weak sparse and scattering partitions, while similar strong partitions are impossible. $\mathbb{R}^d$ with norm $2$ admit $(1,2d)$ scattering partition while weak sparse partition with constant padding will have an exponential number of intersections. Planar graphs admit $(O(1),O(1))$-weak sparse partitions, while it is an open question whether similar scattering partitions exist. Finally, while sparse partitions for general graphs are well understood, we lack a lower bound for scattering partitions.}\label{fig:venn}}
	{\includegraphics[width=9cm]{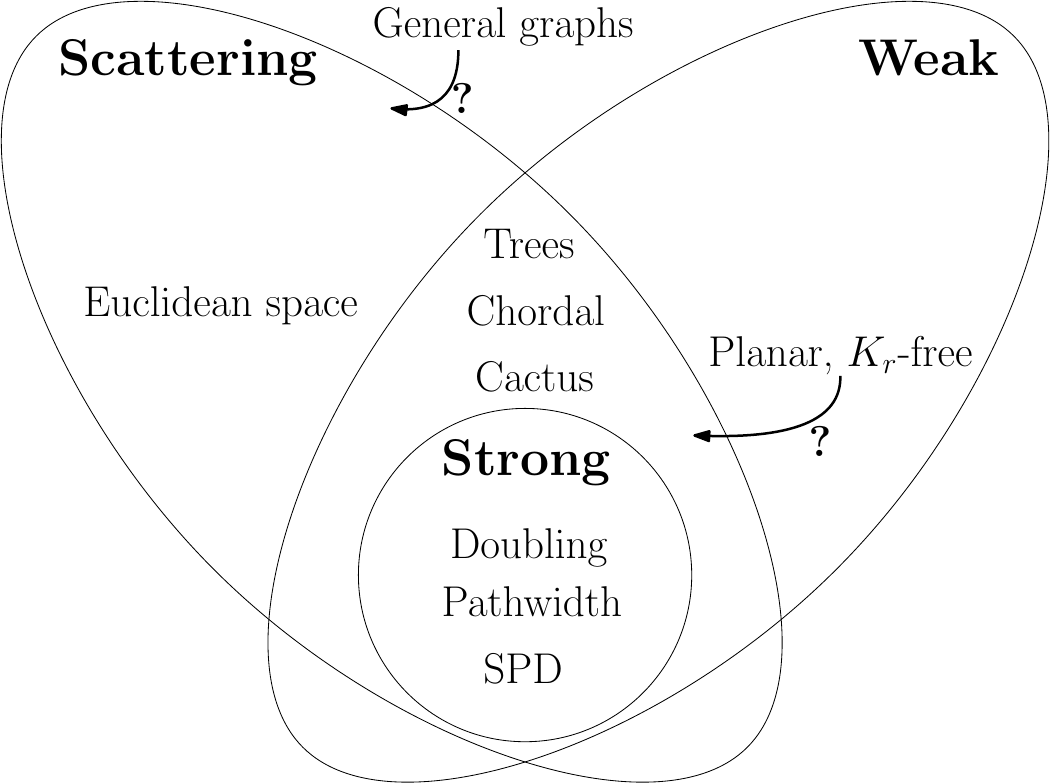}}
\end{figure}
\renewcommand{\arraystretch}{1.18}
\begin{table}[p]
	\begin{tabular}{|l|l|l|l|ll|}
		\hline
		\textbf{Family}&\textbf{Partition type} & \textbf{Padding ($\sigma$)}  & \textbf{\#intersections ($\tau$)}       & \textbf{Ref/Notes}                      &\\ \hline\hline
		\multirow{4}{*}{\textbf{\begin{tabular}[c]{@{}l@{}}General \\ $n$-vertex \\ Graphs\end{tabular}}}&Weak                    & $O(\log n)$              & $O(\log n)$            &  \cite{JLNRS05}                                 &\\  \cline{2-6} 
		&Scattering                  & $O(\log n)$              & $O(\log n)$            & \cite{KKN15}                                &\\  \cline{2-6} 
		&Strong                  & $O(\log n)$              & $O(\log n)$            & \Cref{thm:Generalstrong}   &$^\clubsuit$\\ \cline{2-6} 
		&Weak L.B. & $\Omega(\nicefrac{\log n}{\log\log n})$&$O(\log n)$&\Cref{thm:GeneralWeakLB} &$^\blacksquare$\\ \hline		\hline
		
		\multirow{2}{*}{\textbf{\begin{tabular}[c]{@{}l@{}}\ddim~ doubling   \\ dimension\end{tabular}}}&Weak          & $1$            & $8^{\sddim}$                             &  \cite{JLNRS05}                                      &\\ \cline{2-6}
		&Strong          & $O(\ddim)$           & $\tilde{O}(\ddim)$                 &  \Cref{thm:ddimStrong}  &$^\spadesuit$\\ \hline\hline
		\multirow{3}{*}{\textbf{\begin{tabular}[c]{@{}l@{}}\textbf{Euclidean space} \\ $\boldsymbol{\left(\R^d,\|\cdot\|_2\right)}$ \end{tabular}}}&Scattering           &  $1$          &    $2d$                     &  \Cref{thm:L2scattering}                                    &\\ \cline{2-6}
		&	\multirow{2}{*}{\begin{tabular}[c]{@{}l@{}}Weak L.B.\end{tabular}} & $O(1)$ & $2^{\Omega(d)}$  &\multirow{2}{*}{\Cref{thm:L2LB}} &\multirow{2}{*}{$^\blacklozenge$}\\
		& &$\Omega(\nicefrac{d}{\log d})$ &  $\poly(d)$ &&\\
		\hline\hline
				\multirow{4}{*}{\textbf{Trees}}&Scattering       & $2$         & $3$                                        &   \Cref{thm:tressScat}                                     &\\  \cline{2-6}		
		&Weak             &  $4$         & $3$                                        &     \Cref{thm:treeWeak}                                    &\\  \cline{2-6}		
		&\multirow{2}{*}{\begin{tabular}[c]{@{}l@{}}Strong L.B.\end{tabular}}   & $\nicefrac{\log n}{\log\log n}$& $\log n$ &\multirow{2}{*}{\Cref{thm:treeLBStrong}}&\multirow{2}{*}{$^\bigstar$}\\
		&  & $\sqrt{\log n}$ &$2^{\sqrt{\log n}}$ &  &\\
		\hline\hline
		\multirow{2}{*}{\textbf{\begin{tabular}[c]{@{}l@{}}\textbf{Pathwidth $\rho$}  \\ (\SPDdepth$^{\ref{foot:SPD}}$) \end{tabular}}} &Strong               & $O(\rho)$        & $O(\rho^2)$                             &         \Cref{thm:SPDstrong}                                &\\ \cline{2-6}	
		&Weak               & $8$       & $5\rho$                                       &   \Cref{thm:SPDweak}                                     &\\ \hline\hline
		\multirow{2}{*}{\textbf{Chordal}}&Scattering                               & $2$       & $3$                  &  \Cref{thm:chordalScat}                                      &\\  \cline{2-6}	
		&Weak             	 & $24$   	  & $3$                                     &    \Cref{cor:chordalWeak}                                      &\\ \hline	\hline			
		$\boldsymbol{K_r}$ \textbf{free}&Weak                          & $O(r^2)$         & $2^{r}$       & 
		\Cref{cor:WeakMinor}
		&\\ \hline\hline		
		\textbf{Cactus}&Scattering            & $4$         & $5$                     &\Cref{thm:cactus}&\\ \hline\hline
		\textbf{Series-parallel} & Scattering            & $O(1)$         & $O(1)$                     &\cite{HL22}&\\ \hline
	\end{tabular}
	\caption{\small \it Summary of the various new/old, weak/strong scattering/sparse partitions. 
		\newline Table footnotes:  
		$^\clubsuit$~More generally, there is a partition $\mathcal{P}$ s.t. every ball of radius $\frac{\Delta}{8\alpha}$ intersects at most $\tilde{O}(n^{\nicefrac{1}{\alpha}})$ clusters, for all $\alpha>1$ simultaneously.
		$^\blacksquare$~More generally, it must hold that $\tau\ge n^{\Omega(\nicefrac{1}{\sigma})}$.
		$^\spadesuit$~More generally, there is a partition $\mathcal{P}$ s.t. every ball of radius $\Omega(\frac{\Delta}{\alpha}$) intersects at most $O(2^{\nicefrac{\tddim}{\alpha}})$ clusters, for all $\alpha>1$ simultaneously.
		$^\blacklozenge$~More generally, it must hold that $\tau>(1+\frac{1}{2\sigma})^d$.
		$^\bigstar$~Note that this lower bound holds for chordal/cactus/planar/$K_r$-free graphs. More generally, it must hold that $\tau\ge \Omega(n^{\nicefrac{2}{\sigma+1})})$.
		\label{tab:partitions}}
\end{table}

The main contribution of this paper is the definition of scattering partition and the finding that good scattering partitions imply low distortion solutions for the \SPR problem (\Cref{thm:Scattering_Implies_SPR}).
We construct various scattering and sparse partition schemes for many different graph families, and systematically classify them according to the partition types they admit. In addition, we provide several lower bounds. The specific partitions and lower bounds are described below.  Our findings are summarized in \Cref{tab:partitions}, while the resulting classification is illustrated in \Cref{fig:venn}.

Recall that \cite{JLNRS05} (implicitly) showed that sparse covers imply weak sparse partitions (\Cref{lem:coverToPartition}).
We show that the opposite direction is also true. That is, given a $(\sigma,\tau,\Delta)$-weak sparse partition,  one can construct a $(\sigma+2,\tau,(1+\frac2\sigma)\Delta)$-weak sparse cover (\Cref{lem:partitionToCover}). 
Interestingly, in addition we show that strong sparse partitions imply strong sparse covers, while the opposite is not true. Specifically there are graph families that admit $(O(1),O(1))$-strong sparse cover schemes, while there are no constants $\sigma,\tau$, such that they admit $(\sigma,\tau)$-strong sparse partitions.
All our findings on the connection between sparse partitions and sparse covers, and a classification of various graph families, appears in \Cref{sec:covers} and are summarized in \Cref{fig:Venn_covers}. 

The scattering partitions we construct imply new solutions for the \SPR problem previously unknown. Specifically, for every graph with pathwidth $\rho$ we provide a solution to the \SPR problem with distortion $\poly (\rho)$, independent of the number of terminals (\Cref{cor:pathwidth-SPR}).
After trees \cite{G01} and outerplanar graphs \cite{BG08} $^{\ref{foot:BG08}}$, this is the first graph family to have solution for the \SPR problem independent from the number of terminals.
Furthermore, we obtain solutions with constant distortion for Chordal and Cactus graphs (\Cref{cor:chordalSPR,cor:cactusSPR}).\footnote{Note that the family of cactus graph is contained in the family of outerplanar graph. Basu and Gupta \cite{BG08} solved the \SPR problem directly on outerplanar graphs with constant distortion. However, this manuscript was never published. See also \cref{foot:BG08}.\label{foot:cactusWasKnown}}

The weak sparse partitions we construct imply improved solutions for the \UST (and \UTSP) problem. Specifically, we conclude that for graphs with doubling dimension $\ddim$ a \UST (and \UTSP) with stretch $\poly(\ddim)\cdot \log n$ can be efficiently computed (\Cref{cor:ddimUTSP}), providing an exponential improvement in the dependence on $\ddim$ compared with the previous state of the art \cite{JLNRS05} of $2^{O(\sddim)}\cdot \log n$.
For $K_r$-minor free graphs we conclude that an \UST (or \UTSP) with stretch $2^{O(r)}\cdot \log n$ can be efficiently computed (\Cref{cor:minorUTSP}), providing a quadratic improvement in the dependence on $n$ compared with the previous state of the art \cite{GHR06} of $O(\log^2 n)$. \footnote{This result is a mere corollary obtained by assembling previously existing parts together. Mysteriously, although \UTSP on minor free graphs was studied before \cite{HKL06,BLT14}, this corollary was never drawn, see \Cref{subsec:MinorFreeCovers}.}
Finally, for pathwidth $\rho$ graphs (or more generally, graphs with \SPDdepth $\rho$) we can compute a \UST (or \UTSP) with stretch $O(\rho\cdot \log n)$ (\Cref{cor:SPD-UTSP}), improving over previous solutions that were exponential in $\rho$ (based on the fact that pathwidth $\rho$ graphs are $K_{\rho+2}$-minor free).

Before we proceed to describe our partitions we make two observations.  
\begin{observation}\label{obs:StrongImplyScatt}
	Every $(\sigma,\tau,\Delta)$-strong sparse partition is also a scattering partition and a weak sparse partition with the same parameters.
\end{observation}
\begin{observation}\label{obs:beta1}
	Every $(\sigma,\tau,\Delta)$-scattering partition is also $(1,\sigma\tau,\Delta)$-scattering partition.
\end{observation}
\Cref{obs:StrongImplyScatt} follows as every path of weight $\sigma\Delta$ is contained in a ball of radius $\sigma\Delta$. \Cref{obs:beta1} follows as every shortest path of length $\le\Delta$ can be assembled as a concatenation of at most $\sigma$ shortest paths of length $\le\frac\Delta\sigma$.

\begin{TwoLiners}
	\item[\textbf{General Graphs}:] Given an $n$-vertex general graph and parameter $\Delta>0$ we construct a single partition $\mathcal{P}$ which is simultaneously $\left(8k,O(n^{\nicefrac1k}\cdot \log n),\Delta\right)$-strong sparse partition for all parameters $k\ge 1$ (\Cref{thm:Generalstrong}). Thus we generalize the result of \cite{JLNRS05} and obtain a strong diameter guarantee.
	This partition implies that general graphs are $\left(O(\log n),O(\log n)\right)$-scatterable (reproving \cite{KKN15} via an easier proof), inducing a solution for the \SPR problem with stretch $\polylog(|K|)$ (\Cref{cor:general-SPR}). While quantitatively better solutions are known, this one is arguably the simplest, and induced by a general framework.
	Further, we provide a lower bound, showing that if all $n$-vertex graphs admit $(\sigma,\tau)$-weak sparse partition scheme, then $\tau\ge n^{\Omega(\frac{1}{\sigma})}$ (\Cref{thm:GeneralWeakLB}). In particular there is no sparse partition scheme with parameters smaller than $\left(\Omega(\nicefrac{\log n}{\log\log n}),\Omega(\log n)\right)$. This implies that both our results and \cite{JLNRS05} are tight up to second order terms.
	Although we do not provide any lower bound for scattering partitions, we present some evidence that general graphs are not $\left(O(1),O(1)\right)$-scatterable. Specifically, we define a stronger notion of partitions called \emph{super-scattering} and show that  general graphs are not $\left(1,\Omega(\log n)\right)$-super scatterable  (\Cref{thm:generalLBsuper}).
	
	\item[\textbf{Trees}:] Trees are the most basic of the restricted graph families. Weak sparse partitions for trees follows from the existence of sparse covers. Nevertheless, in order to improve parameters and understanding we construct $(4,3)$-weak sparse partition scheme for trees (\Cref{thm:treeWeak}).
	Further, we prove that trees are $(2,3)$-scatterable (\Cref{thm:tressScat}).
	Finally, we show that there are no good strong sparse partition for trees. Specifically, we prove that if all $n$-vertex trees admit $(\sigma,\tau)$-strong sparse partition scheme, then $\tau\ge \frac13\cdot n^{\frac{2}{\sigma+1}}$ (\Cref{thm:treeLBStrong}). This implies that for strong sparse partitions, trees are essentially as bad as general graphs.
	
	\item[\textbf{Doubling Dimension}:] We prove that for every graph with doubling dimension $\ddim$ and parameter $\Delta>0$, there is a partition $\mathcal{P}$ which is simultaneously $\left(58\alpha,2^{\nicefrac{\tddim}{\alpha}}\cdot\tilde{O}(\ddim),\Delta\right)$-strong sparse partition for all parameters $\alpha\ge 1$ (\Cref{thm:ddimStrong}). Note that this implies an $\left(O(\ddim),\tilde{O}(\ddim)\right)$-strong sparse partition scheme.
	
	\item[\textbf{Euclidean Space}:] We prove that the $d$-dimensional Euclidean space $(\R^d,\|\cdot\|_2)$  is $(1,2d)$-scatterable \footnote{See \Cref{sec:Euclidean} for clarifications.} (\Cref{thm:L2scattering}), while for every $(\sigma,\tau)$-weak sparse partition scheme it holds that $\tau>(1+\frac{1}{2\sigma})^d$ (\Cref{thm:L2LB}). In particular, if $\sigma$ is at most a constant, then $\tau$ must be exponential. This provides an interesting example of a family where scattering partitions have considerably better parameters than sparse partitions.
	
	\item[\textbf{\SPDdepth}:] \footnote{Every (weighted) path graph has an \SPDdepth $1$. A graph $G$ has an \SPDdepth $\rho$ if there exist a \emph{shortest path} $P$, such that every connected component in $G\setminus P$ has an \SPDdepth $\rho-1$. This family includes graphs with pathwidth at most $\rho$, and more. See \cite{AFGN23}.\label{foot:SPD}} We prove that every graph with \SPDdepth $\rho$ (in particular graph with pathwidth $\rho$) admit $\left(O(\rho),O(\rho^2)\right)$-strong sparse partition scheme (\Cref{thm:SPDstrong}). 
	Further, we prove that such graphs admit $\left(8,5\rho\right)$-weak sparse partition scheme (\Cref{thm:SPDweak}).
	
	\item[\textbf{Chordal Graphs}:]
	 We prove that every Chordal graph is $(2,3)$-scatterable (\Cref{thm:chordalScat}).
	 
	\item[\textbf{Cactus Graphs}: ] We prove that every Cactus graph is $(4,5)$-scatterable (\Cref{thm:cactus}). 
\end{TwoLiners}

\subsection{Follow-Up Work}
In a follow up work, Hershkowitz and Li \cite{HL22} proved that series-parallel graphs are $(O(1),O(1))$-scatterable. Using our framework (\Cref{thm:Scattering_Implies_SPR}) they concluded that every instance of the \SPR problem in series parallel graphs admit a solution with $O(1)$ distortion. 
Recently, Chang \etal \cite{CCLMST23,CCLMST24} constructed shortcut partitions with constant parameters for planar, and minor free graphs. Shortcut partitions are a relaxed version of our scattering partitions. 
Chang \etal then generalized our \Cref{thm:Scattering_Implies_SPR} to shortcut partitions, and concluded that planar graphs \cite{CCLMST24Steiner}, as well as fixed minor free graph \cite{CCLMST24} admit a solution with constant stretch for the \SPR problem. The question of whether these graph families admit scattering partitions remains open.

Recently, the author \cite{Fil24} construct sparse covers, and sparse partition for fixed minor free graphs with parameters exponentially improving over our \Cref{cor:WeakMinor}.
Another recent work by Busch \etal \cite{BCFHHR23} constructed a hierarchy of strong sparse partitions for general graphs, doubling graphs, and graphs with bounded pathwidth (building on the techniques in this paper). These partitions implied a solution to the subgraph \UST problem with poly-logarithmic stretch.
Finally, sparse partitions were recently used to solve the facility location problem in the streaming model \cite{CJK0Y22} (they called it geometric hashing). In the context of our paper, Czumaj \etal \cite{CJK0Y22} constructed space efficient sparse partition for the Euclidean space with the same parameters as our \Cref{thm:ddimStrong} (see the arXiv version \cite{CFJKVY22}), or space and time efficient partitions with worse parameters.

\subsection{Technical Ideas}

\paragraph{Scattering Partition Imply \SPR.}
Similarly to previous works on the \SPR problem, we construct a minor via a terminal partition. That is, a partition of $V$ into $k$ connected clusters, where each cluster contains a single terminal. The minor is then induced by contracting all the internal edges.
Intuitively, to obtain small distortion, one needs to ensure that every Steiner vertex is clustered into a terminal not much further than its closest terminal, and that every shortest path between a pair of terminals intersects only a small number of clusters. However, the local partitioning of each area in the graph requires a different scale, according to the distance to the closest terminal.
Our approach is similar in spirit to the algorithm of Englert \etal \cite{EGKRTT14}, who constructed a minor with small expected distortion \footnote{A distribution $\mathcal{D}$ over solutions to the \SPR problem has expected distortion $\alpha$ if $\forall t,t'\in K,~\mathbb{E}_{M\sim\mathcal{D}}[d_M(t,t')]\le\alpha\cdot d_G(t,t')$ .\label{foot:expectedDistortion}} using stochastic decomposition for all possible distance scales. We however, work in the more restrictive regime of worst case distortion guarantee.
Glossing over many details, we create different scattering partitions for different areas, where vertices at distance $\approx\Delta$ to the terminal set are partitioned using a $(1,\tau,\Delta)$-scattering partition. Afterwards, we assemble the different clusters from the partitions in all possible scales into a single terminal partition. We use the scattering property twice. First to argue that each vertex $v$ is clustered to a terminal at distance at most $O(\tau)\cdot D(v)$ (here $D(v)$ is the distance to the closest terminal). Second, to argue that every shortest path, where all the vertices are at similar distance to the terminal set, intersects the clusters of at most $O(\tau^2)$ terminals. 

\paragraph{Miller, Peng and Xu \cite{MPX13} clustering algorithm.}
We use \cite{MPX13} to create partitions for general graphs, and graphs with either bounded doubling dimension or \SPDdepth. 
In short, there is a set of centers $N$, where each center $t$ samples a starting time $\delta_t$. Vertex $v$ joins the cluster of the center $t$ maximizing $f_v(t)=\delta_t-d_G(t,v)$. Denote this center by $t_v$. 
The diameter guarantee obtained using this  algorithm is inherently that of a strong diameter.
The key observation is the following: if the cluster of the center $t$ intersects the ball $B_G(v,r)$, then necessarily $f_v(t)\ge f_v(t_v)-2r$.  Thus in order to bound the number of intersecting clusters it is enough to bound the number of centers whose $f_v$ values falls inside the interval $[f_v(t_v)-2r,f_v(t_v)]$. 
For each family we choose the starting times $\{\delta_t\}_{t\in N}$ appropriately. 	

\paragraph{Strong Sparse Partition for General Graphs.}
For general graphs, we use the \cite{MPX13} clustering algorithm with the set of all vertices as centers. The starting times $\{\delta_t\}_{t\in N}$ are chosen i.i.d. using an exponential distribution with expectation $O(\frac{\Delta}{\log n})$. We then use the memoryless property to argue that for every vertex $v$, and parameter $\alpha$, w.h.p. there are no more than $\tilde{O}(n^{\nicefrac{1}{\alpha}})$ centers whose $f_v$ value falls in $[f_v(t_v)-\Omega(\frac{\Delta}{\alpha}),f_v(t_v)]$. 

\paragraph{Doubling Dimension.}
In a companion paper, the author \cite{Fil19Padded} constructed a strong sparse cover scheme for doubling graphs. Together with \Cref{lem:coverToPartition}, this implies sparse partition with weak diameter guarantee only (and thus does not imply scattering). While both \cite{Fil19Padded} and \Cref{thm:ddimStrong} are \cite{MPX13} based, here we have additional complications.
For a graph with doubling dimension $\ddim$ we use \cite{MPX13}, with a $\Delta$-net serving as the set of centers.
The idea is to use the same analysis as for general graphs in a localized fashion, where each vertex is ``exposed'' to $2^{O(\sddim)}$ centers, thus replacing the $\log n$ parameter in the number of intersections with $\ddim$.
The standard solution will be to use a \emph{truncated exponential distribution}.\footnote{A truncated exponential distribution is an exponential distribution conditioned on the event that the outcome lies in a certain interval. This is usually used when the maximal possible value must be bounded. See e.g. \cite{Fil19Padded}.}
Although this will indeed guarantee the ``locality'' we are aiming for, it is not clear how to analyze the density of the centers in $[f_v(t_v)-\Omega(\frac{\Delta}{\alpha}),f_v(t_v)]$. 
Instead, we are using \emph{betailed exponential distribution}. This is an exponential distribution with a threshold parameter $\lambda_T$ such that all possible values above the threshold collapse to the threshold. 
In order to bound the density of the centers in $[f_v(t_v)-\Omega(\frac{\Delta}{\alpha}),f_v(t_v)]$, we first treat the distribution as a standard exponential distribution. Afterwards we argue that the actual number of ``betailed'' centers is small. Interestingly, we bound the number of centers whose $f_v$ value falls in $[f_v(t_{s})-\Omega(\frac{\Delta}{\alpha}),f_v(t_{s})]$, where $t_{s}$ is the center with the $s$'th largest $f_v$ value. Then, we argue that the bound on the density in the actual interval we care about, can withstand any $s-1$ occurrences of ``betailing''.
Eventually, we use the Lov\'asz Local Lemma to argue that we can choose the starting times such that the density of centers in all the possible intervals is small (``intervalwise'').

\paragraph{\SPD.} Our strong sparse partition for a graph with \SPDdepth $\rho$ is also produced using the \cite{MPX13} clustering algorithm. Interestingly, unlike all previous executions of \cite{MPX13} in the literature, there is no randomness involved.
The \SPD produces a hierarchy of partial partitions $\{\mathcal{X}_1,\dots,\mathcal{X}_\rho\}$, where $\mathcal{P}_i$ is the set of shortest paths deleted from each of the connected components in $\mathcal{X}_i$, to create $\mathcal{X}_{i+1}$.
The set of centers consists of $\eps\Delta$-nets $\{N_i\}_{i=1}^\rho$ taken from all the shortest paths $\{\mathcal{P}_i\}_{i=1}^\rho$, where the starting time $\delta_t$ is equal for all centers $t\in N_i$ with the same hierarchical depth, and decreasing with the steps ($t\in N_i,t'\in N_{i+1}\Rightarrow \delta_t>\delta_{t'}$).
The parameters are chosen in such a way that each vertex $v$ can join the cluster of a center $t\in N_i\subseteq \mathcal{P}_i$, only if $v$ and $t$ belong to the same connected component in $\mathcal{X}_i$.
Consider a small ball $B$. If a vertex from $B$ belongs to $\mathcal{P}_i$, then no vertex of $B$ will join a cluster of a center in $\mathcal{P}_{i'}$ for $i'>i$.
To conclude, in each step $i$, the vertices of $B$ might join the cluster of centers lying on a single shortest path from $\mathcal{P}_i$. As all these centers have equal starting time, $B$ can intersect only a small number of them. A bound on the number of intersections follows.

\paragraph{Chordal Graphs.} The partition is created inductively using the tree decomposition (where each bag is a clique). The label of a vertex $v$, is its distance in $G$ to the cluster center $t$ (w.r.t. $d_G$). We maintain the property that all the vertices with label strictly smaller than $\Delta/2$, within a single bag, belong to the same cluster. Interestingly, while the partition is connected, the diameter guarantee we obtain is only weak (in contrast to all other scattering partitions in the paper).

\paragraph{Lower Bounds.} For the lower bound for strong sparse partitions in trees we use a full $d$-ary tree, of depth $D$. It holds that for every partition with diameter smaller than $2D$, there must be an internal vertex with all its children belonging to different clusters. Thus we have a ball of radius $1$ intersecting $d+1$ clusters. Noting that such a tree has less than $2\cdot d^{D+1}$ vertices, the theorem follows. \\
For the lower bound on weak partitions of general graphs we use expanders. The cluster $A$ will intersect all radius $\frac\Delta\sigma$ balls with centers in $B_G(A,\frac\Delta\sigma)$.
By expansion property, $|B_G(A,\frac\Delta\sigma)|\ge \Omega(1)^{\frac\Delta\sigma}\cdot |A|$.
We uniformly sample a center $v$, the expected number of clusters intersecting $B_G(v,\frac\Delta\sigma)$ is \\ $\Sigma_{A\in\mathcal{P}}\nicefrac{|B_G(A,\frac\Delta\sigma)|}{n}=\Omega(1)^{\frac\Delta\sigma}\cdot\Sigma_{A\in\mathcal{P}}\nicefrac{|A|}{n}=\Omega(1)^{\frac\Delta\sigma}$.\\
For Euclidean space the lower bound follows similar ideas, where the expansion property is replaced by the Brunn-Minkowski theorem.

\subsection{Related Work}\label{subsec:related}
In the functional analysis community, the notion of \emph{Nagata dimension} was studied. The Nagata dimension of a metric space $(X,d)$, $\dim_N X$, is the infimum over all integers $n$ such that there exists a constant $c$ s.t. $X$ admits a $(c,n+1)$-weak sparse partition scheme. In contrast, in this paper our goal is to minimize this constant $c$. See \cite{LS05} and the references therein.

A closely related problem to \UST is the \emph{Universal Traveling Salesman Problem} (\UTSP).
Consider a postman providing post service for a set $X$ of clients with $n$ different locations (with distance measure $d_X$). Each morning the postman receives a subset $S\subset X$ of the required deliveries for the day. In order to minimize the total tour length, one solution may be to compute each morning an (approximation of an) Optimal \TSP tour for the set $S$. An alternative solution will be to compute a \emph{Universal \TSP} (\UTSP) tour. This is a universal tour $R$ containing all the points $X$. Given a subset $S$, $R(S)$ is the tour visiting all the points in $S$ w.r.t. the order induced by $R$.
Given a tour $T$ denote its length by $|T|$. The \emph{stretch} of $R$ is the maximum ratio among all subsets $S\subseteq X$ between the length of $R(S)$ and the length of the optimal \TSP tour on $S$, $\max_{S\subseteq X}\frac{|R(S)|}{|\mbox{Opt}(S)|}$.

All the sparse partition based upper bounds for the \UST problem translated directly to the \UTSP problem with the same parameters.
The first to study the problem were Platzman and Bartholdi \cite{PB89}, who given $n$ points in the Euclidean plane constructed a solution with stretch $O(\log n)$, using \emph{space filling curves}.
Recently, Christodoulou, and Sgouritsa \cite{CS17} proved a lower and upper bound of $\Theta(\log n/\log\log n)$ for the $n\times n$ grid, improving a previous $\Omega(\sqrt[6]{\log n/\log\log n})$ lower bound of Hajiaghayi, Kleinberg, and Leighton \cite{HKL06} (and the $O(\log n)$ upper bound of \cite{PB89}).
For general $n$ vertex graphs Gupta \etal \cite{GHR06} proved an  $O(\log^2 n)$ upper bound, while
Gorodezky, Kleinberg, Shmoys, and Spencer \cite{GKSS10} proved an $\Omega(\log n)$ lower bound.
From the computational point of view, Schalekamp and Shmoys \cite{SS08} showed that if the input graph is a tree, an \UTSP with optimal stretch can be computed efficiently.

The \emph{A Priori \TSP} problem is similar to the \UTSP problem. In addition there is a distribution $\mathcal{D}$ over subsets $S\subseteq V$ and the stretch of tour a $R$ is the expected ratio between the induced solution to optimal  $\mathbb{E}_{S\sim\mathcal{D}}\frac{|R(S)|}{|\mbox{Opt}(S)|}$ (instead of a worst case like in \UTSP). Similarly, \emph{A Priori Steiner Tree} was studied (usually omitting $\rt$ from the problem). See \cite{Jal88,SS08,GKSS10} for further details.
Another similar problem is the  \emph{Online (or dynamic) Steiner Tree} problem. Here the set $S$ of vertices that should be connected is evolving over time, see \cite{IW91,AA92,GGK16} and references therein.

Unlike the definition used in this paper (taken from \cite{JLNRS05}), sparse partitions were also defined in the literature as partitions where only a small fraction of the edges are inter-cluster (see for example \cite{AGMW10}). 
A closely related notion to sparse partitions are padded and separating decompositions. 
A graph $G$ is $\beta$-decomposable if for every $\Delta>0$, there is a distribution $\mathcal{D}$ over $\Delta$ bounded partitions such that for every $u,v\in V$, the probability that $u$ and $v$ belong to different clusters is at most $\beta\cdot\frac{d_G(u,v)}{\Delta}$. Note that by linearity of expectation, a path $\mathcal{I}$ of length $\Delta/\sigma$ intersects at most $1+\beta/\sigma$ clusters in expectation. For comparison, in scattering partition we replace the distribution by a single partition and receive a bound on the number of intersections in the worst case. See \cite{KPR93,Bar96,FT03,GKL03,ABN11,AGMW10,AGGNT19,FN22,Fil19Padded} for further details.

Englert \etal \cite{EGKRTT14} showed that every graph which is $\beta$-decomposable, admits a distribution $\mathcal{D}$ over solution to the \SPR problem with expected distortion $O(\beta\log \beta)$. $^{\ref{foot:expectedDistortion}}$ 
In particular this implies constant expected distortion for graphs excluding a fixed minor, or bounded doubling dimension.

For a set $K$ of terminals of size $k$, Krauthgamer, Nguyen and Zondiner \cite{KNZ14} showed that if we allow the minor $M$ to contain at most ${k\choose 2}^2$ Steiner vertices (in addition to the terminals), then distortion $1$ can be achieved. They further showed that for graphs with constant treewidth, $O(k^2)$ Steiner points will suffice for distortion $1$.
Cheung, Gramoz and Henzinger \cite{CGH16} showed that allowing $O(k^{2+\frac2t})$ Steiner vertices, one can achieve distortion $2t-1$. For planar graphs, Cheung \etal al. achieved $1+\eps$ distortion with $\tilde{O}((\frac k\epsilon)^2)$ Steiner points.

There is a long line of work focusing on preserving the cut/flow structure among the terminals by a graph minor. See 
\cite{Moitra09,LM10,CLLM10,MM10,EGKRTT14,Chuzhoy12,KR13,AGK14,GHP17,KR20}.

There were works studying metric embeddings and metric data structures concerned with preserving distances among terminals, or from terminals to other vertices, out of the context of minors. See \cite{CE06,RTZ05,GNR10,KV13,EFN17,EFN18,BFN19,EN21,KLR19,FGK20,EN22prior}.

\section{Preliminaries}\label{sec:prem}
All the logarithms in the paper are in base $2$. 
We use $\tilde{O}$ notation to suppress constants and logarithmic factors, that is $\tilde{O}(f(j))=f(j)\cdot\polylog(f(j))$.
\paragraph{Graphs.}
We consider connected undirected graphs $G=(V,E)$ with edge weights
$w: E \to \R_{\ge 0}$. 
Let $d_{G}$ denote the shortest path metric in $G$.
$B_G(v,r)=\{u\in V\mid d_G(v,u)\le r\}$ is the ball of radius $r$ around $v$. For a vertex $v\in V$ and a subset $A\subseteq V$, let $d_{G}(x,A):=\min_{a\in A}d_G(x,a)$, where $d_{G}(x,\emptyset)= \infty$. For a subset of vertices
$A\subseteq V$, let $G[A]$ denote the induced graph on $A$, and let $G\setminus A := G[V\setminus A]$.

\paragraph{Doubling dimension.}  The doubling dimension of a metric space is a measure of its local ``growth rate''. 
A metric space $(X,d)$ has doubling constant $\lambda$ if for every $x\in X$ and radius
$r>0$, the ball $B(x,2r)$ can be covered by $\lambda$ balls of radius $r$. The doubling dimension is defined as $\ddim=\log\lambda$.
We say that a weighted graph $G=(V,E,w)$ has doubling dimension $\ddim$, if the corresponding shortest path metric $(V,d_G)$ has doubling dimension $\ddim$.
A $d$-dimensional $\ell_p$ space has $\ddim=\Theta(d)$, every $n$ point vertex graph has $\ddim=O(\log n)$, and every weighted path has $\ddim=1$.
The following lemma gives the standard packing property of doubling metrics (see, e.g., \cite{GKL03}).
\begin{lemma}[Packing Property] \label{lem:doubling_packing}
	Let $(X,d)$ be a metric space  with doubling dimension $\ddim$.
	If $S \subseteq X$ is a subset of points with minimum interpoint distance $r$ that is contained in a ball of radius $R$, then 
	$|S| \le \left(\frac{2R}{r}\right)^{\sddim}$ .
\end{lemma}

\paragraph{Nets.} A set $N\subseteq V$ is called a $\Delta$-net, if for every vertex $v\in V$ there is a net point $x\in N$ at distance at most $d_G(v,x)\le \Delta$, while every pair of net points $x,y\in N$, is farther than $d_G(x,y)>\Delta$.
A $\Delta$-net can be constructed efficiently in a greedy manner. In particular, by \Cref{lem:doubling_packing}, given a $\Delta$-net $N$ in a graph of doubling dimension $\ddim$, a ball of radius $R\ge\Delta$, will contain at most 
$\left(\frac{2R}{\Delta}\right)^{\sddim}$ net points.

\paragraph{Exponential Distribution}
$\Exp(\lambda)$ denotes the \emph{exponential distribution} with mean
$\lambda$ and density function $f(x)=\frac{1}{\lambda}e^{-\frac{x}{\lambda}}$ for $x\ge0$. 
A useful property of exponential distribution is \emph{memoryless}: let $X\sim\Exp(\lambda)$, for every $a,b\ge 0$, $\Pr[X\ge a+b\mid X\ge a]=\Pr[X\ge b]$. In other words, given that $X\ge a$, it holds that $X-a\sim \Exp(\lambda)$. 

\paragraph{Special graph families.}
A graph $H$ is a \emph{minor} of a graph $G$ if we can obtain $H$ from
$G$ by edge deletions/contractions, and vertex deletions.  A graph
family $\mathcal{G}$ is \emph{$H$-minor-free} if no graph
$G\in\mathcal{G}$ has $H$ as a minor.
Some examples of minor free graph families are planar graphs ($K_5$ and $K_{3,3}$ free), outerplanar graphs ($K_4$ and $K_{3,2}$ free), series-parallel graphs ($K_4$ free), Cactus graphs (also known as tree of cycles) (\includegraphics[width=0.02\textwidth]{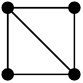} free), and trees ($K_3$ free).

Given a graph $G=(V,E)$, a \emph{tree decomposition} of $G$ is a tree
$T$ with nodes $B_1,\dots,B_s$ (called \emph{bags}) where each $B_i$ is
a subset of $V$ such that the following properties hold:
\begin{OneLiners}
	\item For every edge $\{u,v\}\in E$, there is a bag $B_i$ containing
	both $u$ and $v$.
	\item For every vertex $v\in V$, the set of bags containing $v$ form a
	connected subtree of $T$.
\end{OneLiners}
The \emph{width} of a tree decomposition is $\max_i\{|B_i|-1\}$. The \emph{treewidth} of $G$ is the minimal
width of a tree decomposition of $G$.
A \emph{path decomposition} of $G$ is a special kind of tree
decomposition where the underlying tree is a path. The \emph{pathwidth}
of $G$ is the minimal width of a path decomposition of $G$.

\emph{Chordal graphs} are unweighted graphs where each cycle of length greater
then $4$ contains a chordal. In other words, if the induced graph on a set of vertices $V'$ is the cycle graph, than necessarily $|V'|\le 3$.
Chordal graphs contain interval graphs, subtree intersection graphs and other interesting sub families.
A characterization of Chordal graphs is that they have a tree decomposition such that each bag is a clique. That is, there is a tree decomposition $T$ of $G$ where there is no upper bound on the size of a bag, but for every bag $B\in T$ the induced graph $G[B]$ is a clique.

A \emph{Cactus graph} (a.k.a. tree of cycles) is a graph where each edge belongs to at most one simple cycle. Alternatively it can be defined as the graph family that excludes $K_4$ minus an edge  (\includegraphics[width=0.02\textwidth]{CactusMinorFree}) as a minor.

Abraham \etal \cite{AFGN23} defined \emph{shortest path decompositions} (\SPDs) of ``low depth''.
Every (weighted) path graph has an \SPDdepth $1$. A graph $G$ has an \SPDdepth
$k$ if there exist a \emph{shortest path} $P$, such that every connected
component in $G\setminus P$ has an \SPDdepth $k-1$.
In other words, given a graph, in \SPD we hierarchically delete shortest paths from each connected component, until no vertices remain. See \Cref{sec:SPD} for formal definition (and also \cref{foot:SPD}).
Every graph with pathwidth $\rho$ has \SPDdepth at most $\rho+1$, treewidth $\rho$ implies \SPDdepth at most $O(\rho\log n)$, and every graph excluding a fixed minor has \SPDdepth $O(\log n)$.
See \cite{AFGN23,Fil20face} for further details and applications.

\subsection{Clustering algorithm using shifted starting times \cite{MPX13}}\label{subsec:MPX}
In some of our clustering algorithms we will use a generalized\footnote{There are two difference between the algorithm here and the original \cite{MPX13} clustering algorithm. First, in \cite{MPX13} all the vertices are potential cluster centers (while we are allowing to use only a subset of the vertices as cluster centers). Secondly, in \cite{MPX13} the shifts $\delta_t$ are sampled randomly using exponential distribution. In contrast, here we are allowing to choose the shifts in an arbitrary manner. In particular, in \Cref{thm:Generalstrong,thm:ddimStrong} we will sample the shifts using betailed exponential distribution, while in \Cref{thm:SPDstrong} the shifts will be chosen deterministically.} version of the clustering algorithm by Miller \etal \cite{MPX13}.
This version has also been applied by the author in a companion paper \cite{Fil19Padded}. Here we describe the algorithm and some of its basic properties. One of the main advantages of this algorithm is that it inherently produces clusters with strong diameter. Additionally, the inter-relationship between the clusters is relatively easy to analyze.  

Consider a weighted graph $G=(V,E,w)$. Let $N\subseteq V$ be some set of centers, where each $t\in N$ admits a parameter $\delta_t$. The choice of the centers and the parameters differs among different implementations of the algorithm. 
For each vertex $v$ set a function $f_v:N\rightarrow \R$ as follows: for a center $t$, $f_v(t)=\delta_t-d_G(t,v)$. The vertex $v$ will join the cluster $C_t$ of the center $t\in N$ maximizing $f_v(t)$.  Ties are broken in a consistent manner.\footnote{That is we have some order $x_1,x_2,\dots$. Among the centers $x_i$ that maximize $f_v(t)$, $v$ joins the cluster of the center with minimal index.}
Note that it is possible that a center $t\in N$ will join the cluster of a different center $t'\in N$.
An intuitive way to think about the clustering process is as follows: each center $t$ wakes up at time $-\delta_t$ and begins to ``spread'' in a continuous manner. The spread of all centers is performed in the same unit tempo. A vertex $v$ joins the cluster of the first center that reaches it. 

\begin{claim}\label{clm:MPXshortestpath}
	Suppose that a vertex $v$ joined the cluster of the center $t$. Let $\mathcal{I}$ be a shortest path from $v$ to $t$, then all the vertices on $\mathcal{I}$ joined the cluster of $t$.
\end{claim}
\begin{proof}
	For every vertex $u\in\mathcal{I}$ and center $t'\in N$, it holds that
	\[
	\delta(t)-d_{G}(u,t)=\delta(t)-\left(d_{G}(v,t)-d_{G}(v,u)\right)\ge\delta(t')-d_{G}(v,t')+d_{G}(v,u)\ge\delta(t')-d_{G}(u,t')~,
	\]
	where the first inequality holds as  $f_v(t)\ge f_v(t')$. We conclude that $f_u(t)\ge f_u(t')$, hence $u$ joins the cluster of $t$.
\end{proof}
By \Cref{clm:MPXshortestpath} it follows that if the distance between every vertex to the cluster center $v$ is bounded by $\Delta$ w.r.t. $d_G$, then the cluster has strong diameter $2\Delta$.
\begin{claim}\label{clm:MPXintersectionProperty}
	Consider a ball $B=B_G(v,r)$, and let $t$ be the center maximizing $f_v$. Then for every center $t'$ such that $f_v(t)-f_v(t')>2r$, the intersection between $B$ and the cluster centered at $t'$ is empty.
\end{claim}
\begin{proof}
	For every vertex $u\in B$ it holds that,
	\[
	f_{u}(t)=\delta(t)-d_{G}(u,t)\ge \delta(t)-d_{G}(v,t)-r>\delta(t')-d_{G}(v,t')+r\ge \delta(t')-d_{G}(u,t')=f_{u}(t')~,
	\]
	where the second inequality holds as $f_v(t)>f_v(t')+2r$. It follows that $t'$ is not maximizing $f_u$ for any vertex in $B$, the claim follows.
\end{proof}

\section{From Scattering Partitions to \SPR: Proof of \Cref{thm:Scattering_Implies_SPR}}\label{sec:scat_to_spr}
This entire section is devoted to the proof of  \Cref{thm:Scattering_Implies_SPR}.
We will assume w.l.o.g. that the minimal pairwise distance in the graph is exactly $1$, otherwise we can scale all the weights accordingly.
The set of terminals is denoted by $K=\{t_1,\dots,t_k\}$.
For every vertex $v\in V$, denote by $D(v)=d_G(v,K)$ the distance to its closest terminal. Note that $\min_{v\in V\setminus K}D(v)\ge 1$.

	Similarly to previous papers on the \SPR problem, we will create a minor using \emph{terminal partitions}.
	Specifically, we partition the vertices into $k$ connected clusters, with a single terminal in each cluster. Such a partition induces a minor by contracting all the internal edges in each cluster.
	\begin{figure}[]
		\centering{\includegraphics[scale=0.85]{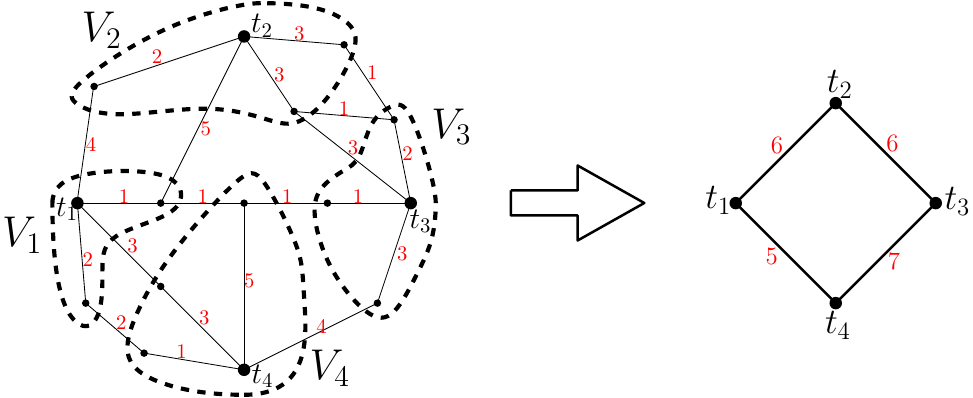}} 
	\caption{\label{fig:contraction}\small \it 
	The left side of the figure contains a weighted graph $G=(V,E)$, with weights specified in red, and four terminals $\{t_1,t_2,t_3,t_4\}$.
	The dashed black curves represent a terminal partition of the vertex set $V$ into the subsets $V_1,V_2,V_3,V_4$.
	The right side of the figure represents the minor $M$ induced by the terminal partition. The distortion is realized between $t_1$ and $t_3$, and is $\frac{d_M(t_1,t_3)}{d_G(t_1,t_3)}=\frac{12}4=3$.}
	\end{figure}
	More formally, a partition $\{V_1,\dots,V_k\}$ of $V$ is called a terminal partition (w.r.t to $K$) if for every $1\le i\le k$, $t_i\in V_i$, and the induced graph $G[V_i]$ is connected. 
		For a vertex $v\in V_i$, we say that $v$ is assigned to $t_i$.
	See \Cref{fig:contraction} for an illustration.
	The \emph{induced minor} by the terminal partition  $\{V_1,\dots,V_k\}$, is a minor $M$, where each set $V_i$ is contracted into a single vertex called (abusing  notation)  $t_i$. 
	Note that there is an edge in $M$ from  $t_i$ to $t_j$ if and only if there are vertices $v_i\in V_i$ and $v_j\in V_j$ such that $\{v_i,v_j\}\in E$. 
	We determine the weight of the edge $\{t_i,t_j\}\in E(M)$ to be $d_G(t_i,t_j)$. Note that by the triangle inequality, for every pair of (not necessarily neighboring) terminals $t_i,t_j$, it holds that $d_M(t_i,t_j)\ge d_G(t_i,t_j)$.
	The \emph{distortion} of the induced minor is  $\max_{i,j}\frac{d_M(t_i,t_j)}{d_G(t_i,t_j)}$.

	\subsection{Algorithm}\label{subsec:scat_to_spr_alg}
We create the terminal partition in an iterative manner, where initially each set $\{t_i\}$ is a singleton, and gradually more vertices are joining.
We will denote the stage of the terminal partition after $i$ steps, using a function $f_i:V\rightarrow K\cup\{\perp\}$. For a yet unassigned vertex $v$ we write  $f_i(v)=\perp$, otherwise the vertex $v$ will be assigned to $f_i(v)$. Initially for every terminal $t_j$, $f_0(t_j)=t_j$ while for every Steiner vertex $v\in V\setminus K$, $f_0(v)=\perp$. In iteration $i$ we will define $f_i$ by ``extending'' $f_{i-1}$. That is, unassigned vertices may be assigned (i.e., for $v$ such that $f_{i-1}(v)=\perp$ it might be $f_i(v)=t_j$), while the function will remain the same on the set of assigned vertices ($f_{i-1}(v)\ne\perp\ \Rightarrow\ f_i(v)=f_{i-1}(v)$). 
We will guarantee that all the vertices in $\cR_i$ will be assigned in $f_i$. In particular, after $\log\left(\max_v D(v)\right)$ steps, all the vertices will be assigned.

Denote by $V_{i}$ the set of vertices assigned by $f_{i}$. Initially $V_0=K=\cR_0$. 
By induction we will assume that $\cup_{j\le i-1}\cR_j\subseteq V_{i-1}$.
Let $G_i=G[V\setminus V_{i-1}]$ be the graph induced by the set of yet unassigned vertices.
Fix $\Delta_i=2^{i-1}$. Let $\cP_i$ be an $(1,\tau,\Delta_i)$-scattering partition of $G_i$. 
Let $\cC_i\subseteq\cP_i$ be the set of clusters $C$ which contain at least one vertex $v\in \cR_i$.
	All the vertices in $\cup\cC_i$ (that is the union of all the clusters in $\cC_i$) will be assigned by $f_i$.
	
We say that a cluster $C\in \cC_{i}$ is at level $1$, noting $\delta_i(C)=1$, if there is an edge $\{v,u_C\}$ (in $G$) from a vertex $v\in C$ to a vertex $u_C\in V_{i-1}$ of weight at most $2^{i}$.
In general, $\delta_i(C)=l$, if $l$ is the minimal index such that there is an edge $\{v,u_C\}$ from a vertex $v\in C$ to a vertex $u_C\in C'$ of weight at most $2^{i}$, such that $\delta_i(C')=l-1$.
In both cases $u_C$ is called the \emph{linking} vertex of $C$.
Next, we define $f_i$ based on $f_{i-1}$.
For every vertex $v\in V_{i-1}$ set $f_i(v)=f_{i-1}(v)$. For every vertex not in $\cup\cC_i\cup V_{i-1}$ set $f_i(v)=\perp$.
For a cluster $C\in \cC_{i}$ s.t. $\delta_i(C)=1$, let $u_C\in V_{i-1}$ be its linking vertex. For every $v\in C$ set $f_i(v)= f_i(u_C)$. 
Generally, for level $l$ suppose that $f_i$ is already defined on all the clusters of level $l-1$. Let $C\in \cC_{i}$ s.t. $\delta_i(C)=l$. Let $u_C$ be the linking vertex of $C$. For every $v\in C$, set $f_i(v)= f_i(u_C)$. 
Note that for every cluster, all the vertices are mapped to the same terminal.
This finishes the definition of $f_i$. 
	
The algorithm continues until there is $f_i$ where all the Steiner vertices are assigned. Set $f=f_i$. The algorithm returns the terminal-centered minor $M$ of $G$ induced by $\{f^{-1}(t_1),\ldots,f^{-1}(t_k)\}$.

	\subsection{Basic Properties}\label{subsec:basciProp}
	It is straightforward from the construction that $f^{-1}(t_1),\ldots,f^{-1}(t_k)$ define a terminal partition. We will prove that every vertex $v$ will be assigned during either iteration $\lceil\log D(v)\rceil$ or $\lceil\log D(v)\rceil-1$ (\Cref{clm:scatteringAssignTime}), to a terminal at distance at most $O(\tau)\cdot D(v)$ from $v$ (\Cref{cor:dvfv}).
	We begin by arguing that in each iteration, the maximum possible level of a cluster is $\tau$.
	\begin{claim}\label{clm:scatteringLevelBound}
	For every cluster $C\in\cC_i$, $\delta_i(C)\le\tau$.
\end{claim}
\begin{proof}

Consider a cluster $C\in\cC_i$, and let $v\in C$ be a  vertex s.t. $D(v)\le2^{i}$. Let $P_v=\{v=v_0,\dots,v_q\}$ be a shortest path from $v$ to a terminal $v_q\in K$ of total weight $D(v)$. Let $s+1$ be a minimal index of a clustered vertex. That is $f(v_{s+1})\ne \perp$, and for every $j\le s$, $f(v_s)=\perp$. Such an index exist as $f(v_q)=v_q$. 
 In particular, $v_{s+1}\in V_{i-1}$.
The prefix $P=\{v=v_0,\dots,v_s\}$ is a shortest path in $G_i$.
$\cP_i$ is a $(1,\tau,\Delta_i)$-scattering partition. Hence the vertices of $P$ are partitioned to $\tau'\le \tau$ clusters $C_1,\dots,C_{\tau'}$ where $v_s\in C_1$, $v_0\in C_{\tau'}$ and there is an edge from $C_j$ to $C_{j+1}$ of weight at most $2^{i-1}<2^i$, while the edge $\{v_s,v_{s+1}\}$ is from $C_1$ to $V_{i-1}$ of weight at most $2^i$. 
It holds that $\delta_i(C_1)=1$, and by induction $\delta(C_j)\le j$. In particular $\delta(C)\le \tau'\le\tau$.
\end{proof}

	\begin{claim}\label{clm:scatteringDistTof}
	For every vertex $v$ which is assigned during the $i$'th iteration (i.e., $v\in C\in\cC_i$) it holds that $d_G(v,f(v))\le 3\tau\cdot 2^i$.
\end{claim}
\begin{proof}
	The proof is by induction on $i$.
	For $i=0$ the assertion holds trivially as every terminal is assigned to itself.
	We will assume the assertion for $i-1$ and prove it for $i$. Let $C\in \cC_{i}$ be some cluster, and let $v\in C$.
	Suppose first that $\delta_i(C)=1$. Let $u_C\in V_{i-1}$ be the linking vertex of $C$. By the induction hypothesis $d_G(u_C,f(u_C))\le 3\tau\cdot 2^{i-1}$.
	As the diameter of $C$ is bounded by $2^{i-1}$, and the weight of the edge towards $u_C$ is at most $2^i$ we conclude $d_G(v,f(v))\le d_G(v,u_C)+d_G(u_C,f(u_C))\le (2^{i-1}+2^i)+3\tau\cdot 2^{i-1}= 3\cdot 2^{i-1}+3\tau\cdot 2^{i-1}$.
	Generally, for $\delta_i(C)=l$, we argue by induction that for every $v\in C$ it holds that $d_{G}(v,f(v))\le l\cdot3\cdot2^{i-1}+3\tau\cdot2^{i-1}$. Indeed, let $u_C$ by the linking vertex of $C$. By the induction hypothesis it holds that $d_G(u_C,f(u_C))\le (l-1)\cdot 3\cdot 2^{i-1}+3\tau\cdot 2^{i-1}$. Using similar arguments, it holds that 
	$d_{G}(v,f(v))\le d_{G}(v,u_{C})+d_{G}(u_{C},f(u_{C}))\le(2^{i-1}+2^{i})+(l-1)\cdot3\cdot2^{i-1}+3\tau\cdot2^{i-1}=l\cdot3\cdot2^{i-1}+3\tau\cdot2^{i-1}$.
	Using \Cref{clm:scatteringLevelBound}, 
	$d_G(v,f(v))\le 3\tau\cdot 2^{i-1}+3\tau\cdot 2^{i-1}= 3\tau\cdot 2^{i}$ as required.
\end{proof}

\begin{corollary}\label{cor:dvfv}
For every vertex $v$ it holds that $d_G(v,f(v))< 6\tau\cdot D(v)$.
\end{corollary}
\begin{proof}
	Let $i\ge 0$ su	ch that $2^{i-1}< D(v)\le 2^i$. The vertex $v$ is assigned at iteration $i$ or earlier. By \Cref{clm:scatteringDistTof} we conclude $d_G(v,f(v))\le 3\tau\cdot 2^i< 6\tau\cdot D(v)$.
\end{proof}

\begin{claim}\label{clm:scatteringAssignTime}
	Consider a vertex $v$ such that $2^{i-1}<D(v)\le 2^i$. Then $v$ is assigned either at iteration $i-1$ or $i$.
\end{claim}
\begin{proof}
	Clearly if $v$ remains un-assigned until iteration $i$, it will be assigned during the $i$'th iteration.
	Suppose that $v$ was assigned during iteration $j$. Then $v$ belongs to a cluster $C\in \cC_j$. In particular there is a vertex $u\in C$ such that $D(u)\le 2^j$. As $C$ has diameter at most $2^{j-1}$, it holds that
	$$2^{i-1}<D(v)\le D(u)+d_G(v,u)\le 2^j+2^{j-1}=3\cdot 2^{j-1}~.$$
	$i,j$ are integers, hence $j\ge i-1$.
\end{proof}

		\subsection{Distortion Analysis}\label{subsec:Distoriton}
		In this section we analyze the distortion of the minor induced by the terminal partition created by our algorithm. We have several variables that are defined with respect to the algorithm. Note that all these definitions are for analysis purposes only, and have no impact on the execution of the algorithm.
		
Consider a pair of terminals $t$ and $t'$.
Let $P_{t,t'}=\left\{ t=v_{0},\dots,v_{\gamma}=t'\right\}$ be the shortest path from $t$ to $t'$ in $G$.
We can assume that there are no terminals in $P_{t,t'}$ other than $t,t'$. This is because if we will prove the distortion guarantee for every pair of terminals  $t,t'$ such that $P_{t,t'}\cap K=\{t,t'\}$, then by the triangle inequality the distortion guarantee will hold for all terminal pairs.

\paragraph{Detours:} 
The terminals $t,t'$ are fixed. During the execution of the algorithm, for every terminal $t_j$ we will maintain a \emph{detour} $\mathcal{D}_{t_j}$ (or shortly $\mathcal{D}_{j}$). A detour is a consecutive subinterval $\{a_j,\dots,b_j\}$ of $P_{t,t'}$, where $a_{j}\in \mathcal{D}_{j}$ is the leftmost (i.e., with minimal index) vertex in the detour and $b_j$ is the rightmost.
Initially $\mathcal{D}_{t}=\{t\}$ and $\mathcal{D}_{t'}=\{t'\}$, while for every $t_j\notin\{t,t'\}$, $\mathcal{D}_{j}=\emptyset$.
Every pair of detours $\mathcal{D}_{j},\mathcal{D}_{j'}$ will be disjoint throughout the execution of the algorithm.
	
A vertex $v\in P_{t,t'}$ is \emph{active}  if and only if it does not belong to any detour. 
It will hold that every active vertex is necessarily unassigned (while there might be unassigned vertices which are inactive).
Initially, $t,t'$ are \emph{inactive}, while all the other vertices of $P_{t,t'}$ are active. 
Next we consider the $i$'th iteration of the algorithm, and we will redefine the detours in each iteration as follows. We go over the terminals according to an arbitrary order $\{t_1,\dots,t_k\}$. Consider the terminal $t_{j}$ with detour $\mathcal{D}_{j}=\{a_j,\dots,b_j\}$ (which might be empty).
	If no active vertices are assigned to $t_j$ at the $i$'th iteration we do nothing. Otherwise, let $a'_{j}\in P_{t,t'}$ (resp. $b'_{j}$)
	be the leftmost (resp. rightmost) active vertex that was assigned to $t_{j}$ during the $i$'th iteration.
	Set $a_j$ to be vertex with minimal index between the former $a_j$ and $a'_j$ ($a'_j$ if there was no $a_j$).
	Similarly $b_j$ is the vertex with maximal index between the former $b_j$ and $b'_j$.
	$\mathcal{D}_{j}$ is updated to be $\{a_j,\dots,b_j\}$.
	All the vertices in $\left\{ a_{j},\dots,b_{j}\right\}= \mathcal{D}_{j}$ become inactive.
	Note that a vertex might become inactive while remaining yet unassigned.
	
	Consider an additional detour $\mathcal{D}_{j'}$. Before the updating of $\mathcal{D}_{j}$ at iteration $i$, $\mathcal{D}_{j},\mathcal{D}_{j'}$ were disjoint.
	Then we defined $a'_j,b'_j$ (to be leftmost/rightmost active vertex assigned to $t_j$) and updated $\mathcal{D}_{j}$ to be the interval from $\min\{a_j,a'_j\}$ to $\max\{b_j,b'_j\}$. Note that as $a'_j,b'_j$ were active (assuming they existed) they cannot belong to $\mathcal{D}_{j'}$. Thus after the update, $a_j,b_j$ did not belong to $\mathcal{D}_{j'}$ as well. 
	Nevertheless, it is possible that after the update $\mathcal{D}_{j}$ and $\mathcal{D}_{j'}$ are no longer disjoint.
	However as $\mathcal{D}_{j'}$ is also an interval not containing the endpoints of the new interval $\mathcal{D}_{j}$, the only such possibility is if $\mathcal{D}_{j'}$ is fully contained in $\mathcal{D}_{j}$.\footnote{This is because given two intervals $[x_1,y_1],[x_2,y_2]$ where $x_1,y_1\notin [x_2,y_2]$, by case analysis it must be the case that either the intervals are disjoint, or $[x_2,y_2]\subseteq [x_1,y_1]$.} If indeed $\mathcal{D}_{j'}\subset\mathcal{D}_{j}$, we will nullify $\mathcal{D}_{j'}\leftarrow\emptyset$ and thus maintain the disjointness property (while not changing the active/inactive status of any vertex).

After we nullify all the detours that were contained in $\mathcal{D}_{j}$, we will proceed to treat the next terminals in turn. Once we finish going over all the terminals, we proceed to the $i+1$ iteration.
Eventually, all the vertices are assigned, and hence are inactive. In particular every vertex belong to some detour. 
	In other words, as the detours are disjoint intervals,  all the vertices of $P_{t,t'}$ are partitioned to consecutive disjoint detours $\mathcal{D}_{\ell_1},\dots,\mathcal{D}_{\ell_s}$.

\paragraph{Intervals:} 
For an \emph{interval}
$Q=\left\{ v_{a},\dots,v_{b}\right\} \subseteq P_{t,t'}$, the \emph{internal length} is 
$L(Q)=d_{G}(v_{a},v_{b})$, while the \emph{external length} is $L^{+}(Q)=d_{G}(v_{a-1},v_{b+1})$. \footnote{For ease of notation we will denote $v_{-1}=t$ and $v_{\gamma+1}=t'$.}
We denote by $D(Q)=D(v_a)$ the distance from the leftmost vertex $v_a\in Q$ to its closest terminal.
Set $\cint=\frac17$ (``int'' for interval).
We partition the vertices in $P_{t,t'}$ into consecutive intervals $\mathcal{Q}$, such that for every $Q\in \mathcal{Q}$,
\begin{eqnarray}
L(Q)\le\cint\cdot D(Q)\le L^{+}(Q)~.\label{eq:Intervallength}
\end{eqnarray}
Such a partition could be obtained as follows: Sweep
along the path $P_{t,t'}$ in a greedy manner, after partitioning the prefix $v_{0},\dots,v_{h-1}$,
to construct the next interval $Q$, simply pick
the minimal index $s$ such that $ L^{+}(\left\{ v_{h},\dots,v_{h+s}\right\} )\ge\cint\cdot D(v_{h})$.
By the minimality of $s$, $L(\left\{ v_{h},\dots,v_{h+s}\right\} )\le L^{+}(\left\{ v_{h},\dots,v_{h+s-1}\right\} )\le\cint\cdot D(v_{h})$ (in the case $s=0$, trivially $L(\left\{ v_{h}\right\} )=0\le\cint\cdot D(v_{h})$).
Note that such $s$ could always be found, as $ L^{+}(\left\{ v_{h},\dots,v_{\gamma}=t'\right\} )=d_{G}(v_{h-1},t')\ge d_{G}(v_{h},t')\ge D(v_{h})=D(Q)$.

Consider some interval $Q=\left\{ v_{a},\dots,v_{b}\right\}\in \mathcal{Q}$. 
For every vertex $v\in Q$, by triangle inequality it holds that $D(Q)-L(Q)\le D(v)\le D(Q)+L(Q)$. Therefore,
\begin{eqnarray}
(1-\cint)D(Q)\le D(v)\le(1+\cint)D(Q)~.\label{eq:DvDQ}
\end{eqnarray}
Note that the set $\mathcal{Q}$ of intervals is determined before the execution of the algorithm, and is never changed. In particular, it is independent from the set of detours (which evolves during the execution of the algorithm).\\
For an interval $Q$, we denote by $i_Q$ the first iteration when some vertex $v$ belonging to the interval $Q$ is assigned.
\begin{claim}\label{clm:scattering-iQ}
 All $Q$ vertices are assigned in either iteration $i_Q$ or $i_Q+1$. 
\end{claim} 
\begin{proof}
Let $u\in Q$ be some vertex which is assigned during iteration $i_Q$. Then $u$ belongs to a cluster $C\in \cC_{i_Q}$, containing a vertex $u'\in C$ such that $D(u')\le 2^{i_Q}$. As $C$ has diameter at most $2^{i_Q-1}$, it holds that
$2^{i_Q}\ge D(u')\ge D(u)-d_{G}(u,u')\ge D(u)-2^{i_Q-1}$.
Hence $D(u)\le\frac32\cdot2^{i_Q}$. 
It follows that
\begin{eqnarray}
D(Q)\overset{(\ref{eq:DvDQ})}{\le}\frac{1}{1-\cint}\cdot D(u)\le\frac{3}{2}\cdot\frac{1}{1-\cint}\cdot2^{i_{Q}}~.\label{eq:DQ}
\end{eqnarray}
For every vertex $v\in Q$ it holds that,
\[
D(v)\le D(Q)+L(Q)\overset{(\ref{eq:Intervallength})}{\le} (1+\cint)\cdot D(Q)\overset{(\ref{eq:DQ})}{\le} \frac{3}{2}\cdot\frac{(1+\cint)}{(1-\cint)}\cdot2^{i_{Q}}=2^{i_{Q}+1}~.
\]
Therefore, by \Cref{clm:scatteringAssignTime}, in the $i_Q+1$ iteration, all the (yet unassigned) vertices of  $Q$ will necessarily be assigned.
\end{proof}

\begin{lemma}\label{lem:scatteringIntPart}
Consider an interval $Q\in \mathcal{Q}$. Then the vertices of $Q$ are partitioned into at most $O(\tau^2)$ different detours.
\end{lemma}
\begin{proof}
By definition, by the end of the $i_Q-1$'th iteration all the vertices of $Q$ are unassigned. 
We first consider the case where by the end $i_Q-1$'th iteration some vertex $v\in Q$ is inactive. It holds that $v$ belongs to some detour $\mathcal{D}_j$. As all the vertices of $Q$ are unassigned, necessarily $Q\subset\mathcal{D}_j$. In particular, all the vertices of $Q$ belong to a single detour. This property will not change till the end of the algorithm, thus the lemma follows. 

Next, we consider the case where by the end of the $i_Q-1$'th iteration all the vertices of $Q$ are active.
The algorithm at iteration $i_Q$ creates an $\left(1,\tau,\Delta_{i_Q}\right)$-scattering partition $\mathcal{P}_{i_Q}$. 
The length of $Q$ is bounded by
\begin{eqnarray}\label{eq:LQbound}
L(Q)\overset{(\ref{eq:Intervallength})}{\le}\cint\cdot D(Q)\overset{(\ref{eq:DQ})}{\le}
\cint\cdot\frac{3}{2}\cdot\frac{1}{1-\cint}\cdot2^{i_{Q}}=\frac{1}{4}\cdot2^{i_{Q}}<\Delta_{i_Q}
\end{eqnarray}
Hence $Q$ is partitioned by $\mathcal{P}_{i_Q}$ to $\tau'\le\tau$ clusters $C_1,\dots,C_{\tau'}\in\mathcal{P}_{i_Q}$.
It follows that by the end of the $i_Q$'th iteration, the inactive vertices in $Q$ are partitioned to at most $\tau$ detours. If all the vertices in $Q$ become inactive, then we are done, as the number of detours covering $Q$ can only decrease further in the algorithm (as a result of detour nullification). Hence we will assume that some of $Q$ vertices remain active.

A  \emph{slice} is a maximal sub-interval $S\subseteq Q$ of active vertices. The active vertices in $Q$ are partitioned to at most $\tau+1$ slices $S_1,S_2,\dots,S_{\tau''}$ .\footnote{Actually, as at least one $Q$ vertex remained active, at the beginning of the $i_Q+1$ iteration the inactive vertices of $Q$ partitioned to at most $\tau-1$ detours. Therefore the maximal number of slices is $\tau$.}
By the end of the $i_Q+1$ iteration, according to \Cref{clm:scattering-iQ} all $Q$ vertices  will be assigned, and in particular belong to some detour. 
The algorithm creates a  $\left(\Delta_{i_Q+1},\tau,1\right)$-scattering partition $\mathcal{P}_{i_Q+1}$ of the unassigned vertices. By equation (\ref{eq:LQbound}) the length of every slice $S$ is bounded by $L(S)\le L(Q)\le \frac{1}{4}\cdot2^{i_{Q}}\le\Delta_{i_{Q}+1}$. Therefore
the vertices $S$ intersect at most $\tau$ clusters of $\mathcal{P}_{i_Q+1}$, and thus will be partitioned to at most $\tau$ detours. Some detours might get nullified, however in the worst case, by the end of the $i_Q+1$ iteration, the vertices in $\cup_iS_i$  are partitioned to at most $\tau\cdot (\tau+1)$ detours. In particular all the vertices in $Q$ are partitioned to at most $O(\tau^2)$ detours.
As the number of detours covering $Q$ can only decrease further in the algorithm, the lemma follows. 
\end{proof}

By the end of algorithm, we will \emph{charge} the intervals for the detours. Consider the detour  $\mathcal{D}_{j}=\{a_j,\dots,b_j\}$ of $t_j$. Let $Q_j\in\mathcal{Q}$ be the interval containing $a_j$. We will charge $Q_j$ for the detour $\mathcal{D}_j$.
Denote by $X(Q)$ the number of detours for which the interval $Q$ is charged for. By \Cref{lem:scatteringIntPart}, $X(Q)=O(\tau^2)$ for every interval $Q\in\mathcal{Q}$.

Recall that by the end of the algorithm, all the vertices of $P_{t,t'}$ are partitioned to consecutive disjoint detours $\mathcal{D}_{\ell_1},\dots,\mathcal{D}_{\ell_s}$, where $\mathcal{D}_{\ell_j}=\{a_{\ell_j},\dots,b_{\ell_j}\}$ and $a_{\ell_j},b_{\ell_j}$ belong to the cluster of $t_{\ell_j}$. In particular $t_{\ell_1}=t$ and $t_{\ell_s}=t'$, as each terminal belongs to the cluster of itself. Moreover, for every $j<s$, there is an edge $\{b_{\ell_j},a_{\ell_{j+1}}\}$ in $G$ between the cluster of $t_{\ell_j}$ to that of $t_{\ell_{j+1}}$. Therefore, in the minor induced by the partition there is an edge between $t_{\ell_j}$ to $t_{\ell_{j+1}}$. We conclude
\begin{align*}
d_{M}(t,t')\le\sum_{j=1}^{s-1}d_{G}(t_{\ell_{j}},t_{\ell_{j+1}}) & \le\sum_{j=1}^{s-1}\left[d_{G}(t_{\ell_{j}},a_{\ell_{j}})+d_{G}(a_{\ell_{j}},a_{\ell_{j+1}})+d_{G}(a_{\ell_{j+1}},t_{\ell_{j+1}})\right]\\
& \le\sum_{j=1}^{s-1}d_{G}(a_{\ell_{j}},a_{\ell_{j+1}})+2\sum_{j=1}^{s}d_{G}(t_{\ell_{j}},a_{\ell_{j}})~.
\end{align*}
Note that $\sum_{j=1}^{s-1}d_{G}(a_{\ell_{j}},a_{\ell_{j+1}})\le d_G(t,t')$ as $P_{t,t'}$ is a shortest path. Denote by $Q_{\ell_j}$ the interval containing $a_{\ell_j}$. By \Cref{cor:dvfv},  
$$d_{G}(t_{\ell_{j}},a_{\ell_{j}})=d_{G}(a_{\ell_{j}},f(a_{\ell_{j}}))\le O(\tau)\cdot D(a_{\ell_{j}})\stackrel{(\ref{eq:DvDQ})}{=} O(\tau)\cdot D(Q_{\ell_{j}})\stackrel{(\ref{eq:Intervallength})}{=} O(\tau)\cdot L^{+}(Q_{\ell_{j}})~.$$
By changing the order of summation we get
\[
\sum_{j=1}^{s}d_{G}(t_{\ell_{j}},a_{\ell_{j}})=O(\tau)\cdot\sum_{Q\in\mathcal{Q}}X(Q)\cdot L^{+}(Q)=O(\tau^{3})\cdot\sum_{Q\in\mathcal{Q}}L^{+}(Q)~.
\]
Finally, note that $\sum_{Q\in\mathcal{Q}}L^{+}(Q)\le 2\cdot d_G(t,t')$ as every edge in $P_{t,t'}$ is counted at most twice. We conclude $d_{M}(t,t')\le O(\tau^3)\cdot d_{G}(t,t')$. \Cref{thm:Scattering_Implies_SPR} now follows.

\section{Equivalence between Sparse Covers and Weak Sparse Partitions}\label{sec:covers}
Jia \etal \cite{JLNRS05} proved (implicitly) that sparse covers imply weak sparse partitions. We provide here a formal statement and a proof, for the sake of clarity and completeness.  
\begin{lemma}[\cite{JLNRS05}]\label{lem:coverToPartition}
	Suppose that a graph $G=(V,E,w)$ admits a $(\sigma,\tau,\Delta)$-weak sparse cover $\mathcal{C}$, then $G$ admits a $(\sigma,\tau,\Delta)$-weak sparse partition.
\end{lemma}
\begin{proof}
	Let $\Delta>0$ be some parameter. We create a partition $\mathcal{P}$ as follows: each vertex $v$ joins to an arbitrary cluster $P_C$ which corresponds to a cluster $C\in\mathcal{C}$ covering $v$, that is $B_G(v,\frac\Delta\sigma)\subseteq C$. 
	Note that every cluster $P_C\in\mathcal{P}$ is contained in cluster $C\in \mathcal{C}$. 
	It immediately follows that $\mathcal{P}$ has weak diameter $\Delta$.
	
	Next consider a ball $B=B_G(v,\frac\Delta\sigma)$. Suppose that a cluster $P_C\in \mathcal{P}$ intersects $B$ at a vertex $u\in P_C\cap B$. As $u$ joined $P_C$, it holds that $B_G(u,\frac\Delta\sigma)\subseteq C$, implying  $v\in C$. Thus $B$ intersects only clusters $P_C\in\mathcal{P}$ corresponding to clusters $C\in\mathcal{C}$ containing $v$. We conclude that $B$ intersects at most $\tau$ clusters.
\end{proof}
We discuss the implications of \Cref{lem:coverToPartition} to minor free graphs in \Cref{subsec:MinorFreeCovers}.
Note that if the given partition \Cref{lem:coverToPartition} has strong diameter guarantee, the resulting sparse partition will still have weak diameter guarantee only. Furthermore there are graphs that admit strong sparse covers but do not admit strong sparse partitions with similar parameters. One example will be trees which by \Cref{thm:treeLBStrong} do not admit $(\sigma,\tau)$-strong sparse partition scheme for any constant $\sigma,\tau$, while having $(O(1),O(1))$-strong sparse cover schemes \cite{BLT14,AGMW10}.

Interestingly, sparse partitions also imply sparse covers. Furthermore, unlike the previous direction,  if the partition had strong diameter, so would the cover.
\begin{lemma}\label{lem:partitionToCover}
	Suppose that a graph $G=(V,E,w)$ admits a $(\sigma,\tau,\Delta)$-weak sparse partition $\mathcal{P}$, then $G$ admits a $(\sigma+2,\tau,(1+\frac2\sigma)\Delta)$-weak sparse cover $\mathcal{C}$.
	Furthermore, if $\mathcal{P}$ has a strong diameter guarantee, so will $\mathcal{C}$.	
\end{lemma}
\begin{proof}
	For every cluster $P\in\mathcal{P}$ set $C_P=B_G(P,\frac\Delta\sigma)$. Set $\mathcal{C}=\left\{C_p\mid P\in \mathcal{P}\right\}$ to be the resulting cover.
	First, note that every ball $B_G(v,\frac\Delta\sigma)$ is contained in the cluster $C_P$ corresponding to the cluster $P\in \mathcal{P}$ containing $v$.
	
	Second, we argue that $\mathcal{C}$ has diameter $(1+\frac2\sigma)\Delta$. Consider a cluster $C_P\in \mathcal{C}$ and let $u,v\in C_P$. There are $u',v'\in P$ at distance at most $\frac\Delta\sigma$ from $u$ and $v$ respectively. By triangle inequality, $d_{G}(u,v)\le d_{G}(u,u')+d_{G}(u',v')+d_{G}(v',v)\le(1+\frac{2}{\sigma})\Delta$. Note that if $P$ has strong diameter at most $\Delta$, then as $C_P$ includes the shortest path from $u$ to $u'$ (and $v$ to $v'$) all these inequalities hold w.r.t. $d_{G[C_P]}$ as well.
	We conclude that $\mathcal{C}$ is a sparse cover with the padding parameter at most $\nicefrac{(1+\frac{2}{\sigma})\Delta}{\frac{\Delta}{\sigma}}=\sigma+2$.
	
	Third, we argue that $\mathcal{C}$ is sparse. Consider a vertex $v\in V$. For a cluster $C_P\in\mathcal{C}$ that contains $v$, necessarily $v\in B_G(P,\frac\Delta\sigma)$. In other words, $B_G(v,\frac\Delta\sigma)\cap P\ne \emptyset$. Thus $C_P$ corresponds to a cluster $P\in \mathcal{P}$ intersecting $B_G(v,\frac\Delta\sigma)$. As the number of clusters in $\mathcal{P}$ intersecting $B_G(v,\frac\Delta\sigma)$ is bounded by $\tau$, we conclude that the number of clusters containing $v$ in $\mathcal{C}$ is also bounded by $\tau$.
\end{proof}
\Cref{lem:partitionToCover} phrased in other words: suppose that a graph $G$ admits a $(\sigma,\tau)$-weak/strong sparse partition scheme, these $G$ also admits $(\sigma+2,\tau)$-weak/strong sparse cover scheme. 
Applying \Cref{lem:partitionToCover} on \Cref{thm:ddimStrong,thm:SPDstrong,thm:SPDweak}, we conclude:
\begin{corollary}\label{cor:coversFromPartitions}
	Suppose that a weighted graph $G=(V,E,w)$ has:
	\begin{enumerate}
		\item Doubling dimension $\ddim$, then for every $t\ge 1$, $G$ admits an $(O(t),2^{\nicefrac{\sddim}{t}}\cdot\tilde{O}(\ddim))$-strong sparse cover scheme.
		\item \SPD of depth $\rho$, then $G$ admits a $(O(\rho),O(\rho^2))$-strong sparse cover scheme.
		\item \SPD of depth $\rho$, then $G$ admits a $(10,5\rho)$-weak sparse cover scheme.
	\end{enumerate}
\end{corollary}
Strong sparse covers for doubling graphs were constructed directly by the author in a companion paper \cite{Fil19Padded}. See discussion in \Cref{subsec:contribution}.
See \Cref{fig:Venn_covers} for illustrations of the connections between the different notions of sparse covers and partitions.

\begin{figure}[t]
	\centering{\includegraphics[scale=0.6]{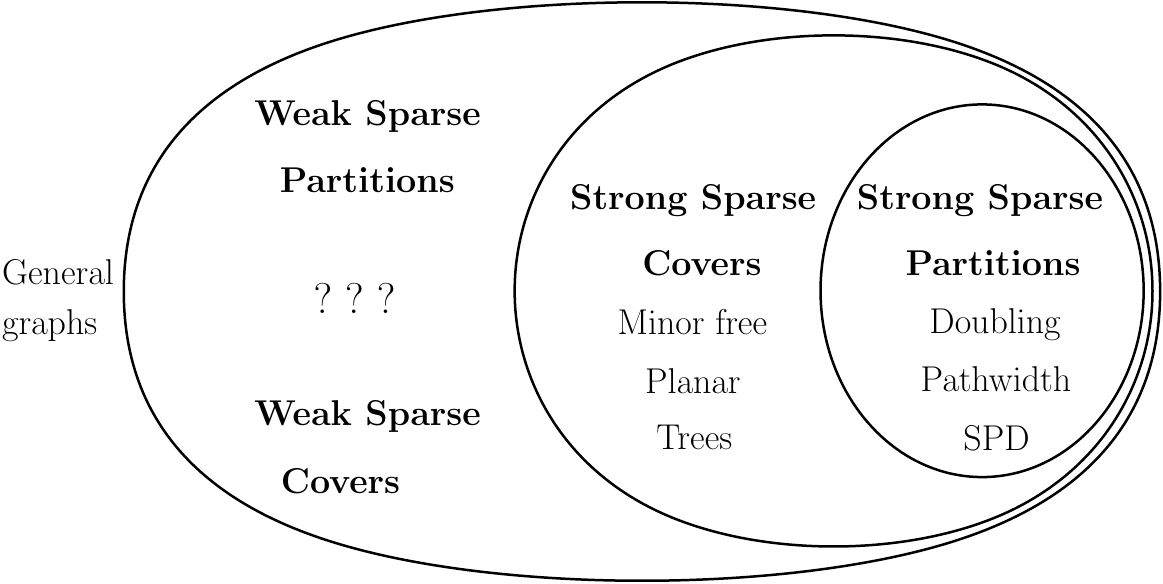}} 
	\caption{\label{fig:Venn_covers}\small \it 
		The Venn diagram demonstrates the containment relations between the set of graphs admitting weak/strong sparse covers/partitions. 
		Graphs with constant doubling dimension or \SPDdepth (or pathwidth) admit strong sparse partitions scheme with constant parameters (\Cref{cor:coversFromPartitions}).
		All graph families excluding a fixed minor admit strong sparse covers  with constant parameters \cite{AGMW10}, while no such strong sparse partitions exist (\Cref{thm:treeLBStrong}).
		The family of general graphs do not admit weak sparse partitions with constant parameters (\Cref{thm:GeneralWeakLB}).
		We currently lack an example of a graph family that admit weak sparse covers but do not admit strong sparse covers. Finding such an example, or alternatively proving that weak sparse covers imply the existence of strong sparse covers with similar parameters remains an open question.}
\end{figure}

\subsection{Implications of \Cref{lem:coverToPartition}: Minor Free Graphs}\label{subsec:MinorFreeCovers}
The graph family which has the most interesting implications due to sparse covers are minor free graphs. As previously mentioned in the literature (see e.g. \cite{AGMW10}), it implicitly follows from the padded decompositions of Klein \etal \cite{KPR93} that $K_{r,r}$ free minor graph admit $(O(r^2),2^{r})$-weak sparse partition scheme. Nevertheless, as this fact was never stated as a theorem, it was overlooked by previous works on \UST/\UTSP. Specifically, covers with worse parameters were used by \cite{HKL06,BLT14}, obtaining solutions with stretch $\Omega(\log^2 n)$ for the these problems.\footnote{Busch \etal \cite{BLT14} argued that their \UTSP construction is deterministic. However the \cite{KPR93} based cover (\Cref{thm:KPR}) is deterministic as well, and hence implies a deterministic construction of \UTSP with better parameters.}
For the sake of clarity and completeness, we provide the explicit statements and a proof sketch.

\begin{theorem}[Implicit from \cite{KPR93,FT03}]\label{thm:KPR}
	Every weighed graph that excludes $K_{r,r}$ as a minor admits an $(O(r^2),2^{r})$-weak sparse cover scheme.	
\end{theorem}
\begin{proof}[Proof Sketch.]
	While we will use the celebrated partition algorithm of \cite{KPR93}, the analysis itself will be similar to the proof of Lemma 3.2 in \cite{KLMN05}.
	Let $\Delta>0$ be some parameter, and let $\Delta'$ be a parameter depending on $\Delta$ to be determined later. Fix $\rt\in V$ to be an arbitrary vertex. For $b\in\{0,\frac12\}$ and $j\ge 0$ set
	\[A_j^b=\left\{v\in V\mid (j-1+b)\Delta'\le d_G(\rt,v)<(j+b)\Delta'\right\}~.\]
	Note that fixing the parameter $b$, we obtain a partition $\mathcal{P}^b=\{A_j^b\}_{j\ge0}$. Thus we created two partitions $\mathcal{P}^0,\mathcal{P}^{\frac12}$. 
	Note also, that every ball of radius smaller than $\frac{\Delta'}{4}$ is fully contained in a cluster in one of the partitions.
	
	Consider each connected component $C$ in each cluster of a each partition $P^b$, and apply the above process again.
	Continuing this way recursively to a total depth of $r$, we obtain $2^r$ partitions of $V$, such that every ball of radius smaller than $\frac{\Delta'}{4}$ if fully contained in a cluster in one of the partitions.
	According to Fakcharoenphol and Talwar \cite{FT03}, there exists a universal constant $c>0$ such that all the created partitions have weak diameter $c\cdot r^2\cdot\Delta'$. Fix $\Delta'=\frac{\Delta}{c\cdot r^2}$. Uniting all the $2^r$ partitions we obtain an $(O(r^2),2^r,\Delta)$-weak sparse cover as required.
\end{proof}

Using \Cref{lem:coverToPartition} combined with \Cref{thm:KPR}, and then applying \Cref{thm:JLNRS05} we conclude:
\begin{corollary}\label{cor:WeakMinor}
	Every graph that excludes $K_{r,r}$ as a minor admits an $(O(r^2),2^{r})$-weak sparse partition scheme.
\end{corollary}
\begin{corollary}\label{cor:minorUTSP}
	Let $G=(V,E,w)$ be an $n$-vertex graph excluding $K_{r,r}$ as a minor. Then there is an efficient algorithm constructing a solution for the \UST problem with stretch $O(4^r\cdot r^2\cdot \log n)$.
\end{corollary}
Note that \Cref{cor:minorUTSP} has a quadratic improvement in the dependence on $n$ compared with previously explicitly stated results.

\section{General Graphs}
Jia \etal \cite{JLNRS05} constructed weak sparse partitions using the sparse covers of Awerbuch and Peleg \cite{AP90}. This approach inherently produces unconnected clusters and therefore provides only weak diameter guarantee.
The first result of this section is an efficient construction of strong sparse partition for general graphs. We construct a single partition that is good w.r.t. all ball sizes simultaneously.
The proof appears in \Cref{subsec:GeneralStrong}.
\begin{theorem}\label{thm:Generalstrong}
	Consider an $n$ vertex weighted graph $G=(V,E,w)$. For every parameter $\Delta>0$ there is a partition $\mathcal{P}$ such that for every $\alpha\ge 1$,  $\mathcal{P}$ is $(8\alpha,O(n^{\nicefrac1\alpha}\cdot\log n),\Delta)$-strong sparse partition.\\ In particular, for every $\alpha\ge 1$, $G$ admits an  $(8\alpha,O(n^{\nicefrac1\alpha}\cdot\log n))$-strong sparse partition scheme.
\end{theorem}
Note that by picking parameter $\alpha=\log n$, we obtain an $\left(O(\log n),O(\log n)\right)$-strong sparse partition scheme. 
\Cref{thm:Generalstrong} is also a generalization of the scattering partitions of \cite{KKN15}, while having a considerably simpler proof.
Using \Cref{thm:Scattering_Implies_SPR} we can induce a solution for the \SPR problem with distortion $\polylog (k)$. While quantitatively better solutions are known, this one has the advantage of being induced by a general framework. Furthermore, the resulting proof is shorter than all the previous ones (and arguably simpler, even though the \texttt{Relaxed Voronoi} algorithm from \cite{Fil19sicomp} is much more elegant).
\begin{corollary}\label{cor:general-SPR}
	Given a weighted graph $G=(V,E,w)$ with a set $K$ of terminals of size $k$, there is an efficient algorithms that returns a solution to the \SPR problem with distortion $\polylog( k)$.
\end{corollary}
\begin{proof}
	According to Krauthgamer, Nguyen, and Zondiner \cite{KNZ14} (Theorem 2.1), $G$ contains a minor with $O(k^4)$ Steiner points, preserving exactly all distances between the terminals.\footnote{This minor is obtained by first deleting all edges which do not lie on a shortest path between two terminals, and then contracting all Steiner vertices of degree $2$.}
	I.e., there is a minor $G'$ of $G$ with $O(k^4)$ vertices containing $K$, such that for every $t,t'\in K$, $d_G(t,t')=d_{G'}(t,t')$.
	Thus we can assume that $|V|=O(k^4)$. By \Cref{thm:Generalstrong}, \Cref{obs:StrongImplyScatt}, and \Cref{obs:beta1}, $G$ and all its induced subgraphs are $\left(1,O(\log^2 k)\right)$-scatterable. The corollary follows by \Cref{thm:Scattering_Implies_SPR}.
\end{proof}

Next, we provide a lower bound on weak sparse partitions for general graphs.
Quantitatively, this lower bound is essentially equivalent to the lower bound for strong sparse partition on trees. However, qualitatively it is different as it requires strong diameter. This lower bound implies that both \Cref{thm:Generalstrong} and the weak sparse partitions of \cite{JLNRS05} are tight up to second order terms. 
The proof appears in \Cref{subsec:GeneralWeakLB}.
\begin{theorem}\label{thm:GeneralWeakLB}
	Suppose that the all $n$-vertex graphs admit a $(\sigma,\tau)$-weak sparse partition scheme. Then 
	$\tau\ge n^{\Omega(\frac{1}{\sigma})}$.
\end{theorem} 

The best scattering partitions we were able to provide are due to \Cref{thm:Generalstrong}. We believe that quadratically better scattering partitions exist. Specifically, that every $n$ vertex weighted graph is $\left(1,O(\log n)\right)$-scatterable, and furthermore that this is tight. See \Cref{con:GeneralScattering} in \Cref{sec:open}.
Unfortunately, we did not provide any lower bound on the parameters of scattering partitions. Nevertheless, we provide some evidence that no better scattering partitions from those conjectured to exist by  \Cref{con:GeneralScattering} can be constructed.
Given a partition $\mathcal{P}$, we say that an edge $\{u,v\}$ is separated if its endpoints belong to different clusters ($P(u)\ne P(v)$). Given a shortest path $\mathcal{I}=\{v_0,v_1,\dots,v_m\}$, denote by $S_{\mathcal{I}}(\mathcal{P})=\left|\{i\in[m]\mid P(v_{i-1})\ne P(v_i)\}\right|$ the number of separated edges along  $\mathcal{I}$. We say that partition $\mathcal{P}$ is $(\sigma,\tau,\Delta)$-\emph{super scattering } if $\mathcal{P}$ is connected, weakly $\Delta$-bounded, and every shortest path $\mathcal{I}$ of length at most $\frac\Delta\sigma$ has at most $\tau$ separated edges $S_{\mathcal{I}}(\mathcal{P})\le \tau$.
We say that a graph $G$ is $(\sigma,\tau)$-\emph{super-scatterable} if for every parameter $\Delta$, $G$ admits a  $(\sigma,\tau,\Delta)$-\emph{super scattering} partition.

It is straightforward to see that a $(\sigma,\tau,\Delta)$-super scattering partition is also a $(\sigma,\tau+1,\Delta)$-scattering. However, the opposite does not always hold. 
Interestingly, some of the scattering partitions created in this paper are actually super-scattering. Specifically our scattering partition for trees, cactus graphs, chordal graphs and Euclidean space are actually super-scattering.
Our strong sparse partition for general graph, doubling graphs, and graph with bounded \SPDdepth are not necessarily super-scattering. 
The proof of the following theorem appears in \Cref{subsec:LBsuper}.
\begin{theorem}\label{thm:generalLBsuper}
	Suppose that all $n$-vertex graphs are $(1,\tau)$-super scatterable. Then $\tau=\Omega(\log n)$.	
\end{theorem}

\subsection{Strong Sparse Partitions for General Graphs: Proof of \Cref{thm:Generalstrong}}\label{subsec:GeneralStrong}

Let $\Delta>0$ be some parameter. 
Our partition will be created using the clustering algorithm of Miller \etal \cite{MPX13} described in \Cref{sec:prem}, with the set of all vertices as centers (that is $N=V$).
For each vertex $t\in N$, we sample a shift $\delta_t$ according to exponential distribution with parameter $\lambda=\frac{\Delta}{4\ln n}$. 
As a result of the execution of \cite{MPX13}, we get a clustering $\mathcal{P}$, where each cluster is connected and associated with some center vertex.

Denote by $\phi$ the event that for all the vertices $t\in V$, $\delta_t\le\frac\Delta2$. By union bound
\[
\Pr\left[\overline{\phi}\right]\le n\cdot \Pr[\delta_t>\frac\Delta2]=n\cdot e^{-\nicefrac{\Delta}{2\lambda}}=\frac 1n~.
\]
Suppose that $\phi$ indeed occurs.	 
Consider some vertex $v\in V$, suppose that $v$ joined the cluster of the center $t\in N$. Thus 
$\delta_{t}-d_{G}(v,t)=f_v(t)\ge f_v(v)=\delta_{v}-d_{G}(v,v)\ge 0$, implying $d_{G}(v,t)\le \delta_{t}\le \frac\Delta2$. 
By \Cref{clm:MPXshortestpath}, for every vertex $v$ in the cluster $C$ of $t$ it holds that $d_{G[C]}(v,t)=d_{G}(v,t)$. 
It follows that (assuming $\phi$) $\mathcal{P}$ has strong diameter $\Delta$.

We drop now any conditioning on $\phi$. Fix some $\alpha\ge 1$.
Let $r_\alpha=\frac{\Delta}{8\alpha}$. Consider an arbitrary vertex $v$ and let $B_{v,\alpha}=B_G(v,r_\alpha)$ be the ball of radius $r_\alpha$ around $v$.
We will bound $Z_{B_{v,\alpha}}$ the number of clusters in $\mathcal{P}$ intersecting $B_{v,\alpha}$.
Consider the set $\{f_v(t)\mid t\in N\}$, and order the values according to decreasing order, that is we denote by $t_{(i)}$ the center corresponding to the $i$'th largest value w.r.t. $f_v$. Specifically $f_v(t_{(1)})\ge f_v(t_{(2)})\ge\dots$~. Note that $t_{(i)}$ is a random variable.
Set $s_\alpha=n^{\nicefrac 1\alpha}\cdot 3\ln n=O(n^{\nicefrac 1\alpha}\cdot \log n)$.
Denote by $\psi_{v,\alpha}$ the event that $f_v(t_{(1)})-f_v(t_{(s_\alpha+1)})> 2r_\alpha$.	
\begin{claim}\label{clm:MPXgap}
	For every $\alpha\ge 1$, $\Pr[\overline{\psi_{v,\alpha}}]=\Pr[f_v(t_{(1)})-f_v(t_{(s_\alpha+1)})\le 2r_\alpha]\le \frac{1}{n^3}$ .
\end{claim}
\begin{proof}
	We will use the law of total probability. Fix the center $t_{(s_\alpha+1)}$ and the set $\mathcal{N}=\{t\in V\mid f_v(t)\ge f_v(t_{(s_\alpha+1)})\}$, note that we did not fix the order of the centers in $\mathcal{N}$. 
	For $t\in\mathcal{N}$, by the memoryless property of exponential distribution it holds that
	\begin{align*}
	& \Pr\left[f_{v}(t)-f_{v}(t_{(s_{\alpha}+1)})\le2r_{\alpha}\mid f_{v}(t)\ge f_{v}(t_{(s_{\alpha}+1)})\right]\\
	& \quad=\Pr\left[\delta_{t}\le2r_{\alpha}+\delta_{t_{(s_{\alpha}+1)}}-d_{G}(v,t_{(s_{\alpha}+1)})+d_{G}(v,t)\mid\delta_{t}\ge\delta_{t_{(s_{\alpha}+1)}}-d_{G}(v,t_{(s_{\alpha}+1)})+d_{G}(v,t)\right]\\
	& \quad\le\Pr[\delta_{t}\le2r_{\alpha}]=1-e^{-\nicefrac{2r_{\alpha}}{\lambda}}=1-n^{-\nicefrac{1}{\alpha}}~.
	\end{align*}
	As all the centers in $\mathcal{N}$ are independent, we conclude that 
	\begin{align*}
	\Pr[f_{v}(t_{(1)})-f_{v}(t_{(s_{\alpha}+1)})\le2r_{\alpha}] & =\Pr[\forall t\in\mathcal{N},~f_{v}(t)-f_{v}(t_{(s_{\alpha}+1)})\le2r_{\alpha}]\\
	& \le\left(1-n^{-\nicefrac{1}{\alpha}}\right)^{s_{\alpha}}<e^{-3\cdot\ln n}=\frac{1}{n^{3}}~.
	\end{align*}
\end{proof}

Suppose that  $\psi_{v,\alpha}$ indeed occurs. 
By \Cref{clm:MPXintersectionProperty}, $B_{v,\alpha}$ will not intersect clusters centered in vertices $t$ for which $f_v(t_{(1)})-f_v(t_{(t)})> 2r_\alpha$. Therefore, at most $s_\alpha$ clusters might intersect $B_{v,\alpha}$.
We conclude that if all the events $\phi,\{\psi_{v,\alpha}\}_{v\in V}$ occur, then $\mathcal{P}$ is $(\frac{\Delta}{r_\alpha},s_\alpha,\Delta)=(8\alpha,n^{\nicefrac 1\alpha}\cdot 3\ln,\Delta)$-strong sparse partition.
Set $A=\{1,1+\frac{1}{\log n},1+\frac{2}{\log n},\dots,\log n\}$ be the arithmetic progression from $1$ to $\log n$ with difference between every pair of consecutive terms being $\frac{1}{\log n}$. 
By union bound, the probability that for every $\alpha\in A$,  $\mathcal{P}$ is $(\frac{\Delta}{r_\alpha},s_\alpha,\Delta)=(8\alpha,n^{\nicefrac 1\alpha}\cdot 3\ln,\Delta)$-strong sparse partition is at least $1-\Pr[\overline{\phi}]-\sum_{\alpha\in A}\sum_{v\in V}\Pr[\overline{\psi_{v,\alpha}}]\ge1-\frac{2}{n}$.
Thus we sample such a partition with high probability. We argue that the theorem holds for this partition.
Indeed consider some $\alpha\ge 1$.
If $\alpha\le \log n$, set $\alpha'$ to be a number in $A$ such that $\alpha'\le\alpha\le\alpha'+\frac{1}{\log n}$. Else, if  $\alpha> \log n$, set $\alpha'= \log n$. 
Every ball of radius $\frac{\Delta}{r_\alpha}$ is contained in a ball of radius $\frac{\Delta}{r_{\alpha'}}$, while the number of clusters it intersects is bounded by $s_{\alpha'}=n^{\nicefrac{1}{\alpha'}}\cdot3\ln n\le n^{\nicefrac{1}{\alpha}}\cdot6\ln n=O(n^{\nicefrac{1}{\alpha}}\cdot\log n)$ as required.
The theorem follows. \QED

\subsection{Lower Bound on Weak Sparse Partitions: Proof of \Cref{thm:GeneralWeakLB}}\label{subsec:GeneralWeakLB}

Let $d>1,c>0$ be some constants such that for every $n$ there is an unweighted $n$-vertex graph $G_n$ with maximal degree $d$ and vertex expansion at least $c$. That is, every subset $A$ of $G_n$ of cardinality at most $\nicefrac{n}{2}$ has at least $c\cdot |A|$ neighbors.
Note that as the maximal degree is bounded by $d$, a ball of radius $r$ contains at most $1+d+d^2+\dots+d^{r}<d^{r+1}$ vertices.

Fix $\Delta=\frac{\log_d\frac{n}{2d}}{1+\frac1\sigma}$. Let $\mathcal{P}$ be a partition with weak diameter $\Delta$. Consider a cluster $A\in\mathcal{P}$. Set $A_r=B_G(A,r)$ to be a ball of radius $r$ around $A$. As $A$ is contained in a ball of radius $\Delta$, $|A_{\frac\Delta\sigma}|\le d^{(1+\frac1\sigma)\cdot \Delta+1}= \frac{n}{2}$. Thus we can use the expansion property in $\frac\Delta\sigma$ consecutive steps, and conclude
\[
|A_{\frac{\Delta}{\sigma}}|\ge(1+c)\cdot|A_{\frac{\Delta}{\sigma}-1}|\ge\dots\ge (1+c)^{\frac{\Delta}{\sigma}}\cdot|A_{0}|=(1+c)^{\frac{\Delta}{\sigma}}\cdot|A|~.
\]

Pick uniformly at random vertex $v\in V$. For $A\in\mathcal{P}$, let $X_A$ be an indicator for the event that $B_G(v,\frac\Delta\sigma)$ the ball of radius $\frac\Delta\sigma$ around $v$ intersects $A$.  It holds that
\[
\mathbb{E}_{v}\left[Z_{B_{G}(v,\frac{\Delta}{\sigma})}(\mathcal{P})\right]=\sum_{A\in\mathcal{P}}\Pr\left[X_{A}\right]=\sum_{A\in\mathcal{P}}\Pr\left[v\in A_{\frac{\Delta}{\sigma}}\right]=\sum_{A\in\mathcal{P}}\frac{\left|A_{\frac{\Delta}{\sigma}}\right|}{n}=(1+c)^{\frac{\Delta}{\sigma}}\cdot\sum_{A\in\mathcal{P}}\frac{\left|A\right|}{n}=(1+c)^{\frac{\Delta}{\sigma}}~.
\] 
By averaging arguments we conclude that $
\tau\ge(1+c)^{\frac{\Delta}{\sigma}}=(1+c)^{\frac{1}{\sigma}\cdot\frac{\log_{d}\frac{n}{2d}}{1+\frac{1}{\sigma}}}=n^{\Omega(\frac{1}{\sigma})}
$.\QED
\subsection{Lower Bound for Super Scattering Partition: Proof of \Cref{thm:generalLBsuper}}\label{subsec:LBsuper}
Let $G=(V,E)$ be the $d$-dimensional hypercube. That is, the vertices $V=\{0,1\}^d$ are corresponding to $0,1$ strings of length $d$, and two vertices are adjacent if their corresponding strings differ in a single bit. 
Set $\Delta=k=\frac {d}{10}$. Note that every cluster of diameter $\Delta$ includes at most ${d\choose k}\le(\frac {de}{k})^k=(10e)^{\frac {d}{10}}<2^{\frac d2}$
vertices, as it is included in a ball of radius $\Delta$.
According to the edge isoperimetric inequality for the hypercube, given such a cluster $A$, the number of outgoing edges from this cluster is at least $|A|(d-\log|A|)\ge \frac d2\cdot|A|$. In particular, given a partition $\mathcal{P}$ into $k$-bounded clusters, the total number of inter-cluster edges is at least $\frac12\sum_{A\in\mathcal{P}} \frac d2\cdot|A|=2^d\cdot\frac d4$.

Let $x_0$ be a uniformly chosen random vertex. Pick $k$ bits $b_1,b_2,\dots,b_k$ from $[d]$ without repetitions where the order is important (we pick uniformly among the $\frac{n!}{(n-k)!}$ possible choices).
Set $x_i$ to be equal to $x_0$ where we flip the bits $b_1,\dots,b_i$. 
Let $X_i$ be an indicator for the event that $x_i$ and $x_{i+1}$ belong to different clusters ($P(x_i)\ne P(x_{i+1})$). Set $X=\sum_{j=1}^k X_i$. Note that for every $i$, $\{x_{i-1},x_i\}$ is a uniformly random edge. Therefore $\Pr[X_i]\ge\frac{2^d\cdot\frac d4}{\frac12\cdot2^d\cdot d}=\frac12$ equals to the ratio between inter cluster edges to all edges.  
We conclude that $\mathbb{E}[X]\ge\frac k2$. By averaging arguments, there exists a shortest path $\mathcal{I}$ of length $k$ with $S_{\mathcal{I}}(\mathcal{P})\ge\frac k2=\Omega(d)$ separated edges. As the number of vertices in the hypercube is $2^d$, the theorem follows.\QED

\section{Trees}
In this section we deal with the most basic graph family of trees. We show that trees admit scattering and weak sparse partitions with constant parameters, while no strong sparse partitions with constant parameters are possible. Thus we have a sharp contrast between the different partition types.

The first result of the section is the scattering partition. The proof appears in \Cref{subsec:trees-scattering}. Using \Cref{thm:Scattering_Implies_SPR} one can induce a (previously known \cite{G01,FKT19}) solution for the \SPR problem on trees with distortion $O(1)$.
\begin{theorem}\label{thm:tressScat}
	Every tree $T=(V,E,w)$ is $\left(2,3\right)$-scatterable.
\end{theorem}

Next we turn our attention to weak sparse partitions. The fact that trees admit $\left(O(1),O(1)\right)$-weak sparse partition scheme actually follows by \Cref{lem:coverToPartition}, as are there sparse covers constant parameters for trees \cite{AGMW10,BLT14}. Nevertheless, as trees are the most basic graph family, we present a direct and elegant proof, obtaining better constants and understanding. 
Our constructions of scattering and weak sparse partitions are similar. \Cref{fig:TreePartition} illustrates both partitions on the full binary tree.
The proof appears in \Cref{subsec:trees-weakSparse}.

\begin{theorem}\label{thm:treeWeak}
	Every tree $T=(V,E,w)$ admits a $\left(4,3\right)$-weak sparse partition scheme.
\end{theorem}

\begin{figure}[t]
	\includegraphics[width=.495\textwidth]{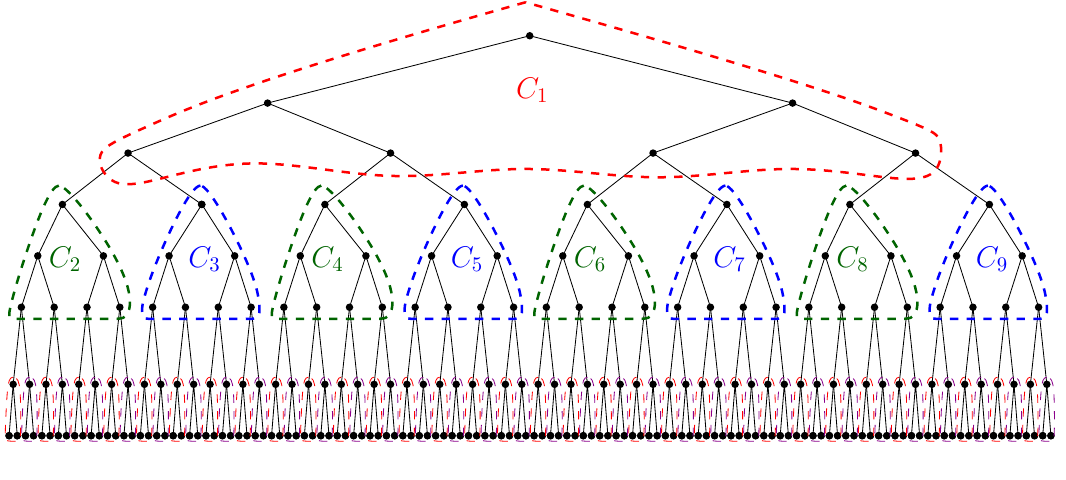}
	\includegraphics[width=.495\textwidth]{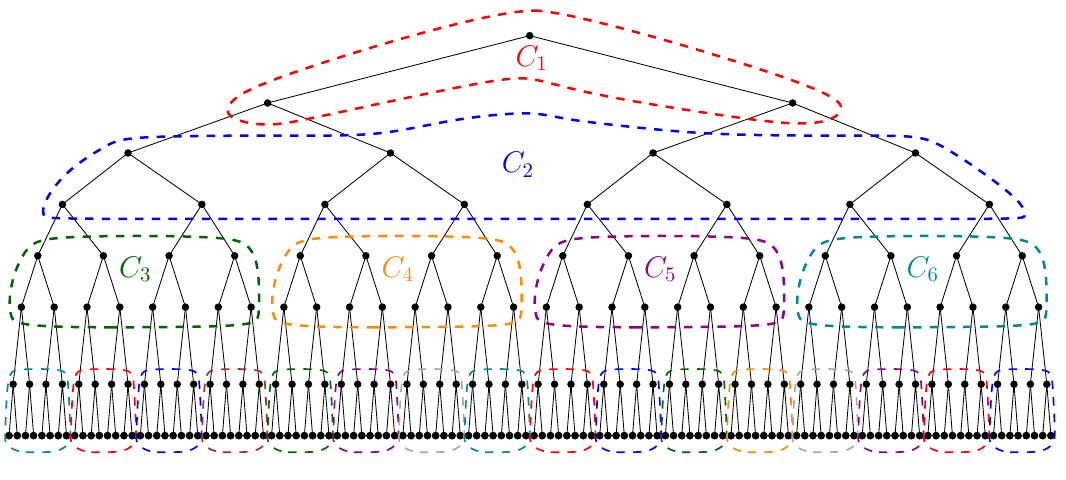}
	\caption{\label{fig:TreePartition}\small \it 
		The figure illustrates two partitions of the unweighted full binary tree of depth $7$. On the left side is the $(2,3,6)$-scattering partition from \Cref{thm:tressScat} for parameter $\Delta=6$, while on the right side is the $(4,3,8)$-weak sparse partition from \Cref{thm:treeWeak} for parameter $\Delta=8$.}
\end{figure}

Finally we turn to the impossibility result. This lower bound is especially interesting as many graph families contain trees, and thus the lower bound applies for them as well. Specifically it also holds for general graphs, all minor free graphs, chordal graphs, etc.
Using the strong sparse partitions for general graphs (\Cref{thm:Generalstrong}), we conclude that the parameters of the lower bound are tight up to second order terms.
We conclude that for strong sparse partition, trees are essentially as hard as general graphs.	
The proof appears in \Cref{subsec:trees-LB_StrongSparse}.
\begin{theorem}\label{thm:treeLBStrong}
	Suppose that all trees with at most $n$ vertices admit a $(\sigma,\tau)$-strong sparse partition scheme. Then $\tau\ge \frac13\cdot  n^{\frac{2}{\sigma+1}}$.
\end{theorem}
We illustrate the lower bound via two different parameter choices.
\begin{corollary}\label{cor:treeLB}
	For every large enough parameter $n>1$, there are $n$ vertex trees $T_1$, $T_2$ such that,
	\begin{OneLiners}
		\item $T_1$ do not admit $\left(\frac{\log n}{\log\log n},\log n\right)$-strong sparse partition scheme.
		\item $T_2$ do not admit $\left(\sqrt{\log n},2^{\sqrt{\log n}}\right)$-strong sparse partition scheme.		 
	\end{OneLiners}
\end{corollary}

\subsection{Scattering Partitions for Trees: Proof of \Cref{thm:tressScat}}\label{subsec:trees-scattering}
Let $\Delta>0$ be some parameter.
Let $\rt$ be an arbitrary root vertex. Partition the graph according to distances from $\rt$. Specifically for every $i\ge 0$ set $\mathcal{R}_i=\{u\mid d_T(u,\rt)\in[i\cdot \frac\Delta2,(i+1)\cdot\frac\Delta2)\}$. The clusters will be the connected components of each induced subgraph $G[\mathcal{R}_i]$. See \Cref{fig:TreePartition} for illustration.

Consider some cluster $C\subseteq \mathcal{R}_i$, and let $v_C\in C$ be the closest vertex in $C$ to $\rt$. Note that for every vertex $u\in C$, its shortest path towards $\rt$ goes through $v_C$.  In particular, $d_T(v_C,u)=d_T(\rt,u)-d_T(\rt,v_C)<(i+1)\cdot\frac\Delta2-i\cdot \frac\Delta2=\frac\Delta2$.
It follows that our partition has strong diameter $\Delta$. 

Consider a path $\mathcal{I}=\{v_0,\dots,v_m\}$ of length at most $\frac\Delta2$. Let $v_j\in\mathcal{I}$ be the closest vertex to $\rt$ among $\mathcal{I}$ vertices. Let $i\ge0$ be the index such that $v_j\in\mathcal{R}_i$, and denote by $C_j$ the cluster of $v_j$.
As the distance of all the vertices in $\mathcal{I}$ from $v_j$ is at most $\frac\Delta2$, all $\mathcal{I}$ vertices belong to either $\mathcal{R}_i$ or $\mathcal{R}_{i+1}$.
Let $a\in[1,j]$ (resp. $b\in[j,m]$) be the minimal (resp. maximal) index such that $v_a\in \mathcal{R}_i$ (resp. $v_b\in \mathcal{R}_i$). 
All the vertices $v_a,\dots,v_b$ are belonging to $\mathcal{R}_i$ and lie in a single connected component. From the other hand, all the vertices $v_0,\dots,v_{a-1},v_{b+1},\dots,v_m$ belong to $\mathcal{R}_{i+1}$ and lie in at most two connected components.\footnote{If $a=0$ then  $v_0,\dots,v_{a-1}$ is an empty set. Similarly for $b=m$.}
We conclude that the vertices of $\mathcal{I}$ intersect at most $3$ different clusters.\QED

\subsection{Weak Sparse Partitions for Trees: Proof of \Cref{thm:treeWeak}}\label{subsec:trees-weakSparse}
Let $\Delta>0$ be some parameter.
Let $\rt$ be an arbitrary root vertex. Partition the graph according to distances from $\rt$. Specifically for every $i\ge 0$, set $\mathcal{R}_i=\{u\mid d_T(u,\rt)\in[i\cdot \frac\Delta4,(i+1)\cdot\frac\Delta4)\}$.  Fix $\mathcal{R}_i$, we say that $u,v\in\mathcal{R}_i$ are equivalent $u\sim_i v$ if and only if $u$ and $v$ have a common ancestor in $\mathcal{R}_{i-1}\cup \mathcal{R}_{i}$. It is straightforward to verify that $\sim_i $ is an equivalence relation. Our partition $\mathcal{P}=\bigcup_{i\ge0}\nicefrac{\mathcal{R}_i}{\sim_i}$ will simply be all the equivalence classes for all indices $i$.  See \Cref{fig:TreePartition} for illustration.

First we argue that our partition has weak diameter $\Delta$. Consider a pair of vertices $u,v\in\mathcal{R}_i$ which are clustered together. As $u\sim_i v$, they have a common ancestor $z\in\mathcal{R}_{i-1}\cup \mathcal{R}_{i}$. Thus 
\begin{align*}
d_{T}(u,v) & \le d_{T}(u,z)+d_{T}(z,v)=\left(d_{T}(\rt,u)-d_{T}(\rt,z)\right)+\left(d_{T}(\rt,v)-d_{T}(\rt,z)\right)\\
& <2\cdot\left((i+1)\cdot\frac{\Delta}{4}-\left(i-1\right)\cdot\frac{\Delta}{4}\right)=\Delta~.
\end{align*}
Consider a pair of vertices such that $u,v\in\mathcal{R}_i$ but $u\not\sim_i v$. Their least common ancestor is at distance greater than $\frac\Delta4$ from both $u,v$, thus $d_T(u,v)>\frac\Delta2$. It follows that every pair of vertices in $\mathcal{R}_i$ at distance at most $\frac\Delta2$ are necessarily equivalent, and hence belong to the same cluster.
Let $B=B_T(v,r)$ be some ball with radius at most $r\le \frac{\Delta}{4}$. The maximal pairwise distance of a pair of vertices in $B$ is at most $\frac\Delta2$. It follows that for every index $i$, $B$ can intersect at most a single cluster from $\mathcal{R}_i$. Suppose w.l.o.g. that $v\in \mathcal{R}_i$, then the vertices of $B$ are contained in $\mathcal{R}_{i-1}\cup\mathcal{R}_i\cup\mathcal{R}_{i+1}$. It follows that $B$ intersects at most $3$ clusters.\QED

\subsection{Lower Bound on Strong Sparse Partition for Trees: Proof of \Cref{thm:treeLBStrong}}\label{subsec:trees-LB_StrongSparse}
Assume for contradiction that all $n$-vertex trees admit a $(\sigma,\tau)$-strong sparse partition scheme where $\tau<\frac13\cdot n^{\frac{2}{\sigma+1}}$.
Set $d=\tau$ and $D=\frac{\sigma+1}{2}$.
Let $T$ be a tree rooted at $\rt$, of depth $D$, such that the root has $d+1$ children, and all other vertices at distance less than $D$ from $\rt$ have exactly $d$ children. The number of vertices is bounded by
\begin{align*}
1+(d+1)+(d+1)\cdot d+\dots+(d+1)\cdot d^{D-1}  =1+(d+1)\cdot\frac{d^{D}-1}{d-1}<3\cdot d^{D}\le n~.
\end{align*}
Fix $\Delta=2D-1$, and let $\mathcal{P}$ be a $(\sigma,\tau,\Delta)=(2D-1,d,\Delta)$-strong sparse partition. 
Consider the cluster $C\in\mathcal{P}$ containing the root $\rt$. If at most one child of $\rt$ belongs to $C$, then the ball $B_T(\rt,1)$ of radius $1=\frac{\Delta}{2D-1}$ intersects $d+1>\tau$ clusters (as every  other child must belong to a different cluster, because $\mathcal{P}$ is connected).
Otherwise, there must be a non-leaf vertex $v\in C$ such that none of $v$'s children belong to $C$, as otherwise $C$ will contain a path of length $2D>\Delta$. All the children of $v$ belong to a different clusters. Therefore the ball $B_T(v,1)$ intersects $d+1$ different clusters, a contradiction.
\QED

\section{Doubling Graphs}\label{sec:ddim}
In this section we construct strong sparse partitions for graphs with bounded doubling dimension. 
Similarly to general graphs, fixing a diameter parameter $\Delta$, we construct a single partition which is simultaneously good for all ball sizes.
\begin{theorem}\label{thm:ddimStrong}
	Let $G=(V,E,w)$ be a graph with doubling dimension $\ddim$. 
	For every parameter $\Delta>0$, there is an efficiently computable partition $\mathcal{P}$, such that for every integer $\alpha\ge 1$,  $\mathcal{P}$ is $\left(58\alpha,2^{\nicefrac{\tddim}{\alpha}}\cdot\tilde{O}(\ddim),\Delta\right)$-strong sparse partition. In particular, for every $\alpha\ge 1$, $G$ admits a $(58\alpha,2^{\nicefrac{\tddim}{\alpha}}\cdot\tilde{O}(\ddim))$-strong sparse partition scheme.
\end{theorem}
In \Cref{thm:GeneralWeakLB} we prove that if all $n$ vertex graphs admit $(\sigma,\tau)$-weak sparse partition scheme, then $\tau\ge n^{\Theta(\frac1\sigma)}$. As every $n$ vertex graph has doubling dimension at most $\log n$, it implies that \Cref{thm:ddimStrong} is tight up to second order terms. In particular, even if one replaced the diameter guarantee in \Cref{thm:ddimStrong} from strong to weak, it will still be tight up to second order terms.

An exponential improvement for \UST and \UTSP problems in the dependence on the dimension is a direct corollary from \Cref{thm:JLNRS05}.
\begin{corollary}\label{cor:ddimUTSP}
	Let $G=(V,E,w)$ be an $n$-vertex graph with doubling dimension $\ddim$. Then there is an efficient algorithm solving the \UST problem with stretch $O(\ddim^3\cdot\log n)$.
\end{corollary}
Unfortunately, the doubling dimension of an induced graph $G[A]$ might be considerably larger than the doubling dimension of the original graph $G$. Thus this partition as is cannot be plugged in directly into \Cref{thm:Scattering_Implies_SPR}. We leave the construction of a solution for the \SPR problem on doubling graphs with distortion $\poly(\ddim)$ as an open problem for future work. 

The rest of this section is devoted to the proof of \Cref{thm:ddimStrong}.

\begin{proof}[Proof of \Cref{thm:ddimStrong}]
	Let $\Delta>0$ be some parameter. We will construct partition which is only $O(\Delta)$ bounded.
	As this works for every $\Delta$, eventually one can re-adjust the parameters accordingly. Let $N\subseteq X$ be a $\Delta$-net.
	Our partition will be created using the modified clustering algorithm of Miller \etal \cite{MPX13} described in \Cref{subsec:MPX}. 
	Set $\ctop=4$. 
	For each net point $t\in N$, we sample a shift $\delta_t$ according to \emph{betailed}  exponential distribution with parameters $\left(\lambda=\frac{\Delta}{\sddim},\lambda_{\top}=\ctop\cdot\Delta\right)$ .\footnote{The author is not aware of previous appearance of this distribution in the literature. The name betailed is inspired by the term ``beheaded'', as we cut off the tail of the distribution.} This distribution is the same as an exponential distribution with parameter $\lambda$, where the only difference is that any value above $\lambda_{\top}$ collapses to $\lambda_{\top}$. Formally, a variable distributed according to a betailed exponential distribution with parameters $(\lambda,\lambda_\top)$ get value in $[0,\lambda_\top]$, where the density function in $[0,\lambda_\top)$ is $f_{\lambda,\lambda_\top}(x)=\frac1\lambda e^{-\frac x\lambda}$, and the probability to get the value $\lambda_\top$ is $e^{-\frac{\lambda_\top}{\lambda}}$.
	
	As a result of the execution of \cite{MPX13} algorithm we get a clustering $\mathcal{P}$, where each cluster is connected and associated with some net point from $N$.  
	Consider some vertex $v\in V$. There is a net point $t_v\in N$ at distance at most $\Delta$ from $v$. Suppose that $v$ joined the cluster of the net point $t\in N$. Hence 
	$\delta_{t}-d_{G}(v,t)=f_v(t)\ge f_v(t_v)=\delta_{t_{v}}-d_{G}(v,t_{v})\ge-\Delta$. Therefore
	$d_{G}(v,t)\le\lambda_\top+\Delta=(\ctop+1)\cdot\Delta$. 
	By \Cref{clm:MPXshortestpath}, for every vertex $v$ in the cluster $C$ of $t$ it holds that $d_{G[C]}(v,t)=d_{G}(v,t)$. 
	It follows that $\mathcal{P}$ has strong diameter $2\cdot (\ctop+1)\cdot\Delta=10\Delta$.
	
	Fix some $\alpha\ge 1$, and let $r_{\alpha}=\frac{\ln 2}{\alpha}\cdot\Delta$. Consider an arbitrary vertex $v$ and let $B_{v,\alpha}=B_G(v,\frac{r_{\alpha}}{2})$ be the ball of radius $\frac{r_{\alpha}}{2}$ around $v$.
	Denote by $N_v\subseteq N$ the set of net points at distance at most $\left(\ctop+2\right)\cdot\Delta$ from $v$.
	For every $t\notin N_v$, $f_v(t_v)-f_v(t)\ge \left(0-d_G(v,t_v)\right)-\left(\lambda_{\top}-d_G(v,t)\right)>\Delta$. By \Cref{clm:MPXintersectionProperty}, the cluster of a net point $t\notin N_v$ will not intersect $B_G(v,\frac\Delta2)$ and in particular $B_{v,\alpha}$.
	Using the packing property (\Cref{lem:doubling_packing}), 
	\begin{eqnarray}
	|N_{v}|\le\left(\nicefrac{2\cdot\left(\ctop+2\right)\cdot\Delta}{\Delta}\right)^{\sddim}\le (3\cdot \ctop)^{\sddim}~.\label{eq:ddim-Nv}
	\end{eqnarray}
	
	\sloppy We will bound $Z_{B_{v,\alpha}}$ the number of clusters in $\mathcal{P}$ intersecting $B_{v,\alpha}$. 
	For the sake of analysis, the sample of $\delta_t$, the shift of the net point $t\in N$, will be computed in two steps: First sample a variable $\tilde{\delta}_t$ according to exponential distributed with parameter $\lambda$, secondly set $\delta_t=\min\{\tilde{\delta}_t,\ctop\cdot\Delta\}$.
	Similarly to the definition of $f_v(t)$, set $\tilde{f}_v(t)=\tilde{\delta}_t-d_G(v,t)$. 
	
	Consider the set $\{\tilde{f}_v(t)\mid t\in N_v\}$, and order the values according to decreasing order, that is we denote by $\tilde{t}_{(i)}\in N_v$ the net point corresponding to the $i$'th largest value w.r.t. $\tilde{f}_v$. Specifically $\tilde{f}_v(\tilde{t}_{(1)})\ge \tilde{f}_v(\tilde{t}_{(2)})\ge\dots$ . Note that $\tilde{t}_{(i)}$ is a random variable.
	Fix $s=\left\lceil 4\ln\left(4e\cdot\ddim^{2}\cdot\left(24\cdot\ctop\cdot\ddim\right)^{\sddim}\right)\right\rceil =\tilde{O}\left(\ddim\right)$. 
	Set $m_{\alpha}=2 s\cdot 2^{\nicefrac{\tddim}{\alpha}}$.
	\begin{claim}\label{clm:Ddim-fTildeGap}
		$\Pr[\tilde{f}_v(\tilde{t}_{(s)})-\tilde{f}_v(\tilde{t}_{(m_{\alpha}+1)})\le  r_{\alpha}]\le e^{-\frac{s}{4}}$ .
	\end{claim}
	\begin{proof}
		We will use the law of total probability.
		Fix the net point $\tilde{t}_{(m_{\alpha}+1)}\in N_v$. Let $\mathcal{N}=\{t\in N_v\mid \tilde{f}_v(t)\ge \tilde{f}_v(\tilde{t}_{(m_{\alpha}+1)})\}$, note that $|\mathcal{N}|=m_{\alpha}$ and that we did not fixed the inner order of the points in $\mathcal{N}$. 
		For $t\in\mathcal{N}$ denote by $X_t$ the event that $\tilde{f}_v(t)-\tilde{f}_v(\tilde{t}_{(m_{\alpha}+1)})> r_{\alpha}$. By the memoryless property of exponential distribution, 
		\[
		\Pr\left[\tilde{f}_{v}(t)-\tilde{f}_{v}(\tilde{t}_{(m_{\alpha}+1)})>r_{\alpha}\mid\tilde{f}_{v}(t)\ge\tilde{f}_{v}(\tilde{t}_{(m_{\alpha}+1)})\right]\ge\Pr[\tilde{\delta}_{\alpha}>r_{\alpha}]=e^{-\nicefrac{r_{\alpha}}{\lambda}}=2^{-\nicefrac{\tddim}{\alpha}}~.
		\]
		Set $X=\sum_{t\in\mathcal{N}}X_t$. Then $\E[X]\ge m_{\alpha}\cdot 2^{-\nicefrac{\tddim}{\alpha}}=2 s$. As all $\{\tilde{f}_v(t)\}_{t\in\mathcal{N}}$ are independent, by Chernoff inequality it holds that
		\[
		\Pr[X\le\frac12\cdot\E[X]]\le e^{-\frac{\E[X]}{8}}
		< e^{-\frac{s}{4}}~.
		\]
		
		In case $X>\frac12\cdot\E[X]\ge s$, it will imply that for at least $s$ net points $t\in N_v$, $\tilde{f}_v(t)>r_\alpha+\tilde{f}_v((\tilde{t}_{(m_{\alpha}+1)})$. In particular, $\tilde{f}_v(\tilde{t}_{(s)})-\tilde{f}_v(\tilde{t}_{(m_{\alpha}+1)})> r_\alpha$. The claim now follows.
	\end{proof}
	
	Denote by $\phi_v$ the event that for at least $s$ different net points $t\in N_v$, it holds that $\delta_t=\lambda_\top$. As all $\{\delta_t\}_{t\in N_v}$ are independent,
	\begin{eqnarray}
	\Pr\left[\phi_v\right]\le{|N_{v}| \choose s}\cdot\left(e^{-\frac{\lambda_{\top}}{\lambda}}\right)^{s}\overset{(\ref{eq:ddim-Nv})}{\le}\left((3\cdot\ctop)^{\sddim}\cdot e^{-\ctop\cdot\sddim}\right)^{s}\le e^{-s\cdot\sddim}~.	\label{eq:ddim-phi-v}
	\end{eqnarray}
	Next, we incorporate the truncations into the analysis. 
	Denote by $\tilde{\psi}_{v,\alpha}$ the event that 	 $\tilde{f}_v(\tilde{t}_{(s)})-\tilde{f}_v(\tilde{t}_{(m_{\alpha}+1)})\le r_{\alpha}$. By \Cref{clm:Ddim-fTildeGap}, $\Pr[\tilde{\psi}_{v,\alpha}]\le e^{-\nicefrac{s}{4}}$.
	Similarly, denote by $\psi_{v,\alpha}$ the event that $f_v(t_{(1)})-f_v(t_{(m_{\alpha}+1)})\le  r_{\alpha}$, where $t_{(i)}$ is the net point having the $i$'th largest value w.r.t. $f_v$ in $N_v$.
	We argue that $\psi_{v,\alpha}\subseteq \tilde{\psi}_{v,\alpha}\cup\phi_v$. It will be enough to show that if both $\tilde{\psi}_{v,\alpha}$ and $\phi_v$ did not occur, neither did $\psi_{v,\alpha}$ ($\overline{\tilde{\psi}_{v,\alpha}}\wedge\overline{\phi_v}\Rightarrow\overline{\psi_{v,\alpha}}$). Indeed, consider the case that both $\tilde{\psi}_{v,\alpha}$ and $\phi_v$ did not occur. As $\phi_v$ did not occur, there is an index $i\in[1,s]$ such that $\delta_{\tilde{t}_{(i)}}=\tilde{\delta}_{\tilde{t}_{(i)}}$. For every $j\ge m_{\alpha}+1$ it holds that 
	$$f_v(t_{(1)})\ge f_v(\tilde{t}_{(i)})=\tilde{f}_v(\tilde{t}_{(i)})\ge \tilde{f}_v(\tilde{t}_{(s)})> r_{\alpha}+ \tilde{f}_v(\tilde{t}_{(m_{\alpha}+1)})\ge  r_{\alpha}+ \tilde{f}_v(\tilde{t}_{(j)})\ge   r_{\alpha}+ f_v(\tilde{t}_{(j)})~.$$
	Thus for all but at most $m_{\alpha}$ net points $t\in N_v$, it holds that $f_v(t_{(1)})-f_v(t)> r_{\alpha}$. Hence $\psi_{v,\alpha}$ did not occur.
	By \Cref{clm:MPXintersectionProperty}, it follows that $B_{v,\alpha}$ did not intersect the clusters of such net points $t\in N_v$. Therefore the number of clusters intersecting $B_{v,\alpha}$ is bounded by $Z_{B_{v,\alpha}}(\mathcal{P})\le m_{\alpha}$.
	By union bound
	\[
	\Pr_{\mathcal{P}}\left[Z_{B_{v,\alpha}}(\mathcal{P})>m_{\alpha}\right]\le\Pr_{\mathcal{P}}[\psi_{v,\alpha}]\le\Pr_{\mathcal{P}}[\phi_v\vee\tilde{\psi}_{v,\alpha}]\overset{(\ref{eq:ddim-phi-v})+\Cref{clm:Ddim-fTildeGap}}{\le}
	e^{-s\cdot\sddim}+e^{-\nicefrac{s}{4}}\le 2\cdot e^{-\nicefrac{s}{4}}~.
	\]
	
	Set $A=\{1,1+\frac{1}{\sddim},1+\frac{2}{\sddim},\dots,\ddim\}$ to be the arithmetic progression from $1$ to $\ddim$ with difference between every pair of consecutive terms being $\frac{1}{\sddim}$.
	Set $\Psi_v=\bigcup_{\alpha\in A}\psi_{v,\alpha}$ the event that $\psi_{v,\alpha}$ holds for some $\alpha\in A$.  
	By union bound $\Pr\left(\Psi_v\right)=\sum_{\alpha\in A}\Pr\left(\psi_{v,\alpha}\right)\le 2\cdot \ddim^2\cdot e^{-\nicefrac{s}{4}}$.
	
	Our next goal is to show that there is a positive probability for the event that simultaneously for every vertex $v$ and all parameters $\alpha\ge 1$, the ball $B_G(v,\frac{r_{\alpha}}{4})$ intersects at most $m_{\alpha}$ clusters.
	Set $\eps=\frac{\Delta}{4\cdot\sddim}$, and let $\hat{N}\subseteq  X$ be an $\eps$-net. 
	Set $\mathcal{A}=\{\Psi_v\}_{v\in\hat{N}}$. For a net point $v\in\hat{N}$ denote $\Gamma_v=\left\{u\in\hat{N}\mid d_G(v,u)\le 3\cdot\ctop\cdot\Delta\right\}$. For every net point $u\in \hat{N}\setminus\Gamma_v$, $N_v$ and $N_u$ are disjoint, thus $\Psi_v$ and $\Psi_u$ are independent. 
	Using the packing property (\Cref{lem:doubling_packing}), $\Psi_v$ is dependent with at most
	\[
	|\Gamma_{v}|\le\left(\frac{6\cdot\ctop\cdot\Delta}{\eps}\right)^{\sddim}=\left(24\cdot\ctop\cdot\ddim\right)^{\sddim}
	\]
	other events from $\mathcal{A}$.
	We will use the constructive version of the Lov\'asz Local Lemma by Moser and Tardos \cite{MT10}.
	\begin{lemma}[Constructive Lov\'asz Local Lemma] \label{lem:lovasz}
		Let $\mathcal{P}$ be a finite set of mutually independent random variables in a probability space. Let $\mathcal{A}$ be a set of events determined by these variables. For $A\in\mathcal{A}$ let $\Gamma(A)$ be a subset of $\mathcal{A}$ satisfying that $A$ is independent from the collection of events $\mathcal{A}\setminus(\{A\}\cup\Gamma(A))$.
		If there exist an assignment of reals $x:\mathcal{A}\rightarrow(0,1)$  such that
		\[\forall A\in \mathcal{A}~:~~\Pr[A]\le x(A)\cdot\prod_{B\in\Gamma(A)}(1-x(B))~,\]
		then there exists an assignment to the variables $\mathcal{P}$ not violating any of the events in $\mathcal{A}$. Moreover, there is an algorithm that finds such an assignment in expected time  $\sum_{A\in\mathcal{A}}\frac{x(A)}{1-x(A)}\cdot\poly \left(|\mathcal{A}|+|\mathcal{P}|\right)$.
	\end{lemma}
	
	For $v\in \hat{N}$ set $x(\Psi_v)=e\cdot \Pr[\Psi_v]$. For every $v\in \hat{N}$ it holds that
	\begin{align*}
	x(\Psi_{v})\Pi_{u\in\Gamma_{v}}(1-x(\Psi_{u})) & \ge\Pr\left[\Psi_{v}\right]\cdot e\cdot\left(1-2e\cdot\ddim^2\cdot e^{-\nicefrac{s}{4}}\right)^{|\Gamma_{v}|}\\
	& \ge\Pr\left[\Psi_{v}\right]\cdot e\cdot\left(1-\frac{1}{2}\cdot\left(24\cdot\ctop\cdot\ddim\right)^{-\sddim}\right)^{\left(24\cdot\ctop\cdot \sddim\right)^{\tddim}} \ge\Pr\left[\Psi_{v}\right]~.
	\end{align*}
	By \Cref{lem:lovasz} we can efficiently choose $\{\delta_v\}_{v\in N}$ such that none of the events $\{\Psi_v\}_{v\in\hat{N}}$ occur. Consider the partition $\mathcal{P}$ obtained by choosing this values $\{\delta_v\}_{v\in N}$. 
	Fix some parameter $\alpha>1$ and some vertex $u$.
	There is a net point $v\in \hat{N}$ at distance at most $\eps\le \frac {r_{\alpha}}{4}$  from $u$.
	If $\alpha\le \ddim$, pick $\alpha'\in A$ such that $\alpha'\le \alpha\le \alpha'+\frac{1}{\sddim}$. Else, if  $\alpha> \ddim$, set $\alpha'= \ddim$. 
	The ball of radius $\frac{r_{\alpha}}{4}$ around $u$ is contained in a ball of radius $\frac{r_{\alpha'}}{2}$ around $v$.
	As $\psi_{v,\alpha'}\subseteq \Psi_{v}$ did not occur, it holds that $Z_{B_{v,\alpha'}}(\mathcal{P})\le m_{\alpha'}$. Therefore 
	$Z_{B_G(u,\frac {r_{\alpha}}{4})}(\mathcal{P})\le Z_{B_{v,\alpha'}}(\mathcal{P})\le
	m_{\alpha'}=2s\cdot2^{\nicefrac{\tddim}{\alpha'}}\le4s\cdot2^{\nicefrac{\tddim}{\alpha}}=2^{\nicefrac{\tddim}{\alpha}}\cdot\tilde{O}(\ddim)$.
	The padding padding parameter is 
	$\frac{2\cdot(\ctop+1)\cdot\Delta}{r_{\alpha}/4}=\frac{40}{\ln2}\cdot \alpha\le58\alpha$.
	We conclude that for every $\alpha\ge 1$, $\mathcal{P}$ is  $\left(58\alpha,2^{\nicefrac{\tddim}{\alpha}}\cdot\tilde{O}(\ddim),10\Delta\right)$-strong sparse partition.
	The theorem follows.
\end{proof}

\section{Euclidean Space}\label{sec:Euclidean}
Consider the $d$ dimensional Euclidean space $(\R^d,\|\cdot\|_2)$. Partitions of this space are well studied. We will study the Euclidean space from the lenses of sparse and scattering partitions. Weak sparse partitions are defined in the natural way.
A cluster $C$ is connected if for every pair of vectors $\vec{x},\vec{y}\in C$, there is a continuous function $f:[0,1]\rightarrow \mathbb{R}^d$ such that $f(0)=\vec{x}$, $f(1)=\vec{y}$, and the entire image of $f$ is inside $C$.
The shortest path between two vectors $\vec{x},\vec{y}$ is simply the interval between them $\{\vec{x}+t\cdot(\vec{y}-\vec{x})\mid t\in[0,1]\}$. 
A partition $\mathcal{P}$ is $(\sigma,\tau,\Delta)$-scattering if each cluster is connected and has weak diameter at most $\Delta$, and for every pair of points at distance at most $\frac\Delta\sigma$, the interval between them intersects at most $\tau$ clusters.

Interestingly, we show that $(\R^d,\|\cdot\|_2)$ has scattering partitions with significantly better parameters compared to the weak sparse partitions it admits. Specifically, for constant padding parameter $\sigma=\Omega(1)$ there are scattering partitions where the number of intersections $\tau=O(d)$ is linear in $d$, while every weak sparse partition with such padding parameter $\sigma$ will have exponential number of intersections $\tau=2^{\Omega(d)}$, that is not better than general space with doubling dimension $d$.

\begin{theorem}\label{thm:L2scattering}
	The Euclidean space $(\R^d,\|\cdot\|_2)$ is $(1,2d)$-scatterable.
\end{theorem}

\begin{theorem}\label{thm:L2LB}
	Suppose that the space $(\R^d,\|\cdot\|_2)$ admits a $(\sigma,\tau)$-weak sparse partition scheme. Then $\tau>(1+\frac{1}{2\sigma})^d$.
\end{theorem}
Another way to represent the parameters in \Cref{thm:L2LB} is  $\sigma>\frac{1}{2}\cdot\frac{1}{\tau^{\frac{1}{d}}-1}>\frac{1}{2}\cdot\frac{1}{1+\frac{2\ln\tau}{d}-1}=\frac{d}{4\ln\tau}$, where the second inequality holds as $\tau^{\frac{1}{d}}=e^{\frac{\ln\tau}{d}}<1+\frac{2\ln\tau}{d}$ (using that $e^{x}<1+2x$ for $x\in(0,1)$).
Note that in order to create a partition with at most polynomially many intersections, the padding parameter must be essentially linear in $d$.

\subsection{Scattering Partitions for Euclidean Space: Proof of \Cref{thm:L2scattering}}
We will prove that $(\R^d,\|\cdot\|_2)$ admits an $(1,2d,\sqrt{d})$-scattering partition.
By scaling, this will imply the general theorem.
Define $\mathcal{P}$ to be the natural partition according to axis parallel hyperplanes. That is for every $\vec{a}=(a_1,\dots,a_d)\in\Z^d$, we will have the cluster $C_{\vec{a}}=\{\vec{x}=(x_1,\dots,x_d)~\mid~\forall i,~a_i\le x_i<a_i+1\}$. It is straightforward that this partition has diameter $\sqrt{d}$.
Note that in this partition, an arbitrarily small ball centered in a vector with integer coordinates will intersect $2^d$ different clusters. Nevertheless, intervals will intersect only a small number of clusters.

Consider a pair of vectors $\vec{a},\vec{b}$ at distance at most $\|\vec{a}-\vec{b}\|_2\le \sqrt{d}$. Denote by $\vec{c}=\vec{b}-\vec{a}$ their difference, and thus the shortest path between them is the interval $I=\{\vec{a}+t\cdot\vec{c}~\mid~t\in[0,1]\}$.
Denote by $h_i^n=\{\vec{x}=(x_1,\dots,x_d)~\mid~x_i=n\}$ the $i$-axis parallel hyperplane at height $n$. Denote by $\mathcal{H}_i=\{h_i^n\mid n\in\Z\}$ the set of all $i$-axis parallel hyperplanes at integer heights.
We say that the interval $I$ crosses $h_i^n$ if there are two vectors $\vec{x},\vec{y}\in I$ s.t. $(\vec{x})_i\le n<(\vec{y})_i$.
Set $X_i$ to be the number of hyperplanes in $H_i$ which are crossed by $I$.
The number of clusters intersecting  $I$ is bounded by the number of axis parallel hyperplanes $I$ crosses, plus the initial cluster containing $\vec{a}$. That is $Z_I(\mathcal{P})\le 1+\sum_{i=1}^d X_i$.

Set $\vec{c}=(c_1,\dots,c_d)$. As the projection of $I$ on the $i$`th coordinate is an interval of length $|c_i|$, it holds that $X_i\le \left\lceil |c_{i}|\right\rceil$. The total number of intersections is bounded by 
\[
Z_I(\mathcal{P})\le 1+\sum_{i=1}^d X_i\le1+\sum_{i=1}^{d}\left\lceil |c_{i}|\right\rceil < d+1+\left\Vert \vec{c}\right\Vert _{1}\overset{(*)}{\le} d+1+\sqrt{d}\cdot\left\Vert \vec{b}-\vec{a}\right\Vert _{2}=2d+1~,
\]
where the inequality $^{(*)}$ holds as for every vector $\vec{x}\in\R^d$, $\|\vec{x}\|_1\le\sqrt{d}\cdot \|\vec{x}\|_1$. As the number of intersection is an integer, we conclude $Z_I(\mathcal{P})\le 2d$. The theorem follows.\QED
\begin{remark}
	The analysis of the partition above is tight. Indeed, consider the interval between the vectors
	$\vec{a}=(-\eps,-2\eps,-3\eps,\dots,-d\cdot\eps)$ and $\vec{b}=(\eps,1+\eps,1+\eps,\dots,1+\eps)$ for small enough $\eps$. The number of clusters the interval $[\vec{a},\vec{b}]$ intersects is exactly $2d$.
\end{remark}

\subsection{Lower Bound on Euclidean Weak Sparse Partitions: Proof of \Cref{thm:L2LB}}
The proof is in $(\R^d,\|\cdot\|_2)$ with Lebesgue measure, we will omit unnecessary notation.
Fix an arbitrary $\Delta>0$ and set $c=4\sigma$. Instead of partitioning the entire space, we will use the assumption to produce a $(\sigma,\tau,\Delta)$-weak sparse partition $\mathcal{P}$ of $B(\vec{0},c\cdot\Delta)$, the origin centered ball with radius $c\cdot\Delta$.
Let $r=\frac{\Delta}{\sigma}$.
Given a set $A$, denote by $A_r=\{x+y\mid x\in A,~y\in \R^d,~\|y\|_2\le r\}$ the set of all points at distance at most $r$ from $A$.
Note that $A_r$ is the Minkowski sum of $A$ and $B(\vec{0},r)$. By the Brunn-Minkowski Theorem (see e.g. \cite{Gar02} Corollary 5.3.), it holds that 
$$\vol^{\frac1d}(A_r)\ge \vol^{\frac1d}(A)+\vol^{\frac1d}(B(\vec{0},r))~.$$
Consider a cluster $C\in \mathcal{P}$.
As the diameter of $C$ is bounded by $\Delta$, $C$ is contained in a ball of radius $\Delta$. Therefore
$\vol(B(\vec{0},r))=(\frac r\Delta)^{d}\cdot \vol(B(\vec{0},\Delta))\ge (\frac 1\sigma)^{d}\cdot \vol(C)$.
We conclude, that for every $C\in\mathcal{P}$ it holds that
\begin{align*}
\vol(C_{r}) & \ge\left(\vol^{\frac{1}{d}}(C)+\vol^{\frac{1}{d}}(B(\vec{0},r))\right)^{d}\\
& \ge\left(\vol^{\frac{1}{d}}(C)+\frac{1}{\sigma}\cdot\vol^{\frac{1}{d}}(C)\right)^{d}=\left(1+\frac{1}{\sigma}\right)^{d}\cdot\vol(C)~.
\end{align*}

\sloppy The proof of the theorem will be concluded using the probabilistic method.
Set $\mathcal{P}'=\{C\cap {B(\vec{0},(c-1)\Delta)}\mid C\in\mathcal{P}\}$ to be the set consisting of $\mathcal{P}$ clusters limited to $B(\vec{0},(c-1)\Delta)$.
Sample a point $x$ from the ball $B(\vec{0},c\cdot\Delta)$ uniformly at random. Denote by $X$ the number of clusters in $\mathcal{P}'$ the ball $B(x,r)$ intersects. In other words, $X=\left|\left\{C\in\mathcal{P}'\mid x\in C_r \right\}\right|$.
For $C\in\mathcal{P}'$, denote by $X_C$ an indicator for the event $x\in C_r$. Then $\Pr[X_C=1]=\frac{\vol\left(C_r\right)}{\vol\left(B(\vec{0},c\cdot\Delta)\right)}$.
Thus
\begin{align*}
\mathbb{E}\left[X\right] & =\sum_{C\in\mathcal{P}'}\mathbb{E}\left[X_{C}\right]=\sum_{C\in\mathcal{P}'}\frac{\vol\left(C_{r}\right)}{\vol\left(B(\vec{0},c\cdot\Delta)\right)}\ge\left(1+\frac{1}{\sigma}\right)^{d}\cdot\sum_{C\in\mathcal{P}'}\frac{\vol\left(C\right)}{\vol\left(B(\vec{0},c\cdot\Delta)\right)}\\
& =\left(1+\frac{1}{\sigma}\right)^{d}\cdot\frac{\vol\left(B(\vec{0},(c-1)\cdot\Delta\right)}{\vol\left(B(\vec{0},c\cdot\Delta\right)}=\left(1+\frac{1}{\sigma}\right)^{d}\cdot\left(1-\frac{1}{c}\right)^{d}>\left(1+\frac{1}{2\sigma}\right)^{d}~.
\end{align*}
By an averaging argument, there is a point $x\in B(\vec{0},c\cdot\Delta)$ such that $B(x,r)$ intersects more than $\left(1+\frac{1}{2\sigma}\right)^{d}$ clusters from $\mathcal{P}'$. In particular, this ball intersects more than $\left(1+\frac{1}{2\sigma}\right)^{d}$ clusters from $\mathcal{P}$, as required. \QED

\section{Graphs with bounded \SPD}\label{sec:SPD}
In the preliminaries (\Cref{sec:prem}), and in \Cref{foot:SPD} we gave a recursive definition for \SPD. Below we provide an alternative and more applicable definition.
\begin{definition}[Shortest Path Decomposition\label{def:SPD} (\SPD) \cite{AFGN23}]
	Given a weighted graph $G=(V,E,w)$, an \SPD of depth $\rho$ is a pair
	$\left\{ \mathcal{X},\mathcal{Q}\right\} $, where $\mathcal{X}$ is a
	collection $\mathcal{X}_{1},\dots,\mathcal{X}_{\rho}$ of partial partitions of
	$V$,\footnote{That is for every $X\in \mathcal{X}_i$, $X\subseteq V$, and
		for every different subsets $X,X'\in\mathcal{X}_{i}$, $X\cap
		X'=\emptyset$.} and $\mathcal{Q}$ is a collection of sets of paths
	$\mathcal{Q}_{1},\dots,\mathcal{Q}_{\rho}$, where $\mathcal{X}_1=\{V\}$,
	$\mathcal{X}_\rho=\mathcal{Q}_\rho$, and the following properties hold:
	\begin{enumerate}
		\item For every $1\leq i\leq \rho$ and every subset $X\in\mathcal{X}_{i}$,
		there exists a unique path $Q_X\in\mathcal{Q}_{i}$ such that $Q_{X}$
		is a shortest path in $G[X]$.
		\item For every $2\leq i\leq \rho$, $\mathcal{X}_i$ consists of all
		connected components of $G[X\setminus Q_{X}]$ over all $X\in\mathcal{X}_{i-1}$.
	\end{enumerate}
\end{definition}
We say that all the paths in $\mathcal{Q}_i$ are deleted at level $i$. In this section we provide two different constructions of sparse partition for graphs with \SPDdepth $\rho$. 
First we construct an $(O(\rho),O(\rho^2))$-strong sparse partition scheme. This partition is constructed using the \cite{MPX13} clustering algorithm. Interestingly, unlike all previous executions of \cite{MPX13}, the shifts $\delta_v$ are chosen deterministically. The proof appears in \Cref{subsec:SPDstrong}.
\begin{theorem}\label{thm:SPDstrong}
	Let $G=(V,E,w)$ be a graph with \SPD $\left\{ \mathcal{X},\mathcal{Q}\right\} $ of depth $\rho$. Then $G$ admits a $(O(\rho),O(\rho^2))$-strong sparse partition scheme.
\end{theorem}
A graph with bounded \SPDdepth might have an induced subgraph with much larger \SPDdepth, thus we cannot apply \Cref{thm:Scattering_Implies_SPR}.
Nevertheless, graphs with pathwidth $\rho$ have \SPDdepth $\rho+1$, and also all their subgraphs have pathwidth $\rho$.  
Using \Cref{obs:StrongImplyScatt} and \Cref{thm:Scattering_Implies_SPR} we conclude:
\begin{corollary}\label{cor:pathwidth-SPR}
	Given a weighted graph $G=(V,E,w)$ with pathwidth $\rho$, and a set $K$ of terminals, there is an efficient algorithms that returns a solution to the \SPR problem with distortion $O(\rho^9)$.
\end{corollary}
For weak diameter guarantee we were able to construct sparse partition with improved parameters. 
The partition algorithm is inspired by previous constructions of padded decompositions for pathwidth graphs \cite{AGGNT19,KK17}.
\begin{theorem}\label{thm:SPDweak}
	Let $G=(V,E,w)$ be a graph with \SPD $\left\{ \mathcal{X},\mathcal{Q}\right\} $ of depth $\rho$. Then $G$ admits a $(8,5\rho)$-weak sparse partition scheme.
\end{theorem}
By \Cref{thm:JLNRS05}, a solution to the \UST (and \UTSP) follows.
\begin{corollary}\label{cor:SPD-UTSP}
	Let $G$ be a graph with \SPDdepth $\rho$, then there is an efficient algorithm constructing a solution for the \UST problem with stretch $O(\rho\cdot \log n)$.
\end{corollary}
The following lemma will be useful for both our constructions:
\begin{lemma}\label{lem:SPDpathNet}
	$G=(V,E,w)$ is a weighted graph with a shortest path $P$, $N\subseteq P$ is a set of vertices such that for every $u,z\in N$, $d_G(u,z)> \eps$. Then for every vertex $v\in V$ and radius $r\ge 0$, it holds that $\left|B_G(v,r)\cap N\right|\le \frac{2r}{\eps}$.
\end{lemma}
\begin{proof}
	Denote $P=v_0,v_1,\dots,v_m$. Let $v_i$ be the vertex with minimal index among $B_G(v,r)\cap N$ (if there is no such $v_i$, the lemma holds trivially).
	We order the vertices in the intersection $B_G(v,r)\cap N=\{v_i=v_{q_0},v_{q_0},\dots,v_{q_s}\}$ w.r.t. to the order induced by $P$. 
	By triangle inequality, $d_G(v_{q_0},v_{q_s})\le d_G(v_{q_0},v)+d_G(v,v_{q_s})\le 2r$. From the other hand, as $P$ is a shortest path
	$d_{G}(v_{q_{0}},v_{q_{s}})=\sum_{j=0}^{s-1}d_{G}(v_{q_{j}},v_{q_{j+1}})> s\cdot\epsilon$. We conclude that $s<\frac{2r}{\eps}$. The lemma now follows.
\end{proof}
\subsection{Strong Diameter for \SPD Graphs: Proof of \Cref{thm:SPDstrong}}\label{subsec:SPDstrong}
Let $\Delta>0$ be some parameter. 
Set $\eps=\frac1\rho$. For every $i$, set $\alpha_i=(1+2\eps)^{\rho+1-i}\cdot\Delta$ and  $\beta_i=\eps\cdot\alpha_{i+1}$.
We will create a partition with strong diameter $2\alpha_1<2e^2\cdot\Delta$. Afterwards, the parameters could be readjusted accordingly.
For a path $Q\in\mathcal{Q}_{i}$, let $N_Q\subseteq P$ be a $\beta_i$-net. For every vertex $v\in N_Q$ for $Q\in\mathcal{Q}_i$, set $\delta_{v}=\alpha_i$.
We execute the clustering algorithm of Miller \etal \cite{MPX13} as described in \Cref{subsec:MPX}, with the set $\bigcup_{i\in[1,\rho],Q\in\mathcal{Q}_i}N_Q$ as centers. As a result we get a partition $\mathcal{P}$. 

We first argue that our partition has strong diameter $2\alpha_1$.
Indeed, consider a vertex $v$ that belongs to a path $Q\in\mathcal{Q}_i$. There is a center $t_Q\in N_Q$ such that $d_G(v,t_Q)\le\beta_i$ and $\delta_{t_Q}=\alpha_i$. Suppose that $v$ joined a cluster centered in $t_v$. Then $\delta_{t_v}-d_G(v,t_v)=f_v(t_v)\ge f_v(t_Q)=\delta_{t_Q}-d_G(v,t_Q)$. Thus,
\[
d_G(v,t_v)\le \delta_{t_v}-\delta_{t_Q}+d_G(v,t_Q)\le \delta_{t_v}-\alpha_i+\beta_i< \delta_{t_v}\le \alpha_{1}~.
\]
By \Cref{clm:MPXshortestpath}, for every vertex $v$ in the cluster $C$ of $t_v$ it holds that $d_{G[C]}(v,t_v)=d_{G}(v,t_v)$. 
It follows that $\mathcal{P}$ has strong diameter $2\alpha_{1}<2e^2\cdot\Delta$.

We say that vertices $u,v$ were separated at level $i$ if they  belong to the same cluster of the \SPD at level $i$ (in $\mathcal{X}_i$), and either belong to different clusters of the \SPD at the $i+1$ level of the hierarchy, or if one of them is deleted during the $i$'th level (that is belong to a path $Q\in\mathcal{Q}_i$).
Set $r=\frac\eps2\cdot\Delta$. Consider some vertex $v$, and let $B=B_G(v,r)$.
\begin{claim}\label{clm:SPDseparated}
	Consider a center $t$.
	Suppose that at level $i$ of the SPD some vertex $u$ on the shortest path from $v$ to $t$ was deleted, while $t$ was not. Then no vertex from $B$ will join the cluster centered at $t$.
\end{claim}

\begin{wrapfigure}{r}{0.14\textwidth}
	\begin{center}
		\vspace{-20pt}
		\includegraphics[width=0.9\textwidth]{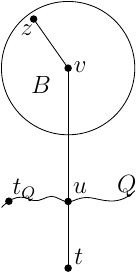}
		\vspace{-5pt}
	\end{center}
	\vspace{-15pt}
\end{wrapfigure}
\emph{Proof.} As $t$ was not deleted on or before level $i$, $\delta_t\le\alpha_{i+1}$. 
The vertex $u$ is laying on a shortest path $Q\in\mathcal{Q}_i$. There is a center $t_Q\in N_Q$ at distance at most $\beta_i$ from $u$, while $\delta_{t_Q}=\alpha_i$. 
Let $z\in B$, by triangle inequality $d_{G}(z,t_{Q})\le d_{G}(z,v)+d_{G}(v,u)+d_{G}(u,t_{Q})\le r+d_{G}(v,u)+\beta_i$. Similarly, $d_{G}(z,t)\ge d_{G}(v,t)-d_{G}(z,v)\ge d_{G}(v,u)-r$, thus $d_{G}(v,u)\le d_{G}(z,t)+r$.
See figure on the right for illustration.\\
As $\alpha_{i+1}+\beta_{i}+2r=(1+\eps)\cdot\alpha_{i+1}+\eps\cdot\Delta<(1+2\eps)\cdot\alpha_{i+1}=\alpha_{i}$, and thus $\alpha_{i+1}<\alpha_{i}-\beta_{i}-2r$, it holds that
\begin{align*}
f_{z}(t_{Q}) & =\delta_{t_{Q}}-d_{G}(z,t_{Q})\ge\alpha_{i}-\left(r+d_{G}(v,u)+\beta_{i}\right)\\
& \ge\alpha_{i}-2r-\beta_{i}-d_{G}(z,t)>\alpha_{i+1}-d_{G}(z,t)\ge\delta_{t}-d_{G}(z,t)=f_{z}(t)~.
\end{align*}
The claim now follows.
\QED

Let $i_B$ be the first level in which some vertex $u\in B$ is deleted.
$u$ belongs to a path $Q_B\in \mathcal{Q}_{i_B}$. There is a center $t_B$ at distance at most $\beta_{i_B}$ from $u$. By triangle inequality, for every vertex $z\in B$, $d_G(z,t_B)\le d_G(z,u)+d_G(u,t_B)\le \beta_{i_B}+2r$.
Furthermore 
\begin{eqnarray}\label{eq:SPD1}
f_z(t_B)=\delta_{t_B}-d_G(z,t_B)\ge \alpha_{i_B}-\beta_{i_B}-2r>\alpha_{i_B+1}~.
\end{eqnarray}
It follows that no vertex of $B$ will join the cluster of a center $t$ that belongs to a path deleted at levels $i_B+1$ and higher.
We conclude that the ball $B$ can be covered only by clusters with centers from the exact $i_B$ paths in the clusters containing $B$ in each level.
\begin{claim}
	Let $Q\in \mathcal{Q}_j$ be the path deleted from the component $X$ containing $B$ at level $j\le i_B$. Then $B$ intersects at most $2\rho+4$ clusters with centers in $N_Q$.
\end{claim}
\begin{proof}
	Consider some center $t\in N_{P}$. If $d_{G[C]}(v,t)>d_{G}(v,t)$ then some vertex on the shortest path from $v$ to $t$ was deleted in an earlier level. By \Cref{clm:SPDseparated}, no vertex from $B$ will join the cluster of $t$.
	Denote by $N'_Q\subseteq N_Q$ the subset of centers for which $d_{G[X]}(v,t)=d_{G}(v,t)$. The centers in $N'_Q$ lie on a shortest path (in $G[C]$), and all are at distance greater than $\beta_j$ apart  (w.r.t. $d_{G[X]}$).
	
	Let $t\in N'_Q$ be a center such that some vertex $z\in B$ joined the cluster of $t$. Then $d_{G}(v,t)\le \alpha_j$, as otherwise $f_z(t)=\delta_{t}-d_G(z,t)\le \alpha_j-d_G(v,t)+r<r$. By equation (\ref{eq:SPD1}), $z$ will not join the cluster of $t$, a contradiction. 	
	By \Cref{lem:SPDpathNet}, there are at most $\frac{2\alpha_j}{\beta_j}=\frac2\eps+4=2\rho+4$ vertices in $N'_Q$.  The claim follows.
\end{proof}
To wrap up, the vertices of $B$ can join only clusters with centers laying on $i_B\le \rho$ paths. In each such path, there are at most $2\rho+4$ centers, to the cluster of which a vertex from $B$ might join.
As the maximal diameter is $e^2\cdot\Delta$, the padding parameter is $\frac{e^2\cdot\Delta}{r}=\frac{2\cdot e^2}{\eps}=2\cdot e^2\cdot \rho$. We conclude that $\mathcal{Q}$ is a $\left(2\cdot e^2\cdot \rho,(2\rho+4)\rho,e^2\cdot\Delta\right)$-strong sparse partition. Thus $G$ admits a  $\left(O(\rho),O(\rho^2)\right)$-strong sparse partition scheme as required.
\QED

\subsection{Weak Diameter for \SPD Graphs: Proof of \Cref{thm:SPDweak}}\label{subsec:SPDweak}
Let $\Delta>0$ be some parameter. 
The clustering will be done in two phases and described in \Cref{alg:weakSPD}.

\begin{algorithm}[]
	\caption{\texttt{\SPD Weak Sparse Partition}}	\label{alg:weakSPD}
	\DontPrintSemicolon
	\SetKwInOut{Input}{input}\SetKwInOut{Output}{output}
	\Input{Graph $G=(V,E,w)$, \SPD $\{\mathcal{X},\mathcal{Q}\}$ of depth $\rho$, parameter $\Delta>0$}
	\Output{$(8,5\rho)$-weak sparse partition}
	\BlankLine
	Let $\mathcal{A} \leftarrow V$\;
	\For {$i=1$ to $\rho$}{
		\For {$X\in \mathcal{X}_i$}{
			$C_X\leftarrow\{v\in \mathcal{A}\mid d_{G[X]}(v,Q_X)\le \frac\Delta4\}$\;
			$\mathcal{A} \leftarrow \mathcal{A}\setminus C_X$\;
		}	
	}
	\For {$i=1$ to $\rho$}{
		\For {$X\in \mathcal{X}_i$}{
			$N_X=\{t_1,t_2,\dots\}$ is a $\frac\Delta4$-net of $Q_X\in\mathcal{Q}_i$ w.r.t. $d_{G[Q_X]}$\;			
			\For {$j\ge 1$}{
				$C_{t_j}\leftarrow\{u\in C_X\mid d_{G[X]}(t_i,u)\le \frac\Delta2\}\setminus\cup_{l<j}C_{t_l}$\;
			}
		}	
	}
	\Return $\{C_{t_j}\}_{i\in[\rho],X\in\mathcal{X}_i,t_j\in N_X}$\;
\end{algorithm}

\paragraph{First phase.} 
This is an iterative process. The set of active vertices will be denoted by $\mathcal{A}$, initially $\mathcal{A}=V$.
In level $i$, for every cluster $X\in \mathcal{X}_i$ recall that $Q_X\in \mathcal{Q}_i$ is a shortest path w.r.t. $G[X]$. 
Set  $C_X=\{v\in \mathcal{A}\mid d_{G[X]}(v,Q_X)\le \frac\Delta4\}$, that is the set of active vertices in $X$ at distance at most $\frac\Delta4$ from $Q_X$ w.r.t. $d_{G[X]}$. All the vertices in $C_X$ cease to be active ($\mathcal{A}\leftarrow\mathcal{A}\setminus C_X$).
By the end of the algorithm all the vertices became inactive, thus we constructed a partition $\{C_X\}_{i\in[k\underline{}],X\in\mathcal{X}_i}$. 
Note that $Q_X$ is not necessarily contained in $C_X$, in particular $C_X$ might not be connected.
Nevertheless, the cluster $C_X$ is contained in $B_{G[X]}(Q_X,\frac\Delta4)$.

\paragraph{Second phase.} 
Next, each cluster $C_X$ is partition into balls of bounded (weak) diameter. Let $N_X$ be a $\frac\Delta4$ net of $Q_X$ w.r.t. $G[Q_X]$. Note that $N_X$ might not be contained in $C_X$. Order the vertices in $N_X=\{t_1,t_2,\dots\}$ arbitrarily. For $t_j\in N_X$, set $C_{t_j}=\{u\in C_X\mid d_{G[X]}(t_j,u)\le \frac\Delta2\}\setminus\cup_{l<j}C_{t_l}$.
Our final partition is simply $\{C_{t_j}\}_{i\in[\rho],X\in\mathcal{X}_i,t_j\in N_X}$. 
Note that each vertex $u\in X$ joins some cluster. Indeed, as $u\in C_X$, there is $w\in Q_X$ s.t. $d_{G[X]}(u,w)\le\frac\Delta4$. Furthermore, there is a net point $t_j$ at distance at most $\frac\Delta4$ from $w$, and thus at most $\frac\Delta2$ from $u$. Hence unless $u$ already joined a cluster before step $j$, $u$ will join $C_{t_j}$. 

It is straightforward that the weak diameter of our partition is bounded by $\Delta$, as every cluster $C_{t_j}$ contained in $B_{G[X]}(t_j,\frac\Delta2)\subseteq B_G(t_j,\frac\Delta2)$.
Fix $r=\frac\Delta8$ and an arbitrary vertex $v$. Denote $B=B_G(v,r)$. 
Let $i_B$ be the first level of the \SPD where some vertex $t_B$ from $B$ is deleted. Denote by $X^1\in \mathcal{X}_1,X^2\in \mathcal{X}_2,\dots, X^{i_B}\in \mathcal{X}_{i_B}$ the clusters containing $B$ in each of the levels before $i_B$. For every $u\in B$ it holds that $d_{G[X^{i_B}]}(u,Q_{X^{i_B}})\le d_{G[B]}(u,t_B)\le 2r=\frac\Delta4$. Thus all the active vertices in $B$ join the cluster $C_{X^{i_B}}$ (and became inactive).
We conclude that $B\subseteq \bigcup_{j\le i_B}C_{X^{j}}$.

Fix $j\in[i_B]$ and set $B_j=B\cap C_{X^j}$. The vertices of $B_j$ can join the the cluster of centers at distance at most $\frac\Delta2+r=\frac58\Delta$ from $v$. By \Cref{lem:SPDpathNet}, $B_j$ is partitioned to at most $\nicefrac{2\cdot \frac58\Delta}{\frac\Delta4}=5$ clusters.
As $i_B\le \rho$, the vertices of $B$ are partitioned to at most $5\rho$ clusters.
\QED

\section{Chordal Graphs}
This section is devoted to Chordal graphs. Brandst{\"{a}}dt Chepoi and Dragan \cite{BCD99} proved that Chordal graphs embed into a single tree with constant distortion. Specifically given a Chordal graph $G=(V,E)$. There is a weighted tree $T=(V,E,w)$, such that for every $v,u\in V$, $d_G(u,v)\le d_T(u,v)\le 6\cdot d_G(u,v)$.
Given such a tree $T$, every $(\sigma,\tau,\Delta)$-weak sparse partition $\mathcal{P}$ of $T$ is also $(6\sigma,\tau,\Delta)$-weak sparse partition of $G$. Using \Cref{thm:treeWeak} we conclude,
\begin{corollary}\label{cor:chordalWeak}
	Every Chordal Graph admits a $\left(24,3\right)$-weak sparse partition scheme.
\end{corollary}

As the family of Chordal graphs include all trees, by \Cref{thm:treeLBStrong} there is no $(O(1),O(1))$-strong sparse partition scheme for Chordal graphs.
The only question left is regarding scattering partitions. One can try the same approach above of embedding Chordal graphs into trees. Unfortunately, Chordal graph embed only into non-subgraph trees. Specifically, Prisner \cite{Pri97} showed that Chordal graphs do not embed into a spanning trees with any constant distortion. Therefore the scattering partition for trees do not imply scattering partition for Chordal graphs.
Nevertheless, we are able to construct scattering partitions for Chordal graphs. Interestingly, the parameters we obtain equal to the parameters of scattering partition for trees.
\begin{theorem}\label{thm:chordalScat}
	Every Chordal graph $G=(V,E)$ is  $\left(2,3\right)$-scatterable.
\end{theorem}
As every induced subgraph of a Chordal graph is also Chordal, using \Cref{thm:Scattering_Implies_SPR} we conclude,\footnote{Note that Chordal graphs are unweighted, while the minor constituting the solution to the \SPR problem will necessarily be weighed.}
\begin{corollary}\label{cor:chordalSPR}
	Given a Chordal graph $G=(V,E)$ with a set $K$ of terminals, there is an efficient algorithm that returns a solution to the \SPR problem with distortion $O(1)$.
\end{corollary}
\begin{proof}[Proof of \Cref{thm:chordalScat}.]
	Let $\Delta\in\N$ be some integer parameter. We can assume that $\Delta\ge 3$, as otherwise the trivial partition where each vertex is in a separate cluster fulfill the requirements.
	Set $r=\frac\Delta2$. We begin by describing the algorithm creating the partition $\mathcal{P}$. Each cluster $C\in \mathcal{P}$ will have a center vertex $\pi(C)\in C$. Each vertex $v\in V$ will admit a label $\delta_v\in[0,r]$. We denote by $\pi(v)$ the center of the cluster containing $v$.
	
	Let $\mathcal{T}$ be a tree decomposition for the chordal graph $G$. We can assume that $\mathcal{T}$ is rooted in a bag $B_{\rt}$ containing a single vertex $\rt$, and that every bag $B$ contains exactly one new vertex not belonging to its parent (w.r.t. the order defined by the root bag). We will denote the first bag introducing a vertex $v$ by $B_v$. Note that this defines a bijection between the vertices to the bags. Moreover, this induces a partial order on the vertices $V$, where $v\preceq u$ if $B_v$ is a decedent of $B_u$.
	
	The partition $\mathcal{P}$ of $G$ is defined inductively w.r.t. this partial order. Initially we create a cluster $C_{\rt}$ and set $\pi(C_{\rt})=\pi(\rt)=\rt$, $\delta_{\rt}=0$.
	Consider a vertex $v\in V$, which is introduced at bag $B_v$. By induction all the other vertices in $B_v$ are already labeled and clustered. Let $u$ be the vertex with minimal label among $B\setminus\{v\}$ (breaking ties arbitrarily).
	\begin{OneLiners}
		\item If $\delta_u<r$, set $\pi(v)=\pi(u)$ and $\delta_v=d_G(v,\pi(v))$.
		\item Else ($\delta_u=0$), create new cluster $C_v$ centered at $v$. Set $\pi(C_v)=\pi(v)=v$ and $\delta_v=0$.
	\end{OneLiners}  
	This finishes the description of the algorithm. 
	As each vertex joins a cluster where it has a neighbor (or starts a new cluster), the connectivity follows. 
	For every vertex $v$ which is not a cluster center, $v$ joined a cluster centered in $\pi(u)$ for some neighbor $u$ of $v$ with label $\delta_u<r$.
	It holds that $d_G(v,\pi(v))\le d_G(u,\pi(u))+1\le r$. In other words, the distance from every cluster center to all other vertices in the cluster is bounded by $r$. It follows that the created partition has weak diameter $2r\le\Delta$.
	
	Next, we claim by induction on the tree decomposition, that in every bag $B$, all the vertices with labels strictly smaller than $ r$, $\{v\in B\mid \delta_v< r\}$ belong to the same cluster. Indeed consider a bag $B_v$ introducing $v$, and let $u$ be a vertex minimizing $\delta_u$ among $B'=B\setminus\{v\}$. If $\delta_u\ge  r$, there is nothing to prove. Otherwise, $\pi(v)=\pi(u)$ and by the induction hypothesis  all the vertices with labels strictly smaller than $ r$ belong to the cluster centered at $\pi(u)$.

	Finally we are ready to prove that the partition is indeed scattering. Consider a path $\mathcal{I}=v_0,v_1,\dots,v_q$ where $q\le r$. We will abuse notation and denote $\mathcal{B}_{v_i}$ by $\mathcal{B}_i$ for $i\in[0,q]$.
	Suppose that $v_i$ is the maximal vertex of $\mathcal{I}$ w.r.t. the partial order induced by $\mathcal{T}$. In particular $v_i$ is the first vertex from $\mathcal{I}$ to be clustered in our algorithm.
	Let $j>i$ be the minimal index such that $\pi(v_j)\ne \pi(v_i)$ (if exist).
	As $\mathcal{I}$ is a shortest path, for every $a\in [0,q-2]$, there is no edges between $v_a$ to $v_{a+2}$. Therefore, the path in $\mathcal{T}$ between $\mathcal{B}_i$ to $\mathcal{B}_j$ must go through $\mathcal{B}_{i+1}\dots,\mathcal{B}_{j-1}$ in that order. In particular, $\mathcal{B}_j$ necessarily contains $v_{j-1}$. 
	As $\pi(v_j)\ne \pi(v_i)$ it follows that $\delta_{v_{j-1}}= r$. 
	It must be that $\mathcal{B}_{j-1}$ contains a vertex $\tilde{u}$ such that $\pi(\tilde{u})=\pi(v_i)=\pi(v_{j-1})$ and $\delta_{\tilde{u}}= r-1$. 
	
	We argue that $\delta_{v_j}\le2$. If $v_j$ is a cluster center, then $\delta_{v_j}=0$.
	Otherwise, there is a vertex $u\in\mathcal{B}_j$ such that $v_j$ joins the cluster centered in $\pi(u)$.
	As $\pi(\pi(u))=\pi(u)\ne \pi(\tilde{u})$, and $\delta_{\pi(u)}=0$, necessarily $\pi(u)\notin \mathcal{B}_{j-1}$ (as $\tilde{u}\in\mathcal{B}_{j-1}$ and $\delta_{\tilde{u}}< r$). In particular, the bag $\mathcal{B}_{\pi(u)}$ introducing $\pi(u)$ lies in $\mathcal{T}$ on the path between  $\mathcal{B}_{j-1}$ and  $\mathcal{B}_{j}$. As $v_{j-1}\in  \mathcal{B}_{j-1}\cap \mathcal{B}_{j}$, it must hold that $v_{j-1}\in  \mathcal{B}_{\pi(u)}$. As $G$ is a Chordal graph, it follows that $\{v_{j-1},\pi(u)\}\in E$. We conclude that $\delta_{v_j}=d_G(v_j,\pi(u))\le d_G(v_j,v_{j-1})+d_G(v_{j-1},\pi(u))=2$. See the figure bellow for illustration.
	
	\begin{center}
		\includegraphics[scale=.9]{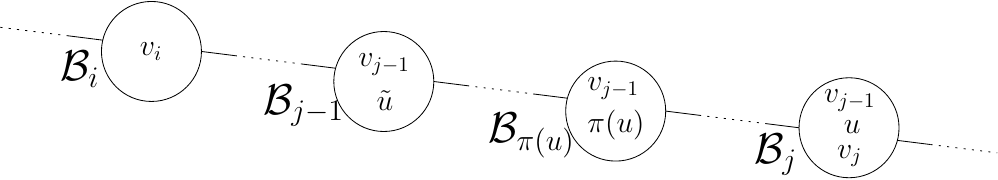}
	\end{center}
	
	Observe, that for every pair of neighboring vertices $\{v,u\}$ where $v$ was introduced before $u$ it holds that $\delta_u\le \delta_v+1$. By induction, for every $j'>j$, $\delta_(v_{j'})\le \delta_(v_{j})+|j'-j|\le 2+|j'-j|$. In particular, it follows that all the vertices $v_j,v_{j+1},\dots,v_{\min\{q,j-2+ r\}}$ belong to the cluster centered at $\pi(v_j)$ (as each cluster $\mathcal{B}_{j'}$ contains some vertex belonging to $C_{\pi(v_j)}$ with label strictly less than $ r$). 
	
	Finally, using case analysis we argue that all the vertices in $\mathcal{I}$ belong to at most $3$ clusters.
	If $1<i<q$, then all the vertices $v_i,\dots,v_q$ belong to at most two clusters ($C_{\pi(v_i)},C_{\pi(v_j)}$). By a symmetric argument, all the vertices $v_1,\dots,v_i$ belong to at most two clusters. Thus the theorem follows.
	For the case $i=1$, we can argue that all the vertices $v_1,\dots,v_{q-1}$ belong to at most two clusters, and hence again the theorem follows. The case $i=q$ is symmetric.
\end{proof}

\section{Cactus Graphs}
In this section we construct scattering partitions for Cactus graphs. Note that weak sparse partitions for this family already follow by \Cref{cor:WeakMinor}, while the lower bound \Cref{thm:treeLBStrong} on strong sparse partitions holds for Cactus graphs, as they contain the family of trees.

\begin{theorem}\label{thm:cactus}
	The family of Cactus graphs is $\left(4,5\right)$-scatterable.
\end{theorem}
As every subgraph of a Cactus graph is also a Cactus graph, using \Cref{thm:Scattering_Implies_SPR} we conclude,
\begin{corollary}\label{cor:cactusSPR}
	Given a Cactus graph $G=(V,E,w)$ with a set $K$ of terminals, there is an efficient algorithm that returns a solution to the \SPR problem with distortion $O(1)$.
\end{corollary}

As a preliminary to the proof of \Cref{thm:cactus}, we begin with a characterization of Cactus graphs:
\begin{claim}\label{clm:CactusDecomp}
	Each weighted Cactus graph $G=(V,E,w)$ can be composed as sequence of Cactus graphs $G_0,G_1,\dots,G_s =G$ (here $G_i=(V_i,E_i)$ ) such that:
	\begin{enumerate}
		\item $G_0$ is a single vertex.
		\item The graph $G_i$ is obtained by attaching a path $P_i=(v_0,v_1,\dots,v_q)$ (disjoint from $V_{i-1}$) to a single vertex $u_i$ of $G_{i-1}$ to either one or both endpoints of $P_i$. That is $V_i=V_{i-1}\cup\{v_0,v_1,\dots,v_q\}$ and $E_i$ consist of the edges in $E_{i-1}$, the edges along the path $P_i$, and at least one of the edges $\{u_i,v_0\},\{u_i,v_q\}$.
	\end{enumerate}
	Further, it holds that the shortest path metrics of $G_i$ and $G$ agree on $V_i$. In other words $\forall u,v\in V_i,~d_{G_i}(u,v)=d_G(u,v)$.
\end{claim} 
\begin{proof}
	The proof is by induction on the number of vertices $n$. The base case is when $G$ is either a single vertex, a path graph, or a cycle graph. All these cases are trivial. Suppose that the claim hols for every Cactus graph with strictly less than $n$ vertices, and consider a cactus graph $G=(V,E)$ with $n$ vertices, which is not covered by tha base of the induction. We define an auxilery graph as follows: Let $\cC$ be the collection of cycles in $G$. Let $E_1\subseteq E$ be subset of $G$ edges that do not belong to a cycle. Our auxilery graph $\cal{T}$ will have $\mathcal{V}=\mathcal{C}\cup E_1$ as its vertex set. For every $X,Y\in \cal{V}$ that share a vertex, we will add an edge to $\cal{G}$. Clearly, as $G$ is a Cactus graph, $\cal{T}$ is a tree. Further, every $X,Y\in \cal{V}$ share at most a single vertex (as every edge belongs to at most a single cycle).
	
	Let $X\in \cal{V}$ be a leaf in the auxilery tree $\cal{T}$ (note that $\cal{T}$ is not a singleton, as we are not in the base case). Let $x\in X$ be the unique vertex that belongs to an additional node in $\cal{T}$ other than $X$. Let $G'$ be the graph obtained by deleting all the vertices in $X$ other than $x$. Note that $G'$ is a cactus graph with strictly less than $n$ vertices, and hence by the induction hypothesis it has a decomposition $G_0,G_1,\dots,G_s =G'$ as in the claim. 
	There are two cases. Suppose first that $X$ is an edge $\{x,y\}$. Then we set $P_{s+1}=(y)$ to be a singleton vertex and attach it to $G_s$ via $x$. $G_{s+1}=G$ as required.
	The second case is that $X$ is a cycle $C=(v_0,v_1\dots,v_q)$ where $v_0=v_q=x$. Here we set $P_{s+1}=(v_1,\dots,v_{q-1})$ and attach it to $G_s$ using the two edges $\{x,v_1\},\{v_{q-1},x\}$ to obtain $G_{s+1}$. Clearly $G_{s+1}=G$, and for every $y,z\in G_s$, $d_G(y,z)=d_{G_s}(y,z)$. The claim now follows.
\end{proof}

We are now ready to proceed to the proof of \Cref{thm:cactus}.	
\begin{proof}[Proof of \Cref{thm:cactus}.]
	Let $\Delta>0$ be some parameter. Set $r=\frac\Delta4$. Each cluster we create $C_T$ will have a center $T$.
	The center $T$ might be either a singleton, or a set of size $2$.
	The clustering is defined inductively using \Cref{clm:CactusDecomp}. 
	First, $G_0=\{v\}$, create a cluster $C_v$ with $v$ as a center. 
	In each step we will either extend previously created clusters or create a new clusters. 
	Consider the $i$'th step in the composition procedure. Suppose that all the vertices in $G_{i-1}$ are already clustered. A new path $P_i=\{v_0,v_1,\dots,v_m\}$ is attached to $G_{i-1}$ at vertex $u\in G_{i-1}$ which belong to a cluster $C_T$. 
	We consider two cases:
	\begin{itemize}
		\item $P_i$ is attached to $u$ via a single edge $\{v_0,u\}$.
		The prefix of the vertices $v_0,\dots,v_{j-1}$ which are at distance at most $r$ (w.r.t. $d_{G_i}=d_{G}$) from $T$ join the cluster $C_T$. Let $v_j$ be the first vertex at distance greater then $r$ from $T$ (if it exist).
		$v_{j}$ is defined as a center of a new cluster $C_{v_j}$. The prefix of the remaining vertices $v_{j},\dots,v_{q-1}$ at distance at most $r$ from $v_j$ join $C_{v_j}$. The vertex $v_q$ (if it exist) is defined as the center of a new cluster $C_{v_q}$. We process  in this manner until all the vertices of $P_i$ are clustered.
		
		\item  $P_i$ is a attached to $\{u\}$ via two edges $\{v_0,u\},\{v_m,u\}$. The clustering of $P_i$ has three phases:
		\begin{itemize}
			\item  The prefix $v_0,v_1,\dots,v_{j-1}$ (resp. the suffix $v_{j'+1},v_{j'+2},\dots,v_{m}$)  of the vertices which are at distance at most $r$ from $T$ join the cluster $C_T$. If not all the vertices of $P_i$ are clustered, denote $t_a=v_j$ and $t_b=v_{j'}$ (it is possible that $t_a=t_b$) and proceed to the next phase.
			\item If $d_{G[P_i]}(t_a,t_b)\le 2r$ go to the next phase. Otherwise, create two new clusters $C_{t_a},C_{t_b}$ centered at $t_a,t_b$ respectively. The prefix (resp. suffix) of the remaining vertices $t_{a}=v_{p'},\dots,v_{p-1}$ (resp. $v_{q+1},\dots,v_{q'}=t_b$) at distance at most $r$ from $t_a$ (resp. $t_b$) joins the cluster $C_{t_a}$ (resp. $C_{t_b}$).
			If all the vertices were clustered we are done. Otherwise, denote $t_a=v_p$ and $t_b=v_{q}$. Repeat phase two.
			\item Set $\tilde{T}=\{t_a,t_b\}$. Create a new cluster $C_{\tilde{T}}$ with $\tilde{T}$ as a center. All the yet unclustered vertices in $P_i$ join $C_{\tilde{T}}$.
		\end{itemize}
	\end{itemize}
	
	\begin{claim}
		The partition has strong diameter $4r=\Delta$.
	\end{claim}
	\begin{proof}
		Consider a cluster $C_T$.
		It is straightforward by induction that for every vertex $v\in C_T$, $d_{G[C_T]}(v,T)=d_{G}(v,T)\le r$. Thus if $T$ is a singleton, then $C_T$ has strong diameter $2r$.
		Otherwise ($|T|=2$), consider a pair of vertices $v,u\in C_T$. There are centers $t_v,t_u\in T$ such that $d_{G[C_T]}(t_v,v)\le r$, $d_{G[C_T]}(u,t_u)\le r$, and $d_{G[C_T]}(t_v,t_u)\le 2r$. The claim follows by triangle inequality.
	\end{proof}
	
	For two indices $i,i'$ denote $P_i^{i'}=\cup_{q=i}^{i'}P_q$, and $P_{>i}=\cup_{q>i}P_q$.
	For a cluster $C_T$ created during the $i$'th phase, denote by $\tilde{C}_T=C_T\cap P_i$ the set of vertices belonging to $C_T$ by the end of the $i$'th phase.
	Using a straightforward induction, we have the following observation.
	\begin{observation}\label{obs:CactusContinuity}
		\sloppy Consider a center $T$ of the cluster $C_T$ created during the $i$'th phase.
		Then $B_{G[\tilde{C}_T\cup P_{>i}]}(T,r)\subseteq C_T$. In words, every vertex $v\in  P_{>i}$ for which there is a path towards $T$ of length at most $r$ containing vertices from $P_{>i}$ and $\tilde{C}_T$ only, will join $C_T$.
	\end{observation}
	Consider a shortest path $Q=\{z_0,z_1,\dots,z_\alpha\}$ of length at most $r=\frac\Delta4$.
	Suppose that in the composition procedure, $P_i$ is the path with minimal index which intersects $Q$.
	The intersection between $Q$ and $P_i$ must be an interval  $Q\cap P_i=\{z_{j'+1},\dots,z_{j-1}\}$.
	This is, as the shortest path between vertices in $G_i$ do not contain edges from $G[P_{> i}]$.
	As $w(Q)\le r$, by case analysis, $Q\cap P_i$ can be divided to at most $3$ consecutive clusters. 
	
	Consider the suffix $\{z_{j},z_{j+1},\dots,z_\alpha\}$. Denote by $C_{z_{j-1}}$ the cluster the vertex $z_{j-1}$ joined to. We argue that other than to $C_{z_{j-1}}$, the suffix vertices can join to at most one additional cluster. 
	By the composition procedure, the vertices $\{z_{j},z_{j+1},\dots,z_\alpha\}$ must be added to $G$ in paths $P_{i_1},P_{i_2}\cdots,P_{i_q}$ where $i_1<i_2<\dots<i_q$ and $(Q\cap P_{i_1})\circ (Q\cap P_{i_{2}}) \circ \dots\circ (Q\cap P_{i_q})=\{z_{j},z_{j+1},\dots,z_\alpha\}$.
	Let $i_p$ be the minimal index such that not all the vertices of $P_{i_p}=\{x_0,x_1\dots,x_\beta\}$ join $C_{z_{j+1}}$. If there is no such index, we are done.
	Note that only $x_0$ and $x_\beta$ might have edges to formerly introduced vertices. Hence $x_0$ or $x_\beta$ belong to $Q\cap P_{i_p}$.	
	W.l.o.g. $Q\cap P_{i_p}=\{x_0,x_1\dots,x_\gamma\}$.
	Let $q$ be the minimal index such that $x_q$ not joining the cluster of $z_{j+1}$. 
	Following the construction algorithm, $x_q$ must belong to the center of a cluster $C_T$. As $Q$ is of length at most $r$, all the vertices $\{x_{q+1}\dots,x_\gamma\}\subseteq Q$ join $C_{T}$. By \Cref{obs:CactusContinuity}, all the vertices in $Q\cap P_{i_{p}+1},Q\cap P_{i_{p}+2},\dots,Q\cap P_{i_q}$ also join $C_{T}$.
	
	By symmetric arguments, the prefix vertices $z_{1},z_{2},\dots,z_{j'}$ join to at most a single cluster other than the cluster of $z_{j'+1}$. The theorem follows.
\end{proof}

\section{Discussion and Open Problems}\label{sec:open}
In this paper we defined scattering partitions, and showed how to apply them in order to construct solutions to the \SPR problem. We proved an equivalence between sparse partitions and sparse covers. Finally, we constructed many sparse and scattering partitions for different graph families (and lower bounds), implying new results for the \SPR, \UST, and \UTSP problems.
An additional contribution of this paper is a considerable list of (all but question (5)) new intriguing open questions and conjectures.
\begin{enumerate}
	\item \textbf{Planar graphs:} The \SPR problem is most  fascinating and relevant for graph families which are closed under taking a minor.
	Note that already for planar graphs (or even treewidth $2$ graphs), the best upper bound for the \SPR problem is $O(\log k)$ (same as general graphs), while the only lower bound is $8$. 
	The most important open question coming out of this paper is the following conjecture:
	\begin{conjecture}\label{conj:minor}
		Every graph family excluding a fixed minor is $(O(1),O(1))$-scatterable.
	\end{conjecture}
	Note that proving this conjecture for a family $\mathcal{F}$, will imply a solution to the \SPR problem with constant distortion. 
	
	\item \textbf{Scattering Partitions for General Graphs:} While we provide almost tight upper and lower bounds for sparse partitions,  for scattering partitions, the story is different.
	\begin{conjecture}\label{con:GeneralScattering}
		Consider an $n$ vertex weighted graph $G$ such that between every pair of vertices there is a unique shortet path. Then $G$ is $\left(1,O(\log n)\right)$-scatterable. Furthermore, this is tight.
	\end{conjecture}
	\Cref{thm:generalLBsuper} provides some evidence that \Cref{con:GeneralScattering} cannot be pushed further. However, any nontrivial lower bound will be interesting.
	Furthermore, every lower bound larger than $8$ for the general \SPR problem will be intriguing.
	
	\item \textbf{Doubling graphs:} While we constructed strong sparse partition for doubling graphs (which imply scattering), it has no implication for the \SPR problem. This is due to the fact that \Cref{thm:Scattering_Implies_SPR} required scattering partition for every induced subgraph. As induced subgraphs of a doubling graph might have unbounded doubling dimension, the proof fails to follow through.
	We leave the required readjustments to future work.

	\item \textbf{Sparse Covers:} We classify various graph families according to the type of partitions/covers they admit, as exhibited in \Cref{fig:Venn_covers}. We currently lack any example of a graph family that admits weak sparse covers but does not admit strong sparse covers. It will be interesting to find such an example, or even more so to prove that every graph that admits weak sparse cover, also has strong sparse cover with (somewhat) similar parameters.
	
	\item \textbf{Treewidth graphs:} The parameters in some of our partitions perhaps might be improved. The most promising example in this context is treewidth $\rho$ graphs. As such graphs exclude $K_{\rho+2}$ as a minor, by \Cref{thm:KPR} they admit $\left(O(\rho^2),2^{\rho+2}\right)$-weak sparse partition scheme. However, they might admit sparse partitions with parameter polynomial, or even logarithmic in $\rho$..
	Recently Hershkowitz and Li \cite{HL22} constructed an $(O(1),O(1))$-scattering partitions for series parallel graphs (alternatively treewidth $2$ graphs), answering an open question posed in the conference version of this paper \cite{Fil20}. However, already for treewidth $3$ graphs the question is wide open.
\end{enumerate}

\section*{Acknowledgments}
The author would like to thank Alexandr Andoni, Robert Krauthgamer, Jason Li, Ofer Neiman, Anastasios Sidiropoulos, and Ohad Trabelsi for helpful discussions.

{\small
	\bibliographystyle{alphaurlinit}
	\bibliography{SteinerBib}
}
\appendix

\end{document}